\documentclass{CSML}
\pdfoutput=1

\usepackage{lastpage}
\newcommand{\dOi}{13(3:15)2017}
\lmcsheading{}{1--\pageref{LastPage}}{}{}%
{Jan.~24,~2017}{Aug.~25, 2017}{}

\newcommand{\lmcscenterline}[1]{\begin{center} #1 \end{center}}
\usepackage{hyperref}
\hypersetup{hidelinks}

\usepackage{amsfonts}
\usepackage{amsmath}
\usepackage{amssymb}
\usepackage{amsbsy}
\usepackage{bm}
\usepackage{pifont}

\usepackage{derivation}

\DeclareMathAlphabet{\mathpzc}{OT1}{pzc}{m}{it}

\renewcommand{\vec}[1]{\bm{#1}}
\newcommand{\set}[1]{\{#1\}}

\newcommand{\Comment}[1]{ }

\newcommand{\ByDef}{\triangleq}

\newcommand{\Gammainf}{\Gamma^{\mathsf{x}}}
\newcommand{\fofder}[1]{\mathsf{f}(#1)}
\newcommand{\fofderaux}[1]{\mathsf{f\mbox{-}aux}(#1)}
\newcommand{\dertoderinf}[1]{\mathsf{fder}(#1)}
\newcommand{\dertoderinfaux}[1]{\mathsf{fder\mbox{-}aux}(#1)}
\newcommand{\nextmove}[1]{\mathsf{nextm}(#1)}
\newcommand{\finplay}[1]{\mbox{{\sc fp}}(\bm{#1})}
\newcommand{\movlabtree}[1]{\mathsf{mlt}(#1)}
\newcommand{\game}{\mathcal{G}}
\newcommand{\es}{\mathcal{E}}
\newcommand{\Puniverse}{\mathfrak{P}_\mathfrak{U}}
\newcommand{\finParts}[1]{{\mathcal{P}_{\!\mbox{\scriptsize fin}}(#1)}}
\newcommand{\euniverse}{\mathbf{E}}
\newcommand{\upto}[2]{#2_{\!/ #1}}
\newcommand{\setof}[1]{\widehat{#1}}
\newcommand{\partecipants}{\mathfrak{P}}
\newcommand{\wf}{\mathcal{W}\!}
\newcommand{\rt}[1]{\mathsf{rt}(#1)}
\newcommand{\rtsigma}[2]{\mathsf{rt}_{#1}(#2)}
\newcommand{\rts}[1]{\mathsf{rts}(#1)}
\newcommand{\srt}[1]{\mathsf{srt}(#1)}
\newcommand{\srtaux}[2]{\mathsf{srt\mbox{-}aux}(#1,#2)}
\newcommand{\rtaux}[3]{\mathsf{rt\mbox{-}aux}_{#1}(#2,#3)}
\newcommand{\orch}[1]{\mathsf{orch}(#1)}

\newcommand{\regtree}[1]{\mathit{tree}(#1)}

\def \External {\Sigma}
\newcommand{\bigExternal}{\Sigma}
\def \Internal {\oplus}
\newcommand{\bigInternal}{\bigoplus}

\newcommand{\modelsACRel}{\models^{\!\mbox{\tiny $\ACRel$}}}
\newcommand{\modelsSubcontr}{\models^{\mbox{\tiny \!$\subcontr$}}}

\newcommand{\Rel}{\mathpzc R}

\newcommand{\FunH}{\mathcal{H}}
\newcommand{\FunK}{{\mathcal K}}

\newcommand{\Subst}[2]{\{#1/#2\}}

\newcommand{\DScrachcard}{\Dual{\tt{scratchcard}}}
\newcommand{\Cheque}{\tt{cheque}}
\newcommand{\Bag}{\tt{bag}}

\newcommand{\DBag}{\Dual{\Bag}}
\newcommand{\Belt}{\tt{belt}}
\newcommand{\DBelt}{\Dual{\Belt}}
\newcommand{\Price}{\tt{price}}
\newcommand{\DPrice}{\Dual{\Price}}
\newcommand{\Card}{\tt{card}}
\newcommand{\DCard}{\Dual{\Card}}
\newcommand{\Cash}{\tt{cash}}
\newcommand{\DCash}{\Dual{\Cash}}

\newcommand{\sBag}{\mbox{\tt \tiny bag}}
\newcommand{\sDBag}{\Dual{\sBag}}

\newcommand{\sPrice}{\mbox{\tt \tiny price}}
\newcommand{\sDPrice}{\Dual{\sPrice}}
\newcommand{\sCard}{\mbox{\tt \tiny card}}
\newcommand{\sDCard}{\Dual{\sCard}}
\newcommand{\sCash}{\mbox{\tt \tiny cash}}
\newcommand{\sDCash}{\Dual{\sCash}}

\newcommand{\Byr}{\mathsf{Buyer}}
\newcommand{\Slr}{\mathsf{Seller}}

\newcommand{\vdBS}{\vdash_{\Byr\!\pp\!\Slr}}

\newcommand{\ored}[1]{\mathstackrel{#1}{\longrightarrow}}      
 
\newcommand{\rlbk}{\sf rbk}

\newcommand{\mathstackrel}[2]{\mathrel{\stackrel{#1}{#2}}}

\newcommand{\Lts}[1]{\mathstackrel{#1}{\Longrightarrow}}

\newcommand{\comply}{\dashv}
\newcommand{\complyR}{\comply^{\mbox{\tiny {\sf r\hspace{-0.3pt}b\!k}}}}

\newcommand{\ncomplyR}{\not\complyR}
\newcommand{\complyTB}{\comply_{\mbox{\tiny {\sf t\!\,b}}}}
\newcommand{\complyTBO}{\complyO_{\mbox{\tiny {\sf t\!b}}}}

\newcommand{\complyF}{\,\mbox{\scriptsize $\backsim\hspace{-3.5pt}\mathbf{\mid}$}\,}

\newcommand{\complyP}{\dashv\!\!\!\dashv}

\newcommand{\complyOco}{\complyO_{co}}
\newcommand{\complyOcok}[1]{\complyO_{#1}}

\newcommand{\complyO}{\comply^{\mbox{\tiny {\sf Orch}}}}
\newcommand{\complyOF}{\complyF^{\!\mbox{\tiny {\sf Orch}}}}

\newcommand{\ACRel}{\,\mathbb{A\!C}\,}
\newcommand{\ACRelco}{\mathbb{A\!C}_{\!\mathit{co}}}

\newcommand{\ACRelk}[1]{\,{\ACRel_{\!#1}}\,}
\newcommand{\subcontrk}[1]{\,{\subcontr_{\!#1}}\,}

\newcommand{\Actdot}{.}
\newcommand{\procdot}{.}

\newcommand{\Funct}{\mathbf{F}}
\newcommand{\auxFunct}[2]{\mathcal{F}(#1,#2)}

\newcommand{\orchAct}[2]{\langle\mbox{\small $#1$},\mbox{\small $#2$}\rangle}

\newcommand{\Dual}[1]{\overline{#1}}
\newcommand{\QDual}[1]{\widehat{#1}}

\newcommand{\subcontr}{\preccurlyeq}

\newcommand{\Sem}[1]{[\hspace{-0.6mm}[ #1 ]\hspace{-0.6mm}]}

\newcommand{\Procinfty}{\PProcinfty\!\!\!\text{\bf -aux}}
\newcommand{\PProcinfty}{{\text{\bf R}^{\!\infty}}}
\newcommand{\Clause}[1]{{\bf Clause} ~#1}
\newcommand{\Proc}{{\text{\bf R-aux}}}
\newcommand{\PProc}{{\text{\bf R}}}

\newcommand{\der}{\vartriangleright}
\newcommand{\derinfty}{\der^{\!\!\mbox{\tiny $\infty$}}}

\newcommand{\derOrch}{\der_{\!\mbox{\tiny \sf o}}}
\newcommand{\derOrchinfty}{\derOrch^{\!\!\mbox{\tiny $\infty$}}}
\newcommand{\derinfOrch}{\der_{\!\mbox{\tiny \sf o}}^{\!\mbox{\tiny \sf x}}}

\newcommand{\Deriv}{\mathcal {D}}
\newcommand{\Der}{\Deriv}

\newcommand{\Prove}{\textbf{Prove}}
\newcommand{\OrchToDerAux}{\textbf{O2D-aux}}
\newcommand{\OrchToDer}{\textbf{O2D}}
\newcommand{\Synth}{\textbf{Synth}}
\newcommand{\IF}{\textbf{if}}
\newcommand{\THEN}{\textbf{then}}
\newcommand{\WHERE}{\textbf{where}}
\newcommand{\ELSE}{\textbf{else}}
\newcommand{\CASE}{\textbf{case}}
\newcommand{\OF}{\textbf{of}}
\newcommand{\AND}{\textbf{and}}

\newcommand{\LET}{\textbf{let}}

\newcommand{\FAIL}{\textbf{fail}}
\newcommand{\IN}{\textbf{in}}

\newcommand{\FA}{\textbf{for all}}

\newcommand{\TcomplHyp}{\mbox{\scriptsize \sc Hyp}}
\newcommand{\TcomplAx}{\mbox{\scriptsize \sc Ax}}

\newcommand{\CkptcomplHyp}{\mbox{\sc Hyp}}
\newcommand{\CkptcomplAx}{\mbox{\sc Ax}}

\newcommand{\Inv}{\mathfrak{Inv}}

\newcommand{\Buf}{\mbox{\small $\mathbb{B}$}}
\newcommand{\Bufsc}[2]{\Buf_{#1}^{\mbox{\tiny $#2$}} \kern-2pt \raise -1.5pt\hbox{$\rfloor$}}
\newcommand{\Bufcs}[2]{\raise -1.5pt\hbox{$\lfloor$}\kern-1pt{}_{#1}^{\mbox{\tiny $#2$}}\Buf}

\def \leftinseq #1|#2|{#1|\,#2\,|}
\def \rightinseq |#1|#2{|\,#1\,|\,#2}

\def\numberstoc|#1|_#2{|#1|_{#2}}
\def\numberctos _#1|#2|{|#1|_{#2}}

\newcommand{\Nat}{\mathbb{N}} 

\newcommand{\Set}[1]{\{\,#1\,\}}
\def\Pred[#1]{~[\,#1\,]}

\newcommand{\Iff}{\Leftrightarrow}

\renewcommand{\implies}{~\Rightarrow~}

\newcommand{\OrchStep}[1]{\stackrel{#1}{\mapsto}}

\newcommand{\playerA}{\mathsf A}
\newcommand{\playerB}{\mathsf B}
\newcommand{\playerC}{\mathsf C}

\newcommand{\Zero}{{\bf 0}}

\newcommand{\ff}{\,\mathsf{to}_{\!\pp}\!\!}
\newcommand{\fftb}{\,\mathsf{to}_{{\mid\hspace{-1.3pt}\mid\hspace{-1.3pt}\mid}}\!}

\newcommand{\buf}[1]{[#1]}

\newcommand{\tblts}[1]{\mathstackrel{\,#1\,}{~\longrightarrow\hspace{-5.5mm}\longrightarrow~}}

\newcommand{\notblts}{\,\,\not\!\!\!\!\!\tblts{}}

\newcommand{\tbltsstar}[1]{\mathstackrel{\,#1\,}{~\longrightarrow\hspace{-5.5mm}\longrightarrow^{*}~}}

\newcommand{\tbolts}[1]{\mathstackrel{\,#1\,}{~\longrightarrow\hspace{-5.5mm}\longrightarrow_{\!o}~}}

\newcommand{\lts}[1]{\mathrel{\xrightarrow{\,#1\,}{}}}
\newcommand{\ltsstar}[1]{\mathrel{\xrightarrow{\,#1\,}{\!\!^*}}}

\newcommand{\nottblts}[1]{\;\;\not\!\!\!\!\!\!\tblts{#1}}
\newcommand{\nottbolts}[1]{\;\;\not\!\!\!\!\!\!\tbolts{#1}}

\newcommand{\ltsOrch}[2]{\OrchStep{\langle #1,#2 \rangle}}
\newcommand{\ltsOrchplus}[2]{\OrchStep{\langle #1,#2 \rangle^+}}

\newcommand{\rec}{{\sf rec} \, }

\newcommand{\notLts}[1]{~\not\!\!\mathstackrel{#1}{\Longrightarrow}}

\newcommand{\Ctrs}{Contracts}
\newcommand{\ASC}{\mathsf{A\!S\!C}}

\newcommand{\Stacks}{{\sf Histories}}
\newcommand{\emptystack}{[\;]}
\newcommand{\cons}{\!:\!}
\newcommand{\ctrs}{contracts}
\newcommand{\SbehavH}{{\sf RCH}}

\newcommand{\CommRule}{\textsf{comm}}

\newcommand{\RbkRule}{\textsf{rbk}}

\newcommand{\rb}{\sf rb}

\newcommand{\Names}{\mathcal{N}}
\newcommand{\CoNames}{\overline{\Names}}

\newcommand{\stopA}{{\bf 1}}
\newcommand{\stopf}{\mathfrak{1}}
\newcommand{\Act}{\mbox{\bf Act}}

\newcommand{\OrchAct}{\mbox{\bf OrchAct}}
\newcommand{\Orch}{\mbox{\sf Orch}}

\newcommand{\Trace}[1]{\mbox{\sf T\!r}(#1)}
\newcommand{\MaxTrace}[1]{\mbox{\sf  M\!\,a\!\,x\!T\!r}(#1)}
\newcommand{\FinMaxTrace}[1]{\mbox{\sf  F\!\,i\!\,n\!\,M\!\,a\!\,x\!T\!r}(#1)}

\newcommand{\sat}[1]{\mbox{\it sat}(#1)}

\newcommand{\snd}[1]{\mathsf{snd}(#1)}
\newcommand{\tbAct}{{\bf tb}\hspace{-1pt}\Act}
\newcommand{\seq}[1]{\langle#1\rangle}

\newcommand{\Sbehav}{{\sf SC}}

\newcommand{\back}{\prec}

\newcommand{\tbASC}{\ASC^{[\,]}}
\newcommand{\sctb}{\tbASC}

\newcommand{\np}[2]{#1{\back}\,#2}

\newcommand{\subcontrF}{\ll}
\newcommand{\pf}[1]{\mathbin{\|_{#1}}}
\newcommand{\tbpf}[1]{\mathbin{\,\tbpp_{\!#1}}}

\newcommand{\pp}{\,\|\,}
\newcommand{\tbpp}{\,\mid\hspace{-4.3pt}\mid\hspace{-4.3pt}\mid\,}

\newcommand{\lab}{\mathit{l}}
\newcommand{\dersc}{\;\blacktriangleright\;}
\newcommand{\derscinfty}{\;\blacktriangleright^{\!\!\mbox{\tiny $\infty$}}\;}

\newcommand{\cmark}{\mbox{\ding{51}}}%
\newcommand{\xmark}{\mbox{\ding{55}}}%

\theoremstyle{definition}\newtheorem{notation}{Notation}

\begin{document}

\title[Retractability, games and orchestrators for session contracts]{Retractability, games and orchestrators\\ for session contracts}

\author[Barbanera and de' Liguoro]{Franco Barbanera}	
\address{Dipartimento di  Matematica e Informatica, University of Catania}	
\email{barba@dmi.unict.it}  

\author[]{Ugo de' Liguoro}	
\address{Dipartimento di Informatica, University of Torino}	
\email{ugo.deliguoro@unito.it}  
\thanks{This work was partially supported by the European Union ICT COST Action IC1405 Reversible computation - extending horizons of computing, the Project FIR 1B8C1 of the University of Catania and Project FORMS 2015 of the University of Turin.}	



\subjclass{[Theory of Computation]:Logics and meanings of programs-}
\titlecomment{The present paper is a reorganised, revised and extended version of the  COORDINATION 2016 conference paper {\em ``A game interpretation of retractable contracts''} \cite{BdL16} by the same authors}


\begin{abstract}
  \noindent Session contracts is a formalism enabling to investigate client/server interaction protocols and to interpret session types. We extend session contracts in order to represent outputs whose actual sending in an interaction depends on a third party or on a mutual agreement between the partners. Such contracts are hence adaptable, or as we say ``affectible''. In client/server systems, in general, compliance stands for the satisfaction of all client's requests by the server. We define an abstract notion of ``affectible compliance''  and show it to have a precise three-party game-theoretic interpretation. This in turn is shown to be equivalent to a compliance based on interactions that can undergo a sequence of failures and rollbacks, as well as to a compliance based on interactions which can be mediated by an orchestrator.
Besides, there is a one-to-one effective correspondence between winning strategies and orchestrators. The relation of subcontract for affectible contracts is also investigated.
\end{abstract}

\maketitle


The notion of {\em contract} ~\cite{CCLP06,LP07,CGP10} has been proposed as an abstraction to formally specify and check the behaviour of software systems, and especially of web services. In particular, in the setting of service-oriented architectures, the concept of agreement, often called {\em compliance}, is of paramount importance while searching for components and ensuring that they will properly interact with each other. The main challenge is that compliance has to meet the contrasting requirements of guaranteeing  correctness of interactions w.r.t. certain safety and liveness conditions, while remaining coarse enough to maximize the possibilities of finding compliant components in a library or services through the web.

The main conceptual tool to face the issue is that of relaxing the constraint of a perfect correspondence among contracts through {\em contract refinement}, also called sub-contract \cite{CGP10,BravettiZ09a} and  sub-behaviour \cite{BdL13} relations, that is pre-order relations such that processes conforming to more demanding contracts (which are lower in the pre-order) can be safely substituted in contexts allowing more permissive ones. Indeed contract refinement closely resembles subtyping, as it is apparent in the case of session types \cite{GH05,BdL13}, and it is related to (but doesn't  coincide with) observational pre-orders and
{\em must-testing} in process algebra \cite{LP07,BernardiH13}.

However, since the first contributions to the theory of contracts \cite{CGP10}, a rather different approach has been followed, based on the idea of filtering out certain actions that, althought unmatched on both sides of a binary interaction, can be neglected or prevented by the action of a mediating process called the {\em orchestrator} \cite{Padovani10,PadTCS,BdL15ice,BdL15}, without compromising the reaching of the goals of the participants, like the satisfaction of all client requests in a client-server architecture.

Another route for the same purpose is to change the semantics of contracts so that interacting processes can adapt  to each other by means of a rollback mechanism: these are the {\em retractable contracts} proposed in  \cite{BDLdL15}.
Although compliance can be decided in advance, interaction among processes exposing retractable contracts undergoes a sequence of failures and backtracks that might be avoided by extracting information from the compliance check.

The contribution of the present paper is to show that the use of orchestrators and retractability  are indeed two equivalent approaches to get compliance by adapting, affecting, the behaviour
of the partners of a client/server interaction,
at least in the case of  binary {\em session contracts} \cite{BdL13,BH13}.
These are contracts that limit the non-determinism by constraining both external and internal choices to a more regular form, so that they can be looked at as an interpretation of session types into the formalism of contracts \cite{BdL13,BH13}. In particular, session contracts can be seen as binary session types without value or channel passing.

The contracts we consider in this paper are session contracts with external choices of outputs
that we abstractly look at, in a sense, as the {\em affectible}, adaptable  parts of a contract.
These contracts are syntactically the same as the retractable session contracts \cite{BDLdL15}, but instead of adding rollback to the usual contract semantics, we 
formalise inside an abstract notion of compliance the fact that the actual sending of {\em affectible}, adaptable outputs can be influenced by an agreement between  the interaction partners or by some entity external to the system, in order to make the partners compliant.  
In particular, {\em affectible compliance}, i.e. 
compliance got by means of a (run-time) adaptation of the contracts' behaviours, will be first  abstractedly presented  as a coinductively defined relation.
This relation will be proved later on to be decidable and to coincide both with the retractable compliance relation of \cite{BDLdL15} involving failures and rollbacks and with the orchestrated compliance, where
the (affectible) synchronizations are influenced by elements of a particular class of orchestrators in the sense of \cite{Padovani10} and \cite{BdL15}. 

The essence of this equivalence is that the above mentioned orchestrators correspond to winning strategies in certain concurrent games that naturally model affectible contracts. In \cite{BartJLAMP} the theory of contracts has been grounded on games over event structures among multiple players; applying this framework to affectible contracts, the interaction among a client and a server can be seen as a play in a  three-party game. Player $\playerA$ moves according to the normal actions of the client; player $\playerB$ moves according to the normal actions of the server, whereas moves by player $\playerC$ correspond to affectible actions on both sides.   The server $\sigma$ is hence affectible compliant with the  client $\rho$ whenever 
$\playerC$ has a winning strategy in the game with players $\playerA$ and $\playerB$, where
player $\playerC$ wins when she/he succeeds to lead the system $\rho\|\sigma$ to a successful state (the client terminates) or the interaction proceeds indefinitely without deadlocking.

The payoff of the game theoretic interpretation is that there is a precise correspondence between winning strategies for player $\playerC$ and elements of a class of orchestrators in the sense of \cite{Padovani10}  and \cite{BdL15}. Such a correspondence is of interest on its own, since strategies are abstract entities while orchestrators are terms of a process algebra and concrete witnesses of the agreement among participants of a session.
Moreover,  we can decide whether a client/server pair is affectible compliant by means of an algorithm that synthesizes an orchestrator, if any, or reports failure.\\
We also show that there is a one-to-one correspondence between orchestrators and derivations
in the formal system axiomatizing the relation of affectible compliance.

The substitutability relation (affectible subcontract) on servers, induced by the relation of affectible compliance,  can be defined as for the usual subcontract relations. Its decidability, however has to be proved in a direct
way: the introduction of adaptability implies that decidability of the subcontract relation cannot be simply inferred from decidability of compliance. It descends instead
from correctness and termination of the proof recontruction algorithm for the formal system for the affectible subcontract relation. Moreover, we shall show how a derivation in such a formal system does correspond in an effective way to a functor on orchestrators. In particular, if $\sigma$
is proved to be a subcontract of $\sigma'$, the functor transforms any orchestrator making
$\sigma$ compliant with a client $\rho$ into an orchestrator making $\sigma'$ compliant with $\rho$.

The present paper is a reorganised, revised and extended version of  \cite{BdL16}, where most of the proofs had been omitted and where the correspondence between derivations for the compliance relation and orchestrators was not present. Besides, the relation of subcontract was not investigated in \cite{BdL16}, as well as the correspondence between derivations for the 
subcontract relation and orchestrator functors.

\subsection*{Overview of the paper.}
In Section \ref{sec:affectcompliance} we define affectible session contracts and the abstract notion of compliance on them. In Section \ref{sec:gameinterpr} we recall the notion of multi-player game from \cite{BartJLAMP} based on event structures. We then show how 
it is possible to interpret a client/server system $\rho\!\pp\!\sigma$ with a three-players game $\game_{\rho\!\pp\!\sigma}$ by means of a turn-based operational semantics.  
Sections \ref{sec:retractablecompliance} will be devoted to the formalization of the notion of interactions with rollbacks and the related notion of retractable compliance. Orchestrators and orchestrated interactions for affectible contracts will be defined in Section \ref{sec:orchcompl},
together with the notion of orchestrated compliance. In \ref{sec:affectible-retractable} 
an axiomatization for affectible compliance will be provided and we shall show how all the above mentioned notions are related with each other: the Main Theorem I will essentially state that
 the abstract notion of affectible compliance and its
game-theoretic interpretation are but an abstract representation of both retractable and orchestrated compliance. The Main Theorem II will show instead how it is effectively possible to
get derivations, winning strategies and orchestrators out of each other. The definition of the subcontract relation, its axiomatization and decidability, together with the correspondence between
subcontract derivations and orchestrators functors will be the topic of Section \ref{sec:subcontractrel}. A Conclusion and Future Work section (Section \ref{sec:concl}) will be the last one before some appendices.
We shall use a simple working example through the various sections in order to clarify the
notions we introduce.
Many proofs and accessory formalisms will be detailed in the appendices at the end of the paper.

%

\section{ Affectible session contracts}
\label{sec:affectcompliance}


{\em Affectible session contracts} (affectible contracts for short) stem from {\em session contracts} \cite{BdL13,BernardiH13}\footnote{The name used in \cite{BdL13} was  actually {\em session behaviours}.}. With respect to session contracts, affectible contracts 
add the affectible output construct, which is called retractable  output in \cite{BDLdL15}.
The affectible output operator aims at representing points where the client/server interaction can be influenced by a third party or by an {\em agreement} between the two partners; consequently it is natural to use the CCS external choice operator as it is the case of the input branching (which is always affectible). Outputs in an internal choice are regarded as unaffectible actions and treated as unretractable in the setting of \cite{BDLdL15}.

\begin{defi}[Affectible session contracts]
\label{Adef:ckpt-behav}
Let $\Names$ {\em (names)} be some countable set of symbols and let $\CoNames = \Set{\Dual{a} \mid a \in \Names}$ {\em (co-names)}, with
$\Names\cap\CoNames = \emptyset$. \\
The set $\ASC$ of {\bf affectible session contracts} is defined as the set of the {\bf closed} (with respect to the binder $\rec$) expressions generated by the following 
grammar, 

\[\begin{array}{lcl@{\hspace{8mm}}l}
\sigma,\rho&:=& \mid ~ \stopA &\mbox{(success)} \\[1mm]
       &     & \mid ~\sum_{i\in I} a_i.\sigma_i  & \mbox{(input)} \\[1mm]
       &     & \mid ~\sum_{i\in I} \Dual{a}_i.\sigma_i& \mbox{(affectible output)}\\[1mm]
       &     & \mid ~\bigoplus_{i\in I} \Dual{a}_i.\sigma_i & \mbox{(unaffectible output)}\\[1mm]
       &     & \mid  ~x  & \mbox{(variable)}\\[1mm]
       &     & \mid ~\rec x. \sigma &  \mbox{(recursion)}
\end{array}
\]
where 
\begin{itemize}
\item $I$ is non-empty and finite;
\item names  and co-names in choices are pairwise distinct;
\item $\sigma$ is not a variable in $\rec x.\sigma$.
\end{itemize}
\end{defi} 
\noindent

The set $\Act = \Names\cup\CoNames$ is the set of \emph{actions}, ranged over by the metavariables $\alpha$, $\alpha'$, $\alpha_1$, $\alpha_2$, etc. On $\Act$ the usual involution
$(\Dual{\,\cdot\,})$ is defined, that is such that $\Dual{\Dual{\alpha}}=\alpha$.

\begin{notation}
\label{not:affcontracts}
As usual, when $I=\Set{1,..,n}$, we shall indifferently use the notations $a_1\Actdot \sigma_1 + \cdots + a_n\Actdot \sigma_n$ and $\External_{i\in I} a_i\Actdot \sigma_i$, as well as $\Dual{a}_1\Actdot \sigma_1 \oplus \cdots \oplus \Dual{a}_n\Actdot \sigma_n$ and $ \bigoplus_{i\in I} \Dual{a}_i\Actdot \sigma_i$. \\
We also shall write $\alpha_k.\sigma_k + \sigma'$ to denote $\External_{i\in I} \alpha_i.\sigma_i$ where $k\in I$ and $\sigma'=\External_{i\in (I\setminus\Set{k})} \alpha_i.\sigma_i$.
Similarly for the internal choice. \\
We assume $I$ never to be a singleton in $\External_{i\in I} \Dual{a}_i.\sigma_i$. This means that
a term like $a_k.\sigma$ has to be unambiguously read as $\bigoplus_{i\in \Set{k}} \Dual{a}_i.\sigma$ and not as $\External_{i\in \Set{k}} \Dual{a}_i.\sigma$
\end{notation}

Recursion is guarded and hence contractive in the usual sense.
Unless stated otherwise, we take the equi-recursive view of recursion, by equating $\rec x \procdot \sigma $ with $\sigma\Subst{x}{\rec x \procdot \sigma}$.
The trailing $\stopA$ is normally omitted: for example, we shall write $a + b$ for $a \Actdot \stopA + b \Actdot \stopA$. 
Affectible contracts will be considered modulo commutativity of internal and external choices.

The notion of compliance for (client-server) pairs of contracts is usually defined by means of an LTS.
A {\em client-server system} (a {\em system} for short) is a pair $\rho\pp \sigma$ of contracts, where $\rho$ plays the role of client and $\sigma$ of server; 
let the relation $\rho\pp \sigma\lts{}\rho'\pp \sigma'$ represent a communication step resulting into the new system $\rho',\sigma'$; now$\rho$ and $\sigma$ are  compliant if $\rho\pp \sigma\lts{*}\rho'\pp \sigma'\not\lts{}$ implies $\rho = \stopA$. Then one studies the properties of the compliance relation, possibly with reference to the contract syntax alone (see e.g. \cite{BdL13}). With affectible contracts, however, the semantics
of $\|$ is more complex as it involves a form of backtracking. For the present exposition we prefer to move from 
an abstract coinductive
definition for the affectible compliance relation. Later on we shall prove that defining compliance out of the retractable operational semantics (Section \ref{sec:retractablecompliance}) is equivalent to defining it out of the orchestrated operational semantics (Section \ref{sec:orchcompl}) and that both approaches are equivalent to the abstract affectible compliance below.


\begin{defi}\label{def:ACRel}
 The $\mathbb{A}$ffectible $\mathbb{C}$ompliance relation $\ACRel \subseteq\ASC\times\ASC$ is coinductively
defined as follows.
 Let $\rho,\sigma\in\ASC$. Then $\rho\ACRel\sigma$ if one of the following conditions holds:
\begin{enumerate}
\item \label{def:ACRel-1} $\rho = \stopA$;
\item \label{def:ACRel-2} $\rho = \sum_{i\in I}\alpha_i.\rho_i$, $\sigma = \sum_{j\in J}\Dual{\alpha}_j.\sigma_j$ and 
	$\exists k \in I \cap J.\; {\rho_k} \ACRel {\sigma_k}$;
\item \label{def:ACRel-3} $\rho = \bigoplus_{i\in I}\Dual{a}_i.\rho_i$, $\sigma = \sum_{j\in J}a_j.\sigma_j$,
	$I\subseteq J$ and $\forall k \in I. \; {\rho_k} \ACRel {\sigma_k}$;
\item \label{def:ACRel-4} $\rho = \sum_{i\in I}a_i.\rho_i$, $\sigma = \bigoplus_{j\in J}\Dual{a}_j.\sigma_j$,
	$I\supseteq J$ and $\forall k \in J. \; {\rho_k} \ACRel {\sigma_k}$.
\end{enumerate}

\end{defi}

\noindent
Let us informally describe the sense of the previous definition. Consider a system $\rho\pp \sigma$ of affectible contracts.
An unaffectible communication between $\rho$ and $\sigma$ is a synchonization involving
unaffectible outputs and inputs. In such a synchronization, any of the unaffectible outputs exhibited by one of the two contracts can be expected, since it depends
on an internal decision of the process whose behaviour is abstractedly represented by the 
contract.

 An affectible  communication, instead, is a synchronization involving affectible outputs and inputs. Affectable outputs are intended not to depend
on internal decisions, but to be {\em influenced  from the outside} in order to enable the system not to get to any stuck state.
In client/server interactions there is a bias towards the client. So a stuck state can be interpreted
by any pair $\rho\pp \sigma$ where $\rho \neq \stopA$ but no communication is possible.\\
So, in a system $\rho\pp \sigma$,  the server $\sigma$ is {\em affectible-compliant} with the client $\rho$ if either $\rho = \stopA$,  namely the client has successfully terminated; or all unaffectible communications of the system $\rho\pp \sigma$ lead to compliant systems; or there exists an affectible communication leading to a compliant system. 

In the above informal description, the worlds `` influenced  from the outside'' are rather abstract; they can be made concrete either via the characterization in terms of retractable computations, as done in Section \ref{sec:retractablecompliance}, or in terms of orchestrated interactions as done in
Section \ref{sec:orchcompl}.

\begin{exa}
{\em
Let us consider the following example from \cite{BDLdL15}.
A $\mathsf{Buyer}$ is looking for a bag ($\DBag$) or a belt ($\DBelt)$; she will decide how to pay, either by credit card ($\DCard$) or by cash ($\DCash$), after knowing the $\Price$ from the $\mathsf{Seller}$:
\[
\mathsf{Buyer} =\DBag.\Price.(\DCard \oplus \DCash) + \DBelt. \Price.(\DCard \oplus \DCash)
\]


\noindent
The $\mathsf{Seller}$ does not accept credit card payments for items of low price, like belts, but only for more expensive ones, like bags:
\[
\mathsf{Seller} = \Belt. \DPrice.\Cash + \Bag.\DPrice.(\Card + \Cash)
\]
From the previous definition it is not difficult to check that $\mathsf{Buyer}\ACRel\mathsf{Seller}$.
}
\end{exa}

\begin{rem}
\label{rem:noduality}
Notice that, unlike for session contracts, we have no notion of syntactical duality for affectible contracts, 
since it would not be definable in a natural way (i.e. in such a way duality be involutive and
a contract be always compliant with its dual).
In fact for affectible contracts we should define
$\Dual{\Dual{a}+\Dual{b}}$  either as $a+b$ or as $\Dual{a}\oplus\Dual{b}$. In both cases, however, we would lose the involutive property of the duality operator.  In the first case we would get  $\Dual{\Dual{\Dual{a}+\Dual{b}}} = \Dual{a+b}=\Dual{a}\oplus \Dual{b}$ (since
the duality operator should reasonably be a conservative extension over session contracts, and hence $\Dual{a+b}=\Dual{a}\oplus \Dual{b}$). A similar problem would arise  
by defining $\Dual{\Dual{a}+\Dual{b}}=\Dual{a}\oplus\Dual{b}$.
\end{rem}

%

\section{Game-theoretic interpretation of affectible contracts}
\label{sec:gameinterpr}

Following \cite{BartJLAMP} we interpret client-server systems of affectible contracts as games over event structures. 
This yields a game-theoretic interpretation of affectible compliance that will be of use to relate this last notion to both retractable and orchestrated compliance.


For the reader's convenience we briefly recall the basic notions of event structure and game associated to an LTS.


\begin{defi}[Event structure \cite{WinskelG:eves}]
Let $\mathbf{E}$ be a denumerable universe of {\em events}, ranged over by $e,e', \ldots$,
and let $\mathbf A$ be a universe of {\em action labels}, ranged over by $\alpha,\alpha',\beta,\ldots$.
Besides, let $E\subseteq \mathbf{E}$ and let $\#\subseteq E\times E$  be 
an irreflexive and symmetric relation (called {\em conflict relation}).
\begin{enumerate}
\item The predicate {\it CF} on sets $X\subseteq E$ and the set {\it Con} of finite {\em conflict-free} sets are defined~by
\lmcscenterline{
{\it CF}$(X)=\forall e,e'\in X. \neg(e\#e')$ \hspace{10mm}{\it Con} = $\Set{X\subseteq_{\mathit{fin}} E \mid \mbox{\it CF}(X)}$}

\item
An {\em event structure} is a quadruple $\mathcal{E} = (E,\#,\vdash,\lab)$  where
\begin{itemize}
\item
$\vdash\,\subseteq \mbox{\it Con}\times E$ is a relation 
such that
$sat(\vdash)=\;\vdash$~ (i.e. $\vdash$ is {\em saturated}), where 
$sat(\vdash)=\,\Set{(Y,e)\mid X\vdash e \And X\subseteq Y\in \mbox{\it Con}}$;
\item 
$\lab: E\rightarrow\mathbf{A}$ is a labelling function.
\end{itemize} 
\end{enumerate}
\end{defi}

\noindent
Given a set $E$ of events, $E^\infty$ denotes the set of sequences (both finite and infinite)
of its elements.
We denote by $\vec{e} = \seq{e_0 e_1\cdots }$ (or simply $e_0 e_1\cdots$) a sequence of events\footnote{Differently than in \cite{BartJLAMP}, we use the notation $\vec{e}$ for sequences instead of $\sigma$, which refers to a contract here.}.
Given $\vec{e}$, we denote by $\setof{\vec{e}}$ the set of its elements,
 by $|\vec{e}|$ its length (either a natural number or $\infty$) and by $\upto{i}{\vec{e}}$, for $i < |\vec{e}|$, 
 the subsequence $\seq{e_0e_1\cdots e_{i-1}}$ of its first $i$ elements.
Given a set $X$ we denote by $|X|$ its cardinality. $\Nat$ is the set of natural numbers.
The symbol `$\rightharpoonup$' will be used, as usual, to denote partial mappings.


\begin{defi}[LTS over configurations \cite{BartJLAMP}]
Given an event structure  {$\es = (E,\#,\vdash, \lab)$}, we define the LTS
$(\finParts{E},E,\rightarrow_\es)$ (where $\finParts{E}$ is the set of states, $E$ the set of labels and the labelled transition $\rightarrow_\es$)  as follows:
\lmcscenterline{
$C\lts{e}_{\!\!\es} C\cup\Set{e}$ \hspace{2mm}if  \hspace{2mm} $C\vdash e$, $e\not\in C$ ~and~ $CF(C\cup\Set{e})$}
\end{defi}
\noindent
We shall omit the subscript in $\rightarrow_\es$ when clear from the context.
For sake of brevity we shall often denote an LTS $(S,L,\lts{})$ by simply $(S,\lts{})$ or $\lts{}$.

Given an LTS $(S,\rightarrow)$ and a state $s\in S$, we denote by $\Trace{s,\rightarrow}$ the set of the (finite or infinite) traces in $(S, \rightarrow)$ starting in $s$, that is
$\Trace{s,\rightarrow} = \Set{s_0 s_1\cdots \mid s_0=s, s_i\lts{}s_{i+1}}$.

\subsection{Multi-player games}
 All the subsequent definitions and terminology are from \cite{BartJLAMP}, except in the case of games 
that we call ``multi-player'' instead of ``contracts'', which would be confusing in the present setting.

A set of participants (players) to a game will be denoted by $\partecipants$, whereas the universe of participants is denoted by $\Puniverse$. We shall use $A$, $B$,\ldots as variables ranging over $\partecipants$ or $\Puniverse$. The symbols $\playerA$, $\playerB$, \ldots will denote particular elements of $\partecipants$ or $\Puniverse$.
We assume that each event is associated to a player by means of a function 
$\pi : \mathbf{E} \rightarrow \Puniverse$. Moreover, given $A\in\Puniverse$ we define $\euniverse_{A} =\Set{e\in \mathbf{E} \mid \pi(e)=	A}$.

\begin{defi}[Multi-player game]\leavevmode
\begin{enumerate}
\item
A {\em game} $\game$ is a pair $(\es,\Phi)$ where
$\es = (E,\#,\vdash, \lab)$ is an event structure
and\linebreak
$\Phi : \Puniverse \rightharpoonup E^\infty \rightarrow \Set{-1,0,1}$  associates each participant and trace with a {\em payoff}.\\
Moreover, for all $X\vdash e$ in $\es$, $\Phi(\pi(e))$ is defined. We say that $\game$ is a game with participants $\mathfrak{P}$ whenever $\Phi A$ is defined for each player $A$ in $\mathfrak{P}$.
\item
A {\em play} of a game $\game=(\es,\Phi)$ is a (finite or infinite) trace of $(\emptyset,\rightarrow_\es)$ i.e. an element of 
$\Trace{\emptyset,\rightarrow_\es}$.
\end{enumerate}
\end{defi}

\begin{defi}[Strategy and conformance]\leavevmode
\begin{enumerate}
\item
A {\em strategy} $\Sigma$ for a participant $A$ in a game $\game$ is a function which maps each finite
play $\vec{e}{=}\seq{e_0\cdots e_n}$ to a  (possibly empty) subset of $\euniverse_A$ 
such that: ~
$
e{\in}\Sigma(\vec{e}) \implies \vec{e}e \mbox{ is a play of $\game$.}
$
\item
A play $\vec{e}=\seq{e_0e_1\cdots }$ {\em conforms} to a strategy $\Sigma$ for a participant $A$ in $\game$ if, for all $0\leq i < |\vec{e}|$,
~$
e_i\in\euniverse_A  \implies e_i\in\Sigma(\upto{i}{\vec{e}}).$
\end{enumerate}
\end{defi}

In general there are neither a turn rule nor alternation of players, similarly to concurrent games in \cite{AM99}.
A strategy $\Sigma$ provides ``suggestions'' to some player on how to legally move continuing finite plays (also called ``positions'' in game-theoretic literature). But $\Sigma$ may be ambiguous at some places, since $\Sigma(\vec{e})$ may contain more than an event; in fact it can be viewed as a partial mapping which is undefined when
$\Sigma(\vec{e}) = \emptyset$. 

We refer to \cite{BartJLAMP} for the general definition of winning strategy for multi-player games
(briefly recalled also in Remark \ref{rem:BartStrat} below),
since it involves the conditions of fairness and innocence, which will be trivially satisfied in our  interpretation of affectible client-server systems, where the notion of winning strategy
corresponds to the one that will be given in Definition \ref{def:Phi}.

We define now the notion of {\em univocal strategy}. When showing the equivalence between the various notions of compliance and the existence of winning strategies, we shall restrict to 
univocal strategies for the sake of simplicity.

\begin{defi}[Univocal strategies]
\label{def:univocalstrat}
$\Sigma$ is {\em univocal} if $\forall \vec{e}.~|\Sigma(\vec{e})| \leq 1$.
\end{defi}

\subsection{Turn-based operational semantics and compliance}
 Toward the game theoretic interpretation of 
a client-server system $\rho\pp\sigma$, we introduce an operational semantics of affectible contracts, 
making explicit the idea of a three-player game.
We interpret the internal choices and the  input actions of the client as moves of a player $\playerA$ and the internal choices and the 
input actions of the server as moves of a player $\playerB$. The synchronisations due to
affectible choices are instead interpreted as moves of the third player $\playerC$.

From a technical point of view this is a slight generalization and adaptation to our scenario of the turn-based semantics of ``session types''
in \cite{BartJLAMP}, \S 5.2. The changes are needed both because we have three players instead of two, and because session types are just session contracts, that is affectible contracts without affectible outputs.

\begin{defi}[Single-buffered $\ASC$]
The set $\sctb$ of {\em single-buffered} affectible contracts is defined by
\[
\tbASC = \ASC\cup\Set{\Zero}\cup\Set{\buf{\Dual{a}}\sigma \mid  \Dual{a} \in \CoNames, \sigma\in \ASC}.\]
\end{defi}
\noindent
We use the symbols $\tilde\rho, \tilde\sigma,\tilde\rho',\tilde\sigma'\ldots$ to denote elements of $\tbASC$.
A {\em turn-based configuration} (a configuration for short) is a pair $\tilde{\rho} \tbpp \tilde{\sigma}$,
where $\tilde\rho,\tilde\sigma\in\tbASC$.

As in \cite{BartJLAMP}, we have added the ``single buffered'' contracts $\buf{\Dual{a}}\sigma$ to represent the situation in which $\Dual{a}$ is the only output offered after an internal choice. Since the actual synchronization takes place in a subsequent step, $\Dual{a}$ is ``buffered'' in front of the continuation $\sigma$.
Such a technical device is adopted for two separate motivations. First, since we wish our game interpretation of configurations to extend the one provided in \cite{BartJLAMP}. Second, since by means of that we shall manage to get easier proofs. In fact, using the following turn-based operational semantics any move of player $\playerA$, unless getting to a stuck state, is necessarily followed by a move of $\playerB$; and any move of player $\playerB$  is necessarily preceded by a move of $\playerA$.


\begin{defi}[Turn-based operational semantics of configurations]
\label{def:tbLTS}
Let  $\tbAct = \{ \playerA,\playerB,\playerC\} \times (\Act \cup \Set{\cmark})$, where $\playerA,\playerB,\playerC$ are particular elements of $\Puniverse$.\\
In Figure \ref{fig:tb-opsem} we define the LTS $\tblts{}$ over
turn-based configurations, with labels in $\tbAct$.
\end{defi}
\noindent
An element $\playerA{:}\Dual{a}$ (resp. $\playerA{:}\Dual{a}$ ) of $\tbAct$ represents an output (resp. input) action performed by
player $\playerA$. Similarly for $\playerB$. An element $\playerC{:}a$ represents one of the possible way the player $\playerC$ can affect the interaction between affectible-outputs and inputs. An element $\playerC{:}\cmark$ represents instead a ``winning'' move.
We use the symbols $\beta,\beta',\ldots$ to denote elements of $\tbAct$.

Comparing $\tblts{}$ with the usual LTS for session contracts \cite{BdL13,BH13}, we observe that $\buf{\Dual{a}}\sigma$ is a duplicate of
$\Dual{a}.\sigma$, with the only difference that now there is a redundant step in \linebreak
$\Internal_{i\in I}\Dual{a}_i.\rho_i\tbpp \tilde\sigma \tblts{\playerA:\Dual{a}_k}\buf{\Dual{a}_k}\rho_k \tbpp \tilde\sigma$ when $I$ is the 
singleton $\Set{k}$. Also, besides the reductions concerning the affectible choices, we have the new reduction $\stopA\tbpp \tilde{\rho} \tblts{\playerC:\cmark} \Zero \tbpp \tilde{\rho}$ to signal when
player $\playerC$ wins.

\noindent
Let $\vec{\beta}\!=\!\seq{\beta_1\cdots\beta_n}\!\in\!\tbAct^*$. We shall use the notation $\tblts{\vec{\beta}}=\tblts{\beta_1}\!\!\circ\cdots\circ\!\!\tblts{\beta_n}$.\\

\begin{figure}[t]
\hrule
\vspace*{2mm}
\[ \begin{array}{rcl@{\hspace{9mm}}rcl}

\Internal_{i\in I}\Dual{a}_i.\rho_i\tbpp \tilde\sigma &\tblts{\playerA:\Dual{a}_k} &\buf{\Dual{a}_k}\rho_k \tbpp \tilde\sigma&



\External_{i\in I} a_i.\rho_i \tbpp \buf{\Dual{a}_k} \sigma &\tblts{\playerA:a_k}& \rho_k \tbpp \sigma \\[2mm]

\tilde\rho\tbpp \Internal_{i\in I}\Dual{a}_i.\sigma_i  & \tblts{\playerB:\Dual{a}_k} &\tilde\rho\tbpp\buf{\Dual{a}_k}\sigma_k &

\buf{\Dual{a}_k} \rho \tbpp \External_{i\in I} a_i.\sigma_i  &\tblts{\playerB:a_k}& \rho \tbpp \sigma_k \\[2mm]

\Dual{a}.\rho + \rho' \tbpp a. \sigma + \sigma'  &\tblts{\playerC:a}& \rho\tbpp \sigma
&
a.\rho + \rho' \tbpp \Dual{a}. \sigma + \sigma'  &\tblts{\playerC:a}& \rho\tbpp \sigma\\[2mm]

\Dual{a}.\rho + \rho' \tbpp a. \sigma &\tblts{\playerC:a}& \rho\tbpp \sigma
&
a.\rho \tbpp \Dual{a}. \sigma + \sigma'  &\tblts{\playerC:a}& \rho\tbpp \sigma\\[2mm]

\multicolumn{6}{c}{
\stopA\tbpp \tilde{\rho} ~~\tblts{\playerC:\cmark} ~~\Zero \tbpp \tilde{\rho}
}\\
\mbox{where } (k\in I)
\end{array} \]
\vspace{-4mm}
\caption{Turn-based operational semantics of turn-based configurations}\label{fig:tb-opsem}
\vspace{8pt}
\hrule
\end{figure}

The relation of turn-based compliance on single-buffered affectible contracts (and hence on affectible contracts) is defined coinductively using the operational semantics formalised by the LTS $\tblts{}$, as follows.

\begin{defi}[Turn-Based Compliance Relation $\complyTB$]\label{ccr}\leavevmode
\begin{enumerate}[beginpenalty=99]
\item \label{ccr1}
Let
$\FunH: \mathcal{P}(\sctb\!{\times}\!\sctb){\rightarrow}  \mathcal{P}(\sctb\!{\times}\!\sctb)$ be such that, 
 for any   $\Rel\!\subseteq\! \sctb\!{\times}\!\sctb$, we have
$(\tilde\rho,\tilde\sigma) \in \mathcal{H}(\Rel)$\, if \,the following conditions hold:

\begin{enumerate}
\item\label{c1} 
$\tilde\rho\tbpp \tilde\sigma \, \notblts{}$~~ implies ~~ $\rho=\Zero$;

\item\label{c2} 
$\forall \tilde\rho',\tilde\sigma' $.~  $ [~ \tilde\rho\tbpp \tilde\sigma \tblts{\beta} {\tilde\rho'}\tbpp {\tilde\sigma'}$~~ implies~~ ${\tilde\rho'}\;\Rel \;{\tilde\sigma'}~]$,\\~~where $\beta\in \{\playerA{:}a,\playerA{:}\Dual{a},\playerB{:}a,\playerB{:}\Dual{a} \mid a\in\Names\}$;

\item\label{c3} 
$\exists a\in\Names . \tilde\rho\tbpp \tilde\sigma \!\!\tblts{\playerC:a}$  implies ~ $\exists \tilde\rho',\tilde\sigma',a $.  $[\tilde\rho\tbpp \tilde\sigma \tblts{\playerC:a} {\tilde\rho'}\tbpp {\tilde\sigma'}$ and ${\tilde\rho'}\;\Rel \;{\tilde\sigma'}]$;
\end{enumerate}
\item \label{ccr2} 
A relation $\Rel\subseteq \sctb\times\sctb$ is a
{\em turn-based compliance relation\/} if 
\hbox{$\Rel\subseteq \FunH(\!\Rel\!)$.} \\
$\complyTB$ is the greatest solution of the equation $X=\FunH(X)$, that is ~~~$\complyTB~=~\nu\FunH$.
\item \label{ccr3}
For $\rho,\sigma\in\ASC$,
we say that $\sigma$ is {\em turn-based compliant} with $\rho$  if ~
$\rho\complyTB \sigma$.
\end{enumerate}
\end{defi}

We can now show that turn-based compliance restricted to contracts in $\ASC$ and affectible compliance do coincide.

\begin{prop}[$\ACRel$ and $\complyTB$ equivalence.]
\label{th:complyequivtbcomply}
Let $\rho,\sigma\in\ASC$.
$$ \rho\complyTB\sigma ~~~~  \Iff   ~~~~ \rho\ACRel\sigma.$$
\end{prop}
\proof
See  Appendix \ref{appendix:complyequivtbcomply}
\qed

It will be relatively easy to get a game interpretation of the relation $\complyTB$. So the proposition \ref{th:complyequivtbcomply} will be used to get a game interpretation for
the relation $\ACRel$.


\subsection{Three-player game interpretation for $\ASC$ client/server systems.}
Using the turn-based semantics, we associate to any client/server system an event structure, and then a three-player game\footnote{
Such interpretation is called {\em semantic-based} in \cite{BartJLAMP} and it applies quite naturally to our context. 
Instead the {\em syntax-based} approach (which is equivalent to the semantic-based one in a two-players setting; see \cite{BartJLAMP} \S 5.3.2) cannot be straightforwardly extended to a three-player game.}, extending the treatment of session types with two-player games in \cite{BartJLAMP}. For our purposes we just consider the  LTS of a given 
client/server system instead of an arbitrary one.

\begin{defi}[ES of affectible-contracts systems]
\label{def:esacs}
Let $\rho\pp\sigma$ be a client/server system of affectible contracts. We define the event structure
$\Sem{\rho\pp\sigma}=(E,\#,\vdash, \lab)$, where
\begin{itemize}
\item
$E=\Set{(n,\beta) \mid n\in\Nat, \beta \in \tbAct}$
\item
$\# = \Set{((n,\beta_1),(n,\beta_2)) \mid n\in\Nat, \beta_1,\beta_2 \in \tbAct, \beta_1\neq\beta_2}$
\item
$\vdash\, =\, \sat{\vdash_{\!\rho\!\pp\!\sigma}}$\\[-2mm]
where 
$\vdash_{\!\rho\!\pp\!\sigma} = \Set{(X,(n,\beta))  \mid \rho\tbpp\sigma \tblts{\snd{X}} \tilde{\rho}'\tbpp\tilde{\sigma}' \!\!\tblts{\beta}\mbox{ and } n=|X|+1}$
\item
$\lab(n,\beta) = \beta$
\end{itemize}
where the partial function $\snd{\mbox{-}}$ maps any $X= \Set{(i,\beta_i)}_{i=1..n}$ to $\seq{\beta_1\cdots\beta_n}$, and it is undefined over sets not of the shape of $X$.
\end{defi}

Events in $\Sem{\rho\pp\sigma}$ are actions in $\tbAct$ paired with time stamps. Two events are in conflict if different actions should be performed at the same time, so that configurations must be linearly ordered w.r.t. time. The relation 
$X \vdash_{\!\rho\!\pp\!\sigma} (n,\beta)$ holds if $X$ is a trace in the LTS of $\rho\pp\sigma$ of length $n-1$; therefore the enabling
$Y \vdash (n,\beta)$ holds if and only if $Y$ includes a trace of length $n-1$ that can be prolonged by $\beta$, possibly including $(n,\beta)$ itself and any other action that might occur after $\beta$ in the LTS.

\begin{exa}
By the above definition, $\vdBS$ in $\Sem{\mathsf{Buyer}\pp\mathsf{Seller}}$
corresponds to
\[\begin{array}{llr}
\big\{ &\emptyset\vdBS (1,(\playerC{:}\Belt)),\ \emptyset\vdBS (1,(\playerC{:}\Bag)),\\
& \Set{\!(1,(\playerC{:}\Belt))\!}  \vdBS (2,(\playerB{:}\DPrice)),
\Set{\!(1,(\playerC{:}\Bag))\!} \vdBS (2, (\playerB{:}\DPrice)),\\
& \Set{\!(1,(\playerC{:}\Belt)), (2,(\Slr{:}\DPrice))\!}\vdBS (3,(\playerA{:}\Price)), \ldots\\
& \ldots X_1 \vdBS (6,(\playerC,\cmark))& \big\}
\end{array}\]
where $X_1 =  \Set{\!(1,(\playerC{:}\Bag)), (2,(\playerB{:}\DPrice)), (3,(\playerA{:}\Price)),
(4,(\playerA{:}\DCash)),(5,(\playerB{:}\Cash))
 }$.
The $\vdash_{\!\rho\!\pp\!\sigma}$ of this simple example is finite. It is not so in general for systems with recursive contracts.
\end{exa}

\noindent
The following definition is a specialisation of Definitions 4.6 and 4.7 in \cite{BartJLAMP}.\\
We use $\MaxTrace{s,\rightarrow}$ and $\FinMaxTrace{s,\rightarrow}$ to denote the
set of maximal traces and finite maximal traces, respectively, of $\Trace{s,\rightarrow}$.

\begin{defi}[Game $\game_{\rho\!\pp\!\sigma}$]
\label{def:Phi}\noindent
Given $\rho,\sigma \in \ASC$, we define the game $\game_{\rho\!\pp\!\sigma}$ as
$(\Sem{\rho\pp\sigma},\Phi)$, where:
$\pi(n,\beta) = A$ if $\beta = A{:}\alpha$,\hspace{2mm}$\Phi A$ is defined only if $A\in\Set{\playerA,\playerB,\playerC}$ and\\
\[\Phi A \vec{e} = \left\{ 
\begin{array}{c@{~~~~}l@{~~~~}l}
1 &  \mbox{ if }~~ \mathbf{P}(A,\vec{e})\\
-1 & \mbox{ otherwise }
\end{array}\right.
\]
where $\mathbf{P}(A,\vec{e})$ holds whenever
\[
\vec{e}\!\in \!\Trace{\emptyset,\rightarrow_{\Sem{\rho\!\pp\!\sigma}}}
\!\And \! [\vec{e}\!\in\! \FinMaxTrace{\emptyset,\rightarrow_{\Sem{\rho\!\pp\!\sigma}}}
\implies \exists \vec{e'}\!,\!n.\,\vec{e}= \vec{e'}(n,\!(A{:}\cmark))].
\]
In words $\mathbf{P}(A,\vec{e})$ holds if whenever $\vec{e}$ is a maximal trace in the
LTS $\Trace{\emptyset,\rightarrow_{\Sem{\rho\!\pp\!\sigma}}}$ which is finite, then it ends
by the event $(n,\!(A{:}\cmark)$, where $n$ is just a time stamp, $A{:}\cmark$ is the successful action
performed by the participant $A$. Note that if $\vec{e}$ is not a trace of the LTS or it is not maximal, then
$\mathbf{P}(A,\vec{e})$ trivially holds.

A player $A$ {\em wins} in the sequence of events $\vec{e}$ if 
$\Phi A\,\vec{e} > 0$.
A strategy $\Sigma$ for player $A$ is {\em winning} if $A$ wins in all  plays conforming to $\Sigma$.
\end{defi}
Note that, if an element $\vec{e}$ of $\Trace{\emptyset,\rightarrow_{\Sem{\rho\!\pp\!\sigma}}}$ is infinite, $\mathbf{P}(A,\vec{e})$ holds for any $A$ since the implication is
vacuously satisfied. If  $\vec{e}$ is finite, instead, $\mathbf{P}(A,\vec{e})$ holds only in case
the last element of the sequence $\vec{e}$ is of the form $(n,\!(A{:}\cmark))$.

\begin{exa}
\label{ex:gamestrat}
For the game $\game_{\mathsf{Buyer}\pp\mathsf{Seller}}$, it is possible to check that,
for instance,
\begin{center}
$\Phi \playerC \vec{s_1} = 1$,~~~ $\Phi \playerA \vec{s_1} = -1$,~~~ $\Phi \playerB \vec{s_2} = -1$,~~~ $\Phi \playerC \vec{s_3} = -1$
\end{center}
where
\begin{center}
  \begin{tabular}{l}
 $\vec{s_1} \!=${\small  $(1,(\playerC{:}\Bag))(2,(\playerB{:}\DPrice))(3,(\playerA{:}\Price))
(4,(\playerA{:}\DCash))(5,(\playerB{:}\Cash))(6,(\playerC,\cmark))$},\\
$\vec{s_2} =$ {\small  $(4,(\playerA{:}\Bag))(1,(\playerC{:}\DPrice))$}\\
$\vec{s_3} =$ {\small  $(1,(\playerC{:}\Bag))(2,(\playerB{:}\DPrice))(3,(\playerA{:}\Price))
(4,(\playerA{:}\DCash))(5,(\playerB{:}\Cash))$}
\end{tabular}
\end{center}

Let us define a particular strategy $\widetilde{\Sigma}$ for $\playerC$ in $\game_{\mathsf{Buyer}\pp\mathsf{Seller}}$ as follows:
\begin{equation}
\label{eq:tildesigma}
\widetilde{\Sigma}(\vec{s}) = 
\left\{ \begin{array}{l@{~~~}l}
       \Set{(1,(\playerC{:}\Bag))} & \mbox{ if } \vec{s}=\seq{}\\
       \Set{(6,(\playerC,\cmark))} & \mbox{ if } \vec{s}=\vec{s_3}\\
        \emptyset                       & \mbox{ for any other play }
          \end{array} \right.
\end{equation}
The strategy $\widetilde{\Sigma}$ for $\playerC$ in $\game_{\mathsf{Buyer}\pp\mathsf{Seller}}$ is winning (actually it is the only winning strategy for the game of this example).
Moreover, $\widetilde{\Sigma}$ is univocal (in our simple example there are not non univocal strategies for player $\playerC$).
\end{exa}

In a game $\game_{\rho\!\pp\!\sigma}$ it is not restrictive to look, for player $\playerC$, at univocal strategies only, as established in the next lemma.

We say that $\Sigma$ {\em refines} $\Sigma'$, written $\Sigma\leq\Sigma'$, if and only if $\Sigma(\vec{e}) \subseteq \Sigma'(\vec{e})$ for all $\vec{e}$.

\begin{lem}\label{lem:refinement}
If $\playerC$ has a winning strategy $\Sigma$, then $\playerC$ has a univocal winning strategy $\Sigma'$ such that $\Sigma'\leq\Sigma$.
\end{lem}

\begin{proof}
Let $\Sigma(\vec{e}) = \Set{e_1,\ldots, e_{n}}$. 
Since events are a countable set we may suppose w.l.o.g. that there is a total and well founded ordering over $E$, so that if $\Set{e_1,\ldots, e_{n}} \neq \emptyset$ there exists an $e' = \min \Set{e_1,\ldots, e_{n}}$;
let us define $\Sigma'(\vec{e}) = e'$.

By definition $\Sigma'$ is
a univocal refinement of $\Sigma$ such that any play 
which conforms $\Sigma'$ conforms $\Sigma$ as well. 
Now if $\Sigma$ is winning then the payoff $\Phi \playerC \vec{e} = 1$ for any play $\vec{e}$ conforming $\Sigma$ and hence a fortiori for any play conforming $\Sigma'$; hence $\Sigma'$ is winning for $\playerC$.
\end{proof}

\begin{rem}
\label{rem:BartStrat}
According to \cite{BartJLAMP}, $A$ wins in a play if $\wf A \vec{e} > 0$, where
$\wf A \vec{e} = \Phi A \vec{e}$ if all players are ``innocent'' in $\vec{e}$, while $\wf A \vec{e} = -1$ if $A$ is ``culpable''. Moreover, if $A$ is innocent and someone else culpable, $\wf A \vec{e}=  +1$.
A strategy $\Sigma$ of $A$ is winning if $A$ wins in all {\em fair} plays conforming to $\Sigma$.
A play $\vec{e}$ is ``fair'' for a strategy $\Sigma$ of a player $A$ if any event in $E_A$ which is infinitely often enabled is eventually performed. Symmetrically $A$ is ``innocent'' in $\vec{e}$ if she eventually plays all persistently enabled moves of her in $\vec{e}$, namely if she is fair to the other players, since the lack of a move by $A$ might obstacle the moves by others; she is ``culpable'' otherwise. 
As said above, Definition \ref{def:Phi} is a particularisation of the general definitions in \cite{BartJLAMP}. In fact in a game 
$\game_{\rho\!\pp\!\sigma}$
no move of any player can occur more than once in a play $\vec{e}$ because of time stamps. Therefore no move can be ``persistently enabled'',
nor it can be prevented, since it can be enabled with a given time stamp only if
there exists a legal transition in the LTS with the same label.
Hence any player is innocent in a play $\vec{e}$ of  $\game_{\rho\!\pp\!\sigma}$ and all plays are fair. Therefore $\wf$ coincides with $\Phi$.
\end{rem}

The last lemma is useful in proofs, and it is essentially a definition unfolding.

\begin{lem}
\label{lem:winC}
In a play $\vec{e}$ of a three-player game $\game_{\rho\!\pp\!\sigma}$, player $A$ wins if and only if \\
either $\vec{e}$ is infinite
or $A = \playerC$ and $\upto{k}{\vec{e}}=\seq{e_0e_1\cdots (k,\playerC\!:\!\cmark)}$ for some $k$.
\end{lem}

\begin{proof}
 By Definition \ref{def:Phi}
$\Phi A \vec{e} > 0$ if either $\vec{e}$ is infinite or it contains  a move $(k,(A:\cmark))$ which has to be a move by $A$. But the only player that can play such a move is $\playerC$ (see Fig.~\ref{fig:tb-opsem}). \end{proof}

Notice that, by Remark~\ref{rem:BartStrat}, in our setting this lemma holds also in case we take into account the definition of {\em winning in a play} provided in \cite{BartJLAMP}, where $\wf$ is used.

%

\section{Retractable operational semantics and retractable compliance.}
\label{sec:retractablecompliance}

We provide now an operational semantics for client-server systems of affectible contracts, 
based
on a rollback operation, as proposed in  \cite{BDLdL15}.
The rollback operation acts on the recording of the  branches discarded during the {\em retractable}
syncronizations.
The retractable syncronizations are those involving affectible outputs. The actual sending of such outputs can hence be looked at as depending on an {\em agreement} between client and server.  The computation can roll back to such agreement points
when the client/server interaction gets stuck. When this happens, the interaction branch that had been followed up to this moment is discarded and another, possibly successful branch of interaction is pursued.\\

We begin by defining the notion of {\em contracts with history} as follows:

\begin{defi}[\Ctrs\ with history]
Let $\Stacks$  be the set of expressions (referred to also as {\em stacks}) generated by the grammar:
\lmcscenterline{$\vec{\gamma} ::= \emptystack \mid  \vec{\gamma} \cons  \sigma$ ~~~
where $\sigma\in\ASC\cup\{\circ\}$.} Then the set of {\em \ctrs\ with history} is defined by:
\lmcscenterline{$\SbehavH = \{ \np{\vec{\gamma}}\sigma \mid \vec{\gamma}\in \Stacks, \sigma\in \ASC\cup\{\circ\}\, \}.$}
\end{defi}

Histories are finite lists of contracts representing the branches that have not been followed in the previous agreement points. The effect of rolling back to the last agreement point is modeled by restoring the last contract of the history as the current
contract and by  trying a different branch, if any. In case a contract in a history was a sum whose branches have all been tried, the empty sum of remaining alternatives is represented by $\circ$ (see the next definition of transition). We use the notation $\np{\vec{\gamma}}\sigma$ for the contract $\sigma$ with history $\gamma$.

This is formalised by the operational semantics of contracts with histories that is defined as follows.

\begin{defi}[LTS of \Ctrs\ with History]\label{scs}
{\small
\[\begin{array}{rl@{\hspace{16mm}}ll}
(+) & \np{\vec{\gamma}}{\alpha.\sigma + \sigma'}
			\ored{\alpha} 
	\np{\vec{\gamma} \cons \sigma'}\sigma  & 
(\oplus) &  \np{\vec{\gamma}}{{ \Dual{a}.\sigma \oplus \sigma'}}
			\ored{\tau} 
		\np{\vec{\gamma}}\Dual{a}.\sigma 
\\ 
(\alpha)
& \np{\vec{\gamma}}\alpha.\sigma \ored{\alpha} \np{\vec{\gamma}\cons\circ}{\sigma}
&
(\rb)  &
	\np{\vec{\gamma}\cons\sigma'}{\sigma}  \lts{\rb}  \np{\vec{\gamma}}\sigma'
\end{array}\]
}
\end{defi}

Rule $(+)$ formalises the selection of a branch of an external choice. The effect of the selection is that the unselected branches are memorised on top of the stack as last contract of the history.
Notice that the memorization on the stack occurs also in case of an external input choice, like for instance in
$\np{\vec{\gamma}}{a.\sigma_1 + b.\sigma_2}
			\ored{a} 
	\np{\vec{\gamma} \cons b.\sigma_2}\sigma_1$. This because a possible 
partner could be an external output choice and the resulting syncronization would hence be a retractable one.
However, in case the partner were an internal choice and hence in case the resulting syncronization were unretractable,
the memorized contract should not be taken into account by any future rollback. The operational semantics manages to
take care of such a possibility by means of the use of the symbol `$\circ$' (see Example \ref{ex:interaction} below).\\
Rule ($\oplus$) formalises the fact that when an internal choice occurs, the stack remains unchanged. Rule ($\alpha$) concerns instead the execution of a single action:  the history is modified by pushing a `$\circ$' on the stack, meaning that the only available branch has been tried and no alternative is left. 
Rule $(\rb)$ recovers the contract on the top of the stack 
(if the stack is different from $\emptystack$) and replaces the current one with it. Note that the combined effect of rules ($\oplus$) and $(\alpha)$ is that the alternative branches of an internal choice are unrecoverable.

The interaction of a client with a server is modeled by the reduction of their parallel composition, that can be either
forward, consisting of CCS style synchronisations and single internal choices, or backward if there is no possible forward reduction, the client is different than $\stopA$ (the fulfilled contract) and rule $(\rb)$ is applicable on both sides.

\begin{defi}[TS of Client/Server Pairs]\label{def:cspairlts} 
We define the relation $\ored{}$ over pairs of  \ctrs\ with histories by the following rules:
{\small
\[\begin{array}{c}
\prooftree
\np{\vec{\delta}}\rho \ored{\alpha} \np{\vec{\delta'}}{\rho'} \qquad \np{\vec{\gamma}}\sigma \ored{\Dual{\alpha}} \np{\vec{\gamma'}}{\sigma'} 
\justifies
\np{\vec{\delta}}\rho\pp \np{\vec{\gamma}}\sigma \ored{} \np{\vec{\delta'}}{\rho'}\pp \np{\vec{\gamma'}}{\sigma'}
\using (\text{\em \CommRule})
\endprooftree \\[5mm] 
\prooftree
\np{\vec{\delta}}\rho \ored{\tau} \np{\vec{\delta} }{\rho'} 
\justifies
\np{\vec{\delta}}\rho\pp \np{\vec{\gamma}}\sigma \ored{} \np{\vec{\delta}}{\rho'}\pp \np{\vec{\gamma}}{\sigma}
\using (\tau)
\endprooftree\\[5mm] 
\prooftree
\np{\vec{\gamma}}{\rho} \lts{\rb} \np{\vec{\gamma'}}{\rho'} \qquad \np{\vec{\delta}}{\sigma} \lts{\rb} \np{\vec{\delta'}}{\sigma'} 
	\qquad  \rho\neq\stopA
\justifies
\np{\vec{\gamma}}{\rho}\pp\np{\vec{\delta}}{\sigma} \ored{} \np{\vec{\gamma'}}{\rho'}\pp\np{\vec{\delta'}}{\sigma'}
\using (\text{\em \RbkRule})
\endprooftree 
\end{array}\]
}

\noindent
plus the rule symmetric to $(\tau)$ w.r.t. $\|$. Moreover, rule $(\text{\em \RbkRule})$ applies only if neither $(\text{\em \CommRule})$ nor $(\tau)$ do.
\end{defi}

\begin{exa}
\label{ex:interaction}
Let us consider the following client/server system
\lmcscenterline{
$\np{\vec{\gamma}}{a.c.\rho_1 + b.\rho_2} \pp \np{\vec{\delta}}{\Dual{a}.d.\sigma_1 + \Dual{b}.\sigma_2}$}
We are in presence of an agreement point, and the operational semantics defined above 
is such that during the following syncronization the unchosen branches are memorized:
\lmcscenterline{
$\np{\vec{\gamma}}{a.c.\rho_1 + b.\rho_2} \pp \np{\vec{\delta}}{\Dual{a}.d.\sigma_1 + \Dual{b}.\sigma_2}
\ored{} 
\np{\vec{\gamma}\cons b.\rho_2}{c.\rho_1 } \pp \np{\vec{\delta}\cons \Dual{b}.\sigma_2}{d.\sigma_1}
$}
Now, no syncronization is possible in the system $\np{\vec{\gamma}\cons b.\rho_2}{c.\rho_1 } \pp \np{\vec{\delta}\cons \Dual{b}.\sigma_2}{d.\sigma_1}$, and hence the rules dealing with rollback allow to recover the 
branches which were not selected in the previous agreement point:
\lmcscenterline{
$\np{\vec{\gamma}\cons b.\rho_2}{c.\rho_1 } \pp \np{\vec{\delta}\cons \Dual{b}.\sigma_2}{d.\sigma_1}
\ored{} 
\np{\vec{\gamma}}{ b.\rho_2} \pp \np{\vec{\delta}}{\Dual{b}.\sigma_2}
$}
The interaction can now proceed along the previously unchosen branches.

Let us consider now, instead, the following client/server system:
\lmcscenterline{
$\np{\vec{\gamma}}{a.c.\rho_1 + b.\rho_2} \pp \np{\vec{\delta}}{\Dual{a}.d.\sigma_1 \oplus \Dual{b}.\sigma_2}$}
We are not in presence of an agreement point and the interaction involves an internal choice
by the server:
\lmcscenterline{
$\np{\vec{\gamma}}{a.c.\rho_1 + b.\rho_2} \pp \np{\vec{\delta}}{\Dual{a}.d.\sigma_1 \oplus \Dual{b}.\sigma_2}
\ored{} 
\np{\vec{\gamma}}{a.c.\rho_1 + b.\rho_2} \pp \np{\vec{\delta}}{\Dual{a}.d.\sigma_1}
$}
Now the syncronization in $\np{\vec{\gamma}}{a.c.\rho_1 + b.\rho_2} \pp \np{\vec{\delta}}{\Dual{a}.d.\sigma_1}$ causes, in the client, the memorization of $b.\rho_2$, namely the client's branch not selected by the 
output $\Dual{a}$. In the server, instead, the single action $\Dual{a}$ causes
the memorization of `\,$\circ$'  on the stack:
\lmcscenterline{
$\np{\vec{\gamma}}{a.c.\rho_1 + b.\rho_2} \pp \np{\vec{\delta}}{\Dual{a}.d.\sigma_1}
\ored{} 
\np{\vec{\gamma}\cons b.\rho_2}{c.\rho_1} \pp \np{\vec{\delta}\cons \circ}{d.\sigma_1}
$}
Since the system $\np{\vec{\gamma}\cons b.\rho_2}{c.\rho_1} \pp \np{\vec{\delta}\cons \circ}{d.\sigma_1}$ is stuck, we recover the tops of the stacks as current contracts:
\lmcscenterline{
$\np{\vec{\gamma}\cons b.\rho_2}{c.\rho_1} \pp \np{\vec{\delta}\cons \circ}{d.\sigma_1}
\ored{} 
\np{\vec{\gamma}}{ b.\rho_2} \pp \np{\vec{\delta}}{\circ}
$}
The fact that `\,$\circ$' cannot synchronize with anything forces us to apply again rule 
(\RbkRule) on $\np{\vec{\gamma}}{ b.\rho_2} \pp \np{\vec{\delta}}{\circ}$,
 that is to keep on rolling back
in order to get to the first agreement point memorized in $\vec{\gamma}$ and $\vec{\delta}$.
The contract $b.\rho_2$ and the symbol `\,$\circ$' recovered from the top of the stacks will be simply ``thrown away''.
\end{exa}

\noindent
Up to the rollback mechanism, compliance in the retractable setting is defined as usual with client-server contracts.

\begin{defi}[Retractable Compliance, $\complyR$]\label{ccr}\hfill
\begin{enumerate}
\item \label{c1}
The binary relation $\complyR$ on \ctrs\ with histories is defined as follows: \\
$\text{ for any } 
\vec{\delta'},\rho',\vec{\gamma'},\sigma'$,~
$\np{\vec{\delta}}\rho\complyR \np{\vec{\gamma}}\sigma$   holds whenever 
\lmcscenterline{$\np{\vec{\delta}}\rho\pp \np{\vec{\gamma}}\sigma
                    \ored{*}
                     \np{\vec{\delta'}}\rho'\pp \np{\vec{\gamma'}}\sigma'
                     \not\!\!\ored{} 
\mbox{ implies }\rho'=\stopA$}
\item\label{c2} The relation $\complyR$ on affectible \ctrs\ is defined by:~~
$\rho\complyR\sigma ~~~~\mbox{if}~~~~ \np\emptystack\rho\complyR \np\emptystack\sigma.$
\end{enumerate}
\end{defi}

\begin{exa}[\cite{BDLdL15}] \label{example1}
In the $\mathsf{Buyer}/\mathsf{Seller}$ example we have that, in case
a belt is agreed upon and the buyer decides to pay using her credit card, the system gets stuck in an unsuccessful state. This causes a rollback enabling a successful
state to be reached.\\
{\small
\[\begin{array}{lrcl}
  &
\np{\emptystack}{\mathsf{Buyer}}
 & \pp & \np{\emptystack}{\mathsf{Seller}}\\[2mm]

 \ored{\CommRule} &
	\np{\DBag.\Price.(\DCard \oplus \DCash)}{\Price.(\DCard \oplus \DCash)}
& \pp & \np{\Bag.\DPrice.(\Card + \Cash)}{\DPrice.\Cash}\\[2mm]

 \ored{\CommRule} &
	\np{\DBag.\Price.(\DCard \oplus \DCash)\cons \circ }{(\DCard \oplus \DCash)}
 & \pp & \np{\Bag.\DPrice.(\Card + \Cash)\cons \circ }{\Cash}\\[2mm]

 \ored{\tau} &
	\np{\DBag.\Price.(\DCard \oplus \DCash)\cons \circ }{\DCard}
& \pp & \np{\Bag.\DPrice.(\Card + \Cash)\cons \circ }{\Cash}\\[2mm]

\ored{\RbkRule} &
	\np{\DBag.\Price.(\DCard \oplus \DCash)}{\circ }
 & \pp & \np{\Bag.\DPrice.(\Card + \Cash)}{\circ }\\[2mm]

\ored{\RbkRule} &
	\np{\emptystack}{\DBag.\Price.(\DCard \oplus \DCash) }
 & \pp & \np{\emptystack}{\Bag.\DPrice.(\Card + \Cash)}\\[2mm]

\ored{\CommRule} &
	\np{\circ}{\Price.(\DCard \oplus \DCash) }
 & \pp & \np{\circ}{\DPrice.(\Card + \Cash)}\\[2mm]

\ored{\CommRule} &
	\np{\circ\cons\circ}{(\DCard \oplus \DCash) }
 & \pp & \np{\circ\cons\circ}{(\Card + \Cash)}\\[2mm]

\ored{\tau} &
	\np{\circ\cons\circ}{\DCard }
 & \pp & \np{\circ\cons\circ}{(\Card + \Cash)}\\[2mm]

\ored{\CommRule} &
	\np{\circ\cons\circ\cons\circ}{\stopA }
 & \pp & \np{\circ\cons\circ\cons\Cash}{\stopA}\\[2mm]

\,\,\not\!\!\ored{}
	
\end{array}\]
}

\noindent
There are other possible reduction sequences:
the one in which a belt is agreed upon but the buyer decides to pay by cash;
the one in which a bag is agreed upon and  the buyer decides to pay by cash;
the one in which a bag is agreed upon and  the buyer decides to pay by credit card.\\
It is possible to check that also for all these other reduction sequences a successful state is  reached. 
So we have that $\mathsf{Buyer}\complyR\mathsf{Seller}$.
\end{exa}


It can be checked that if
$\np{\vec{\gamma}}{\rho}\pp\np{\vec{\delta}}{\sigma} \ored{} \np{\vec{\gamma'}}{\rho'}\pp\np{\vec{\delta'}}{\sigma'}$ and 
$\vec{\gamma}$ and $\vec{\delta}$ have the same length, then this is also true of $\vec{\gamma'}$ and $\vec{\delta'}$.\\
Now, let $\np{\emptystack}\rho\pp \np{\emptystack}\sigma
                    \ored{*}
                     \np{\vec{\gamma'}}\rho'\pp \np{\vec{\delta'}}\sigma'
                     \not\!\!\ored{}$.
Since $\np{\vec{\gamma'}}\rho'\pp \np{\vec{\delta'}}\sigma'
                     \not\!\!\ored{}$, a fortiori
rule $(\text{\em \RbkRule})$ does not apply and then, by rule $(\rb)$,   either  $\vec{\gamma'}$ or $\vec{\delta'}$ must be empty. Besides, by what said before, they must have the same lenght. Hence
they are both empty.\\ This is stated in the item (\ref{lem:stack-len-1}) of the following lemma (see Appendix \ref{subsec:rollprop} for the proof)
and it will be essential, together with item (\ref{lem:stack-len-2}), to prove some relevant results in Section \ref{sec:affectible-retractable} about the retractable compliance relation.

\begin{lem}\label{lem:stack-len}\leavevmode
\begin{enumerate}
\item
\label{lem:stack-len-1}
If
	$\np{\emptystack}{\rho} \pp \np{\emptystack}{\sigma} \ored{*} \np{\vec{\delta}}{\rho'} \pp \np{\vec{\gamma}}{\sigma'} \not\!\!\ored{}$, then	 $\vec{\delta} = \vec{\gamma} = \emptystack$.
\item
\label{lem:stack-len-2}
 If
$\np {\vec{\delta}}{\rho} \complyR \np {\vec{\gamma}}{\sigma}$, then 
	 $\np {\vec{\delta'} \cons \vec{\delta}}{\rho} \complyR \np {\vec{\gamma'} \cons \vec{\gamma}}{\sigma}$
	 for all $\vec{\delta'}$, $\vec{\gamma'}$. 
\end{enumerate}
\end{lem}

%

\section{Orchestrated operational semantics and orchestrated compliance}
\label{sec:orchcompl}

In the present section we define an operational semantics for affectible contracts  and  for client-server systems, where
the interaction between a client and a server -- in particular when affectible outputs are concerned -- is driven by  a third process, a mediator, called {\em orchestrator}.
For a general discussion on orchestrators for contracts and session-contracts,
we refer to \cite{Padovani10,PadTCS} and \cite{BdL15} respectively.

\begin{defi}[LTS for affectible session contracts]\label{scs}
Let $\Act = \Names\cup\CoNames\cup\Set{\Dual{a}^+\mid \Dual{a}\in\CoNames}$.
\[\begin{array}{lrcl@{\hspace{16mm}}lrcl}
(+) &   a.\sigma + \sigma'
			&\ored{a} &
	\sigma
& 
(\Dual{+}) &   \Dual{a}.\sigma + \sigma'
			&\ored{\Dual{a}^+} &
	\sigma
\\[2mm]
(\oplus) &  \Dual{a}.\sigma \oplus \sigma'
			&\ored{} &
		\Dual{a}.\sigma
&
(\alpha) &   \alpha.\sigma
			&\ored{\alpha} &
	\sigma
\end{array}\]

\end{defi}
As for other orchestrated formalisms, the task of an orchestrator is to mediate the interaction between a client and a server, by selecting or constraining the possible interactions.
In the present setting orchestrators, that we dub {\em strategy-orchestrators}, are defined as a variant of the session-orchestrators of \cite{BdL15}, which in turn are a restriction of orchestrators in \cite{Padovani10}.

A strategy orchestrator does mediate between two affectible session contracts by selecting one of the possible affectible choices and constraining non-affectible ones.
The orchestration actions have two possible shapes: either
$\orchAct{\alpha}{\Dual{\alpha}}$, enabling the unaffectible
synchronization, or  
 $\orchAct{\alpha}{\Dual{\alpha}}^+$, enabling the affectible
synchronization. 

\begin{defi}[Strategy Orchestrators]\leavevmode 
\begin{enumerate}
\item
The set $\OrchAct$ of {\em strategy-orchestration actions} is defined by
\lmcscenterline{
$\OrchAct = \Set{\orchAct{\alpha}{\Dual{\alpha}} \mid \alpha \in \Names\cup\CoNames}\cup \Set{\orchAct{\alpha}{\Dual{\alpha}}^+ \mid \alpha \in \Names\cup\CoNames}
$}
We let $\mu,\mu',\ldots$ range over elements of $\OrchAct$ with the shape 
 $\orchAct{\alpha}{\Dual{\alpha}}$, and $\mu^+,{\mu'}^+,\ldots$ range over elements of $\OrchAct$ with the shape 
 $\orchAct{\alpha}{\Dual{\alpha}}^+\!$.
\item
We define the set $\Orch$ of {\em strategy orchestrators}, ranged over by $f, g, \ldots$, as the {\em closed} (with respect to the binder $\rec$) terms generated by the following grammar:
 \[ \begin{array}{lrl@{\hspace{4mm}}l@{\hspace{8mm}}l}
f, g & ::= & \stopf  & \mbox{(idle)}\\
 & \mid & \mu^+.f & \mbox{(prefix)} 
\\
 & \mid & \mu_1.f_1\vee\ldots\vee\mu_n.f_n & \mbox{(disjunction)}\\
 & \mid & x & \mbox{(variable)}\\
 & \mid & \rec x.f & \mbox{(recursion)}
\end{array} \] 
where the $\mu_i$'s in a disjunction are pairwise distinct. Moreover, 
we impose strategy orchestrators to be {\em contractive}, i.e.~the $f$ in $\rec x.f$ is assumed not to be a variable.
\end{enumerate}
\end{defi}

We write $\bigvee_{i\in I}\mu_i.f_i$ as short for $\mu_1.f_1\vee\ldots\vee\mu_n.f_n$, where
$I=\Set{1,\ldots, n}$.\\
If not stated otherwise, we consider recursive orchestrators up-to unfolding, that is we equate $\rec x \procdot f $ with $f\Subst{x}{\rec x \procdot f}$. We omit trailing $\stopf$'s.

Strategy orchestrators are ``simple orchestrators'' in \cite{Padovani10} and ``synchronous orchestrators'' in \cite{PadTCS},
but for the kind of prefixes which are allowed in a single prefix or in a disjunction. In fact a prefix
$\orchAct{\alpha}{\Dual{\alpha}}^+$ cannot occur in disjunctions, 
where all the orchestrators must be prefixed by $\orchAct{\alpha}{\Dual{\alpha}}$ actions.
When clear from the context we shall refer to stategy orchestrators just as orchestrators.

\begin{defi}[Strategy orchestrators LTS] 
We define the labelled transition system
$(\Orch, \OrchAct , \OrchStep{})$ by
{\small
 \[ \begin{array}{c@{\hspace{8mm}}c@{\hspace{8mm}}c@{\hspace{8mm}}c}
\mu^+.f \OrchStep{\mu^+} f 
 & 
 &
&
 (\bigvee_{i\in I}\mu_i.f_i) \,\OrchStep{\mu_k}\, f_k
 \hspace{3mm}(k\in I)
\end{array} \]
}
\end{defi}

An {\em orchestrated system} $\rho \pf{f} \sigma$ is the client-server system $\rho \pp \sigma$ whose interactions are mediated by the orchestrator $f$.

\begin{defi}[LTS for orchestrated systems\footnote{The present LTS corresponds to a three-way synchronous communication,
which could cause some problems  from the
implementation point of view. We abstract away from these problems in the
theoretical setting of the present paper; an actual implementation might modularize 
each orchestrated synchronizations into two distinct binary synchronizations.}]
\hfill\\
Let $\rho,\sigma\in\ASC$ and $f\in\Orch$. 
{\small
\[ \begin{array}{c@{\hspace{24mm}}c}
\prooftree
 \rho\lts{}\rho'
\justifies
 \rho \pf{f} \sigma \lts{} \rho' \pf{f} \sigma
\endprooftree
&
\prooftree
 \sigma\lts{}\sigma'
\justifies
 \rho \pf{f} \sigma \lts{} \rho \pf{f} \sigma'
\endprooftree
\\[6mm]
\prooftree
 \rho \lts{a}\rho' \quad f\ltsOrchplus{\Dual{a}}{a} f' \quad \sigma \lts{\Dual{a}^+}\sigma'
\justifies
 \rho \pf{f} \sigma \lts{+} \rho' \pf{f'} \sigma'
\endprooftree
&
\prooftree
 \rho \lts{\Dual{a}^+}\rho' \quad f\ltsOrchplus{a}{\Dual{a}} f' \quad \sigma \lts{a}\sigma'
\justifies
 \rho \pf{f} \sigma \lts{+} \rho' \pf{f'} \sigma'
\endprooftree
\end{array} \] 
\[ \begin{array}{c}
\prooftree
 \rho \lts{\alpha}\rho' \quad f\ltsOrch{\Dual{\alpha}}{\alpha} f' \quad \sigma \lts{\Dual{\alpha}}\sigma' 
\justifies
 \rho \pf{f} \sigma \lts{\tau} \rho' \pf{f} \sigma'
\using (\mbox{\small $\alpha\in\Names\cup\CoNames$})
\endprooftree
\end{array} \] 
}

\noindent
Moreover we define $\Lts{} \,=\, \ltsstar{}\circ \, (\lts{\tau}\cup\lts{+})$.
\end{defi}

In both transitions $\lts{+}$ and $\lts{\tau}$, synchronization may happen only if the orchestrator has a transition with the appropriate orchestration action. This is because in an orchestrated interaction both client and server are committed to the synchronizations allowed by the orchestrator only. 
Because of its structure an orchestrator always selects exactly one synchronisation of affectible actions on client and server side, while the disjunction of orchestrators represents the constraint that only certain synchronisations of unaffectible actions are permitted.

\begin{defi}[Strategy-orchestrated Compliance]\hfill
\label{def:disstrictcompl}
\begin{enumerate}
\item
\label{def:disstrictcompl-i}
$f: \rho\complyO \sigma$ ~~if ~~  for any $\rho'$ and $\sigma'$, the following holds:
\[ \begin{array}{rcl}
\rho \pf{f} \sigma \Lts{}^* \rho' \pf{f'} \sigma' \notLts{} &~~\mbox{implies}~~& \rho'=\mbox{\em $\stopA$}.
 \end{array} \]
\item
$ \begin{array}{@{}rcl}
\rho\complyO \sigma &~~\textrm{if}~~& \exists f .~\Pred[ f: \rho\complyO \sigma].
 \end{array} $

\end{enumerate}
\end{defi}



\begin{exa}\label{ex:orchderseq}
Let $f = \orchAct{\Bag}{\DBag}^+.\orchAct{\DPrice}{\Price}.(\orchAct{\Card}{\DCard}.\stopf\vee\orchAct{\Cash}{\DCash}.\stopf)$.
For our $\mathsf{Buyer}/\mathsf{Seller}$ example, we can notice that the orchestrator $f$ forces 
the buyer and the seller to agree on a bag, and then leaves to the client the choice of the type of
payment. Hence, in case the client chooses to pay by card, the orchestrated client/server interaction
proceeds as~follows.

\[\small\begin{array}{lrll}
  &
{\mathsf{Buyer}}
 & \pf{\orchAct{\sBag}{\sDBag}^+.\orchAct{\sDPrice}{\sPrice}.(\orchAct{\sCard}{\sDCard}.\stopf\vee\orchAct{\sCash}{\sDCash}.\stopf)} & {\mathsf{Seller}}\\[3mm]

 \ored{} &
	{\Price.(\DCard \oplus \DCash)}
& \pf{\orchAct{\sDPrice}{\sPrice}.(\orchAct{\sCard}{\sDCard}.\stopf\vee\orchAct{\sCash}{\sDCash}.\stopf)} & {\DPrice.(\Card + \Cash)}\\[3mm]

 \ored{} &
	{(\DCard \oplus \DCash)}
 & \pf{\orchAct{\sCard}{\sDCard}.\stopf\vee\orchAct{\sCash}{\sDCash}.\stopf} & {\Card + \Cash}\\[3mm]

 \ored{} &
	{\DCard}
& \pf{\orchAct{\sCard}{\sDCard}.\stopf\vee\orchAct{\sCash}{\sDCash}.\stopf} & {\Card + \Cash} \\[3mm]

\ored{} &
	{\stopA }
 & \pf{\stopf} & {\stopA }\\[3mm]
\,\,\not\!\!\ored{}
\end{array}\]

\noindent Since there are not infinite reduction sequences and a succesful state is reached for any other possible maximal reduction sequence (the only other possible one being the one where
the buyer chooses to pay by cash), we have that $f : \mathsf{Buyer}\complyO\mathsf{Seller}$.
\end{exa}

\begin{rem}
\label{rem:nondetorch}
As we shall see, strategy-orchestrators do correspond to univocal strategies for player $\playerC$ in $\game_{\rho\!\pp\!\sigma}$ and are technically easier to work with.
(In the proofs it will be easier to deal with partial functions rather than relations). On the other hand, we can recover
a full correspondence among unrestricted strategies for $\playerC$ and orchestrators by allowing disjunctions of affectible synchronization actions 
$\orchAct{\alpha}{\Dual{\alpha}}^+$. 
In a session-based scenario, however, we expect any nondeterminism to depend solely on 
either the client or the server. By allowing $f=\orchAct{\Dual{a}}{a}^+.f_1\vee\orchAct{\Dual{b}}{b}^+.f_2$ 
in the system ${a.\rho_1 + b.\rho_2 \pf{f} \Dual{a}.\sigma_1 + \Dual{b}.\sigma_2}$, the 
nondeterminism would depend on the orchestrator as well. 
\end{rem}

%

\section{Linking up all together: Main Results.}
\label{sec:affectible-retractable}

In this section we connect all the notions defined up to now.
We begin by showing that the relation $\ACRel$ can be axiomatized by means of the following formal system.
In such a system the symbol
$\complyF$ is the syntactic counterpart for the $\ACRel$ relation.

\begin{defi}[Formal system $\der$ for $\ACRel\!$]\label{def:formalCompl}
An environment $\Gamma$ is a finite set of expressions of the form $\delta\complyF\gamma$ where
$\delta,\gamma\in\ASC$. The judgments of System $\der$ are expressions of the form
$\Gamma\der\rho\complyF\sigma$. The axioms and rules of $\der$ are as in Figure
\ref{fig:systder}, where in rule $(\mbox{\footnotesize$+\cdot+$})$ we assume 
that a term of the form $a.\rho$ can be used instead of $a.\rho+\rho'$.
\end{defi}

\begin{figure}
\hrule
\vspace{2mm}
\label{fig:systder}
\[\begin{array}{c@{\hspace{8mm}}c}
(\mbox{\footnotesize$\CkptcomplAx$}):\Inf{}
	{\Gamma\der \stopA \complyF \sigma}
&
(\mbox{\footnotesize$\CkptcomplHyp$}):\Inf{}
{\Gamma, \rho\complyF\sigma \der \rho\complyF\sigma }
\end{array}\]
\[\begin{array}{c}
(\mbox{\footnotesize$+\cdot+$}):
\Inf{
     \Gamma, \alpha.\rho+\rho'\complyF\Dual{\alpha}.\sigma+\sigma'
     	\der \rho
    	  \complyF
		\sigma}
{ \Gamma\der \alpha.\rho+\rho'\complyF\Dual{\alpha}.\sigma+\sigma'}
\\[5mm]
(\mbox{\footnotesize$\oplus\cdot+$}):
\Inf{
    \forall i\in I.~ \Gamma,\mbox{\small $\bigoplus$}_{i\in I} \Dual{a}_i.{\rho}_i\complyF
    	\mbox{\small $\sum$}_{j\in I\cup J} a_j.{\sigma}_j\der 
    	\rho_i
    	\complyF
    		\sigma_i}
{ \Gamma\der \mbox{\small $\bigoplus$}_{i\in I} \Dual{a}_i.{\rho}_i\complyF
    	\mbox{\small $\sum$}_{j\in I\cup J} a_j.{\sigma}_j}
\\[5mm]
(\mbox{\footnotesize$+\cdot\oplus$}):
\Inf{
    \forall i\in I.~ \Gamma,\mbox{\small $\sum$}_{j\in I\cup J} a_j.{\sigma}_j\complyF \mbox{\small $\bigoplus$}_{i\in I} \Dual{a}_i.{\rho}_i \der 
    	\rho_i
    	\complyF
    		\sigma_i}
{ \Gamma\der\mbox{\small $\sum$}_{j\in I\cup J} a_j.{\sigma}_j\complyF \mbox{\small $\bigoplus$}_{i\in I} \Dual{a}_i.{\rho}_i    }	\\[6mm]
\end{array}\]
\caption{System $\der$}
\vspace{2mm}
\hrule
\end{figure}



We shall write $\der\rho\complyF \sigma$ as an abbreviation for $\emptyset\der\rho\complyF \sigma$. Moreover, we shall write $\Der :: \Gamma\der\rho\complyF \sigma$ to denote that
$\Der$ is a derivation having the judgment $\Gamma\der\rho\complyF \sigma$ as conclusion.
When clear from the context, we shall write simply $\Gamma\der\rho\complyF \sigma$ to
state that there exists a derivation $\Der$ such that $\Der :: \Gamma\der\rho\complyF \sigma$.

\begin{exa}
\label{ex:derivation}
Let us derive in system $\der$ the  judgment 
$\emptyset\der\mathsf{Buyer}\complyF\mathsf{Seller}$
\begin{equation}
\label{eq:exampleder}
\Inf[+,+]{
 \Inf[+,+]{
      \Inf[\oplus,+]
           {
             \Inf[\CkptcomplAx]
             {}
             {\Gamma'' \der \stopA \complyF \stopA}
\qquad
             \Inf[\CkptcomplAx]
             {}
             {\Gamma'' \der \stopA \complyF \stopA}
             }
            {\Gamma'  \der \DCard \oplus \DCash \complyF \Card + \Cash}
            }
            {\mathsf{Buyer}  \complyF \mathsf{Seller} \der \Price.(\DCard \oplus \DCash) \complyF \DPrice.(\Card + \Cash)}
         }{
     \der \mathsf{Buyer} \complyF \mathsf{Seller} 
           }
\end{equation}
\noindent
where
\begin{center}
\footnotesize
$\Gamma' = \mathsf{Buyer}  \complyF \mathsf{Seller},\ \Price.(\DCard \oplus
\DCash) \complyF \DPrice.(\Card + \Cash)$ \\ and
 $\Gamma''= \Gamma' ,\ \DCard \oplus \DCash \complyF \Card + \Cash$
\end{center}
\end{exa}
\medskip
We now prove the relevant properties of this system. First,
the proposed axiomatisation exactly captures $\ACRel$ as intended.

\subsection{Soundness and Completeness of $\der$ with respect to $\ACRel\!$}
In this subsection we prove that the system $\der$ is sound and complete with respect 
to the affectible compliance relation $\ACRel$, namely
$$\der\rho\complyF \sigma \hspace{1mm}\textit{ if and only if } \hspace{1mm}\rho\ACRel\sigma$$
Moreover, we shall get  decidability of $\ACRel$ as a corollary.

We begin by showing that the relation $\ACRel$ of Definition \ref{def:ACRel}  can be equivalently defined in a stratified way.

\begin{defi}\label{def:coACRel}\hfill
\begin{enumerate}
\item
\label{def:coACRel-i}
Let
$\FunK: \mathcal{P}(\ASC\times\ASC)\rightarrow  \mathcal{P}(\ASC\times\ASC)$ be such that, 
 for any   $\Rel\!\subseteq\! \ASC\!\times\!\ASC$, we get
$(\rho,\sigma) \in \FunK(\Rel)$ if either $\rho=\stopA$ or one of the following holds:
\begin{enumerate}
\item \label{lem:coinductiveChar-2} 
	$\rho = \sum_{i\in I}\alpha_i.\rho_i$, $\sigma = \sum_{j\in J}\Dual{\alpha}_j.\sigma_j$  and
	$\exists h \in I \cap J.\; ({\rho_h},{\sigma_h})\in\Rel$;
\item \label{lem:coinductiveChar-3} 
      $\rho = \bigoplus_{i\in I}\Dual{a}_i.\rho_i$, $\sigma = \sum_{j\in J}a_j.\sigma_j$,
	$I\subseteq J$ and $\forall h \in I. \; ({\rho_h},{\sigma_h})\in\Rel$;
\item \label{lem:coinductiveChar-4} 
       $\rho = \sum_{j\in J}a_j.\rho_j$, $\sigma = \bigoplus_{i\in I}\Dual{a}_i.\sigma_i$,
	$I\subseteq J$ and $\forall h \in I. \; ({\rho_h},{\sigma_h})\in\Rel$.
\end{enumerate}

\item 
A relation $\Rel\subseteq \ASC\times\ASC$ is an
{\em affectible compliance relation\/} if 
\hbox{$\Rel\subseteq \FunK(\!\Rel\!)$.}

\item
For any $n\in\Nat$ we define $\ACRelk{n}\subseteq\ASC \times \ASC$ as follows:\\
$\ACRelk{0} = \ASC \times \ASC$, whereas, for $n > 0$
$\ACRelk{n}=\FunK(\ACRelk{n-1})$

\item
We define $\ACRelco =\bigcap_{n\in\Nat} \ACRelk{n}$
\end{enumerate}
%
\end{defi}

\begin{fact}
\label{fact:relrelco}
$\ACRel = \ACRelco = \nu(\FunK)$.
\end{fact}

\noindent
Notice that $\ACRelk{k} \subseteq \ACRelk{k-1}$ for all $k$.
We define now a stratified notion of validity for judgments in system $\der$.

\begin{defi}[Stratified $\ACRel$\!-semantics for $\der$] \label{def:StratComplModel} 
Let $\Gamma$ be a set of statements of the form $\rho\complyF\sigma$ and let $k\in\Nat$. We define
\begin{enumerate}
\item
$ \begin{array}{rcl}
\modelsACRel_k \Gamma &\mbox{ if }& \forall {(\rho' \complyF \sigma') \in \Gamma} .
 \end{array} $ \hspace{-2mm}
$\rho' \ACRelk{k}\sigma'$\,;
\item 
$ \begin{array}{rcl}
\Gamma\modelsACRel_k \rho \complyF \sigma & \mbox{ if } & \modelsACRel_k \Gamma \implies \rho\ACRelk{k}\sigma
 \end{array}$.
\end{enumerate}
\end{defi}

We can now proceed by proving the soundness of $\der$.


\begin{figure}[t]
\hrule
\vspace{2mm}
\centering
{\small
\begin{tabbing}
\Prove\=$(\Gamma\der \rho\complyF\sigma)$ \\ [1mm]
\IF \> $\rho = \stopA$~~ \THEN ~~~~~
	$(\mbox{\scriptsize$\CkptcomplAx$}):\Inf{}{\Gamma\der \stopA \complyF \sigma}$\\ [1mm]
\ELSE \> \IF \= ~~~$\rho\complyF\sigma \in \Gamma$ ~~\THEN~~~~~
	$(\mbox{\scriptsize$\CkptcomplHyp$}):\Inf{}{\Gamma, \rho\complyF\sigma \der \rho\complyF\sigma}$ \\ [2mm]

\ELSE \> \IF \> ~~\=~~ \= $\rho = \bigoplus_{i\in I}\Dual{a}_i.\rho_i$ ~\= \AND~ $\sigma = \sum_{j\in J}a_j.\sigma_j$~\AND~ $I\subseteq J$\\[2mm]
\>\>\> \AND~ \FA~$k\in I$ ~~$\mathcal{ D}_k = $~\Prove$(\Gamma, \rho\complyF\sigma \der \rho_k\complyF\sigma_k) \neq \FAIL$\\ [2mm]
\>  \THEN ~~~~~
	$(\mbox{\scriptsize$\oplus\cdot+$}):\Inf{
    		(\forall k\in I) ~~\mathcal{ D}_k\qquad }{
    		\Gamma\der \rho\complyF\sigma}$ \\ [1mm]
\ELSE \> \IF \> ~~~$\rho = \sum_{i\in I}a_i.\rho_i$ ~\AND~ $\sigma = \bigoplus_{j\in J}\Dual{a}_j.\sigma_j$~\AND~ $I\supseteq J$\\ [2mm]
\>\>\> \AND~~\FA~$k\in J$~~~$\mathcal{ D}_k = $~\Prove$(\Gamma, \rho\complyF\sigma \der \rho_k\complyF\sigma_k)\neq \FAIL$\\ [2mm]
\>  \THEN ~~~~~
	$(\mbox{\scriptsize$+\cdot\oplus$}):\Inf{
    		(\forall k\in J) ~~\mathcal{ D}_k\qquad}{
    		\Gamma\der \rho\complyF\sigma}$  ~~~~\ELSE ~~\FAIL\\ [1mm]

\ELSE \> \IF ~~~~  $\rho = \sum_{i\in I}\alpha_i.\rho_i$ ~ \AND~$\sigma = \sum_{j\in J}\Dual{\alpha}_j.\sigma_j$ \\[2mm]
\> \> \> \AND~~~{\bf exists} $k \in I\cap J$ {\bf s.t.} $\mathcal{ D} = $ \Prove$(\Gamma, \rho\complyF\sigma \der \rho_k\complyF\sigma_k) \neq \FAIL$ \\ [2mm]
\>  \THEN  ~~~~~
	$(\mbox{\scriptsize$+\cdot+$}):\Inf{\Der}
	{\Gamma\der \rho\complyF\sigma}$ ~~~~\ELSE ~~\FAIL\\ [2mm]

\ELSE \> \FAIL
\end{tabbing}
}
\caption{The procedure\;\Prove.}\label{fig:Prove}
\vspace{2mm}
\hrule
\end{figure}

\begin{prop}[Soundness of $\der$]
\label{prop:soundnessACRel}
If \,$\der\rho\complyF \sigma$, then \,$\modelsACRel\rho\complyF\sigma$.
\end{prop}
\proof
Actually we show a stronger property, namely that
\lmcscenterline{
\,$\Gamma\der\rho\complyF \sigma$ implies\,$\Gamma\modelsACRel\rho\complyF\sigma$.
}

\noindent
By Fact \ref{fact:relrelco},  it is enough to show that 
\[ \begin{array}{rcl}
\Gamma\der \rho \complyF \sigma &\text{implies}& \Gamma\modelsACRel_k \rho \complyF \sigma \textrm{ for all }k. 
\end{array} \]

We proceed by simultaneous induction over the derivation $\Deriv :: \Gamma\der \rho \complyF \sigma$ and over $k$. 
Since 
$\Gamma\modelsACRel_0\rho \complyF \sigma$ trivially holds, we shall keep the case $k = 0$ implicit. Let $k > 0$;
we distinguish the possible cases of the last rule in $\Deriv$.

\begin{description}

\item [Case $(\TcomplAx)$] 
Then $\Deriv$ consists of the inference 
 \[
\Inf[\TcomplAx] { \Gamma\der\stopA \complyF \sigma }
 \]
and the thesis is immediate since $\stopA \ACRelk{k} \sigma$;

\item [Case $(\TcomplHyp)$] 
Then $\Deriv$ consists of the inference: 
 \[ 
\Inf[\TcomplHyp] { \Gamma, \rho \complyF \sigma \der \rho \complyF \sigma }
 \]
and $\Gamma, \rho \complyF \sigma \modelsACRel_k \rho \complyF \sigma$ trivially holds.

\item [Case $(+\cdot +)$] Then $\Der$ ends by 
 \[ \Inf[+\cdot +] { \Gamma, \alpha.\rho+\rho'\complyF\Dual{\alpha}.\sigma+\sigma' \der \rho \complyF \sigma }{ \Gamma\der \alpha.\rho+\rho'\complyF\Dual{\alpha}.\sigma+\sigma'} 
 \]
 If $\modelsACRel_k\Gamma$ then $\modelsACRel_{k{-}1} \Gamma$; by induction over $k$ we know that
 $\Gamma\modelsACRel_{k{-}1} \alpha.\rho+\rho'\complyF\Dual{\alpha}.\sigma+\sigma'$ that implies that
 $  \alpha.\rho+\rho'\ACRelk{k{-}1} \Dual{\alpha}.\sigma+\sigma'$ and hence 
	$\modelsACRel_{k{-}1} \Gamma, \alpha.\rho+\rho'\complyF\Dual{\alpha}.\sigma+\sigma'$. From this, by induction over $\Der$, 
we
get $\rho \ACRelk{k-1} \sigma$; by definition of $\ACRelk{k}$ this implies
$\alpha.\rho+\rho'  \modelsACRel_{k} \Dual{\alpha}.\sigma+\sigma'$ and we conclude that 
$\Gamma\modelsACRel_{k} \alpha.\rho+\rho'\complyF\Dual{\alpha}.\sigma+\sigma'$ as desired.
\item [Cases $(+\cdot\oplus)$  and $(+\cdot\oplus)$] Similar to case $(+\cdot +)$.
\qed
\end{description}

To get decidability and to prove the completeness property of system $\der$, we study the proof-search procedure \Prove\ defined in
Figure \ref{fig:Prove}. The procedure  \Prove, given a judgment $\Gamma\der\rho\complyF \sigma$, attempts to reconstruct a derivation of it in system $\der$. Such a procedure is correct and terminating: it either returns a derivation, if any, or it fails, in case the judgment is not derivable in the system.

\begin{lem}\label{prop:provecorr}
The proof search algorithm $\mathbf{Prove}$\ for $\der$ is correct and terminating. In particular,
\begin{enumerate}
\item 
\label{prop:provecorr-i}
{\em \Prove}$(\Gamma\der \rho\complyF\sigma)=\Der \neq$ {\em \FAIL}~~~~implies~~~~$\Der::\Gamma\der \rho\complyF\sigma$;
\item 
\label{prop:provecorr-i-bis}
{\em \Prove}$(\Gamma\der \rho\complyF\sigma)=$ {\em \FAIL}~~~~implies~~~~$\Gamma \not\modelsACRel_k \rho\complyF\sigma$ 
for some $k$;
\item 
\label{prop:provecorr-ii}
{\em \Prove}$(\Gamma\der \rho\complyF\sigma)$ terminates for all judgments $\Gamma\der \rho\complyF\sigma$.
\end{enumerate}
\end{lem}

\proof \leavevmode
\begin{enumerate}
\item Immediate by construction of \Prove.

\item Let $\Prove(\Gamma\der \rho\complyF\sigma)=\FAIL$. Then the procedure terminates: let $h$ be the 
number of nested calls of \Prove\ in this execution. We claim that $\Gamma \not\modelsACRel_{h+1} \rho\complyF\sigma$ which we prove by induction over $h$. If $h = 0$ then $ \rho\complyF\sigma \not \in \Gamma$ and none of the conditions defining $\ACRelk{1}$ is satisfied.
If $h > 0$ then again $\rho\complyF\sigma \not \in \Gamma$ and also $\rho\neq \stopA$. As $h > 0$ there is at least one recursive call
of \Prove\ and hence either $\rho = \sum_{i\in I}\alpha_i.\rho_i$, $\sigma = \sum_{j\in J}\Dual{\alpha}_j.\sigma_j$ and $I\cap J \neq \emptyset$
or $\rho = \bigoplus_{i\in I}\Dual{\alpha}_i.\rho_i$, $\sigma = \sum_{j\in J}\alpha_j.\sigma_j$ and
$I\subseteq J$ or $\rho = \sum_{i\in I}\Dual{\alpha}_i.\rho_i$, $\sigma = \bigoplus_{j\in J}\alpha_j.\sigma_j$ and 
$I\supseteq J$. In the first case the hypothesis that $\Prove(\Gamma\der \rho\complyF\sigma)=\FAIL$ implies that
for all $i \in I \cap J$, $\Prove(\Gamma\der \rho_i\complyF\sigma_i)=\FAIL$; but the number of recursive calls in any call of
$\Prove(\Gamma\der \rho_i\complyF\sigma_i)$ will be less than $h$, hence by induction there exists some $l < h$ such that
$\Gamma \not\modelsACRel_{l+1} \rho_i\complyF\sigma_i$, that implies
$\Gamma \not\modelsACRel_{h} \rho_i\complyF\sigma_i$ for all $i \in I \cap J$ since $\ACRelk{h}\subseteq \ACRelk{l+1}$. It follows that
$\Gamma \not\modelsACRel_{h+1} \rho\complyF\sigma$ by clause (2) in the definition of $\ACRelk{h+1}$ that is the only one that applies.
The other cases of $\rho$ and $\sigma$ are treated similarly.

\item
Notice that
in all recursive calls \Prove$(\Gamma, \rho\complyF\sigma \!\der\! \rho_k\complyF\sigma_k)$
inside
\Prove$(\Gamma\der \rho\complyF\sigma)$ the expressions $\rho_k$ and  $\sigma_k$ are subexpressions of $\rho$ and $\sigma$
respectively.  Since contract expressions generate regular trees, there are only finitely many such subexpressions; therefore the if clause 
$\rho\complyF\sigma \in \Gamma$ corresponding to axiom ($\CkptcomplHyp$) cannot fail infinitely many times.
This implies that the number of nested calls of procedure \Prove\ is always finite. 
\qed
\end{enumerate}
Decidability now immemdiately descends as a corollary.

\begin{cor}\label{cor:ACRelDecidable}
The relation $\ACRel$ is decidable.
\end{cor}
\noindent
The previous lemma also enables us to get the completeness property.

\begin{prop}[Completeness  w.r.t $\modelsACRel$] \label{thr:completeness} If \hspace{0.5mm}
${\modelsACRel}\rho\complyF\sigma $, then $\der \rho\complyF\sigma.$
\end{prop}
\proof
Now suppose that $\rho \ACRel \sigma$; by \ref{prop:provecorr}.(\ref{prop:provecorr-ii}) the computation of $\Prove(\emptyset \der \rho\complyF \sigma)$ terminates; the value cannot be \FAIL\ by  \ref{prop:provecorr}.(\ref{prop:provecorr-i-bis}), since 
$\rho \ACRelk{k} \sigma$ for all $k$ by hypothesis, hence the algorithm \Prove\ yields a derivation $\Der$ of $\emptyset \der \rho\complyF \sigma$ by
\ref{prop:provecorr}.(\ref{prop:provecorr-i}). \qed

%
%
%
%

%


%
%
%
%

\subsection{Characterizations of  $\ACRel$}

The relation $\rho \ACRel \sigma$ can be characterized in terms of the existence of a winning strategy for player $\playerC$ in the game $\game_{\rho\!\pp\!\sigma}$, a condition that in turn is equivalent to $\rho$ and $\sigma$ being retractable compliant as well as being orchestrated compliant. This is the content of the next theorem, whose proof establishes a tight correspondence among strategies and orchestrators.

\begin{thm}[Main Theorem I]\hfill\\
\label{th:complyAwinstrat}
Let  $\rho,\sigma\in\ASC$,
The following conditions are equivalent:
\begin{enumerate}[label=(\arabic*)]
\item
\label{th:complyAwinstrat-1}
$\rho\ACRel \sigma$
\item
\label{th:complyAwinstrat-2}
$\rho\complyR \sigma$
\item
\label{th:complyAwinstrat-3}
There exists a winning strategy for player $\playerC$ in $\game_{\rho\pp\sigma}$.
\item
\label{th:complyAwinstrat-4}
There exists an orchestrator $f$ such that ~$f: \rho\complyO \sigma$.
\end{enumerate}
\end{thm}

By the above theorem, soundness and completeness of system $\der$, as well as the decidability property, immediately transfers from  $\ACRel$  to both $\complyR$ and $\complyO$.
This might look a bit weird at a first sight, since the  judgements of system $\der$ abstract away  from both histories and orchestrators, which are essential, respectively, for the definition of rollback (and hence of retractable compliance) and for the definition of orchestrated  compliance. 
However, much as it happens with logic, ``proofs'', namely derivations, can be interpreted as strategies in games determined by their conclusion;
on the other hand the informative contents of a derivation lies in the choice of actions and co-actions involved in the
interaction among a client and a server, which is exactly the effect of an orchestrator.\\

The proof of Theorem \ref{th:complyAwinstrat} will be developed in Appendix \ref{appendix:mainthmI} by proving the following equivalences:
\[\begin{array}{ccccc}
                                      &\mbox{\footnotesize \bf\ref{subsect:rbkiffaffect}} &        &  \mbox{\footnotesize \bf \ref{subsec:tbstratequiv}}\\[-2mm]
\ref{th:complyAwinstrat-2} & \Iff & \ref{th:complyAwinstrat-1} & \Iff & \ref{th:complyAwinstrat-3}\\[1mm]
                                       &   & \Updownarrow & \hspace{-8mm}\mbox{\footnotesize \bf\ref{subsec:affectorchequiv}}\\[1mm]
                                       &  & \ref{th:complyAwinstrat-4}
\end{array}\]

The proofs of such equivalences roughly follow  the following schemas.

\begin{description}

\item[$\ref{th:complyAwinstrat-1} \Iff \ref{th:complyAwinstrat-2}$]
The relation $\complyR$ is completely characterized by the properties
defining the relation $\ACRel$, using in an essential way Lemma \ref{lem:stack-len}.

\item[$\ref{th:complyAwinstrat-1} \Iff \ref{th:complyAwinstrat-3}$]
Since $\ACRel\ =\ \complyTB$ by Theorem \ref{th:complyequivtbcomply},
it is enough to show that 
\lmcscenterline{\hspace{-6mm}
$\rho \complyTB \sigma \mbox{\ \em if and only if \
there exists a winning strategy for player $\playerC$ in $\game_{\rho\pp\sigma}$}$
\hspace{4mm}} 
This is proved by using a characterization of  $\complyTB$ in terms of regular trees
without ``synchronization-failure'' leaves. A tree of this sort can be obtained out of a winning strategy, and vice versa.

\item[$\ref{th:complyAwinstrat-1} \Iff \ref{th:complyAwinstrat-4}$]
We provide a formal system $\derinfOrch$ that is sound and complete with respect to the 
$\complyO$ relation. Then we define a procedure that, given a derivation $\Der ::\ \der \rho\complyF\sigma$, returns a derivation $\Der'::\ \derinfOrch \fofder{\Der}: \rho\complyOF\sigma$,
where  $\fofder{\cdot}$ is a map from derivation to orchestrators, simultaneously defined together with
the first procedure.
\end{description}

\subsection{Getting strategies, derivations and orchestrators out of each other}
Strategies, derivations and orchestrators mentioned in the Main Theorem I \ref{th:complyAwinstrat} can
effectively be computed out of each other as
stated in the following theorem.


\begin{thm}[Main Theorem II]
\label{th:derstratorchequiv}\hfill
\begin{enumerate}[label=(\arabic*)]
\item
\label{th:derstratorchequiv-i}
Given a derivation $\Der$ for $\der\rho\complyF \sigma$,
an orchestrator $f_\Der$ can be computed out of $\Der$, such that $f_\Der: \rho\complyO \sigma$;
\item
\label{th:derstratorchequiv-ii}
Given an orchestrator $f$ such that $f: \rho\complyO \sigma$, a strategy $\Sigma_f$ which is winning for player $\playerC$ in game $\game_{\rho\pp\sigma}$ can be effectively obtained
out of $f$;
\item
\label{th:derstratorchequiv-iii}
Given a winning strategy $\Sigma$ for player $\playerC$ in the game $\game_{\rho\pp\sigma}$,
an orchestrator $f_\Sigma$ such that $f_\Sigma: \rho\complyO \sigma$ can be effectively obtained out of the strategy;
\item
\label{th:derstratorchequiv-iv}
Given an orchestrator $f$ such that $f: \rho\complyO \sigma$, a derivation $\Der_f$ for $\der\rho\complyF \sigma$ can be effectively obtained out of $f$.
\end{enumerate}
\end{thm}

\noindent
The proof of Main Theorem II is in Appendix \ref{appendix:mainthmII} and the related constructions are provided along the following lines:

\begin{description}
\item[\ref{th:derstratorchequiv-i} \ref{subsect:derstrat-i}] We use the function $\fofder{\cdot}$ from derivations to orchestrators defined in the proof of $\ref{th:complyAwinstrat-1} \Iff \ref{th:complyAwinstrat-4}$ of Theorem \ref{th:complyAwinstrat}.

\item[\ref{th:derstratorchequiv-ii} \ref{subsect:derstrat-ii}] A ``turn-based'' version ($\complyTBO$) of the relation
$\complyO$ 
is provided at the beginning of Appendix \ref{appendix:mainthmII}. Given an orchestrator $f$ such that $f: \rho\complyTBO \sigma$ 
(and hence $f: \rho\complyO \sigma$) we define a procedure yielding a suitable regular tree which can be decorated in order obtain
a winning strategy $\Sigma_f$
for player $\playerC$ in $\game_{\rho\!\pp\! \sigma}$.

\item[\ref{th:derstratorchequiv-iii} \ref{subsect:derstrat-iii}]
We use a construction defined in the proof of $\ref{th:complyAwinstrat-1} \Iff \ref{th:complyAwinstrat-3}$ of Theorem \ref{th:complyAwinstrat}. 
Given a
univocal winning strategy $\Sigma$ for player $\playerC$ in the game $\game_{\rho\!\pp\!\sigma}$ we obtain
a tree $\mathsf{T}$ representing all the possible plays of the game where player $\playerC$ follows the
stategy $\Sigma$. Then an orchestrator $f_\Sigma$ is obtained out of 
$\mathsf{T}$, such that $f_\Sigma: \rho\complyTBO \sigma$ (and hence $f_\Sigma: \rho\complyO \sigma$).

\item[\ref{th:derstratorchequiv-iv} \ref{subsec:orchtoder}] We define a procedure $\OrchToDer$ that, given $f$, $\rho$ and $\sigma$ such that $f: \rho\complyO \sigma$, returns a derivation $\Der_f::\ \der \rho\complyF \sigma$. The procedure
$\OrchToDer$ is obtained by adaptating the  proof search procedure $\Prove$. In particular, in $\OrchToDer$ the search is
driven by the orchestrator $f$. Correctness and termination of $\OrchToDer$ are proved as for $\Prove$.
\end{description}

\begin{exa}
\label{ex:startoutorch}
Let $\Der$ be the derivation in Example \ref{ex:derivation}. If we consider the function $\fofder$ from derivation to orchestrators mentioned before (and formally defined in Appendix \ref{subsec:affectorchequiv}, Definition \ref{def:derinforch}), we have that:
$$\fofder{\Der} = \orchAct{\Bag}{\DBag}^+.\orchAct{\DPrice}{\Price}.(\orchAct{\Card}{\DCard}.\stopf\vee\orchAct{\Cash}{\DCash}.\stopf)
\footnote{Actually the application of $\mathsf{f}$ to the derivation $\Der$ does produce some vacuous $\rec$ binders. We omit them here for sake of readability.}$$

\noindent
If we dub $f= \fofder{\Der}$, we have that the strategy $\Sigma_f$, obtained by the construction in the proof of
\ref{th:derstratorchequiv-ii}, is such that
$$\Sigma_f(\vec{e})= \left\{ \begin{array}{l@{~~~}l}
       \Set{(1,(\playerC{:}\Bag))} & \mbox{ if } \vec{e}=\seq{}\\
       \Set{(6,(\playerC,\cmark))} & \mbox{ if } \vec{e}=\vec{s_3}\\
        \emptyset                       & \mbox{ for any other play }
          \end{array} \right.
$$
where $\vec{s_3} =$ {\small  $(1,(\playerC{:}\Bag))(2,(\playerB{:}\DPrice))(3,(\playerA{:}\Price))
(4,(\playerA{:}\DCash))(5,(\playerB{:}\Cash))$}.
Observe that $\Sigma_f$ corresponds to the strategy $\tilde\Sigma$ as defined in Example \ref{ex:gamestrat}.
The construction of \ref{th:derstratorchequiv-iii}, instead, yields $f$ out of $\Sigma_f$,
whereas $\mathbf{O2D}(f,\mathsf{Buyer},\mathsf{Seller}) = \Der$.
\end{exa}

\subsection{Orchestrator synthesis}
Working on Proposition \ref{prop:provecorr} and Theorem \ref{th:derstratorchequiv}\ref{th:derstratorchequiv-i} we obtain a synthesis algorithm $\Synth$ that is defined in Figure \ref{fig:algSynth}.
The algorithm  $\Synth$
takes a (initially empty) set of assumptions $\Gamma$ and two affectible contracts $\rho$ and $\sigma$, and returns a set $O$ of  orchestrators
(and hence a set of strategies by the above results), if any, such that for any $f\in O$ we have $f: \rho\complyO \sigma$; the algorithm returns the empty set in case no such an orchestrator exists.
In the algorithm $\Synth$ we consider orchestrators as explicit terms, that is we do not consider recursion up-to recursion unfolding.\\

\begin{exa}
It is not difficult to check that by computing $\mathbf{Synth}(\emptyset, \mathsf{Buyer},\mathsf{Seller})$ we get a set just consisting of 
exactly the orchestrator $f$ of Example \ref{ex:orchderseq}, which we have
shown to be such that $f:\mathsf{Buyer}\complyP\mathsf{Seller}$:
$$
\mathbf{Synth}(\emptyset,\! \mathsf{Buyer},\!\mathsf{Seller})\!=\!\big\{ \orchAct{\Bag}{\!\DBag}^+.\orchAct{\DPrice}{\!\Price}(\orchAct{\Card}{\DCard}.\stopf\vee\orchAct{\Cash}{\DCash}.\stopf)\big\}
$$
\end{exa}
\begin{figure}[t]
\label{fig:algSynth}
\hrule
\vspace{2mm}
{\small
\begin{tabbing}
\Synth\=$(\Gamma, \rho,\sigma)$ = \+ \\ [2mm]
\IF\ $x:\rho\complyOF\sigma \in \Gamma$ \THEN\ $\Set{x}$ 
	\\ [2mm]
\ELSE\; \= \IF\ $\rho = \stopA$ \THEN\ $\Set{\stopA}$ 
	\\ [2mm]

\ELSE \> \IF\ \= $\rho = \bigoplus_{i\in I} \Dual{a}_i\Actdot {\rho}_i $ \AND\ $\sigma = \sum_{j\in I\cup J} a_j\Actdot {\sigma}_j$ 
			 \THEN\ \\ 
	 \> \LET\ \= $\Gamma' = \Gamma,\; x{:}\rho \complyOF \sigma$ ~
\ \IN \\ 
\>\> $\Set{ \rec x \procdot  \bigvee_{i\in I} \orchAct{a_i}{\Dual{a}_i}.f_i \mid  \forall i\in I. f_i\in \Synth(\Gamma', \rho_i, \sigma_i}$ \\ [2mm]

\ELSE \> \IF\ \= $\rho = \sum_{j\in I\cup J} \Dual{a}_j\Actdot {\rho}_i $ \AND\ $\sigma = \bigoplus_{i\in I} a_i\Actdot {\sigma}_i$ 
			 \THEN\ \\ 
	 \> \LET\ \= $\Gamma' = \Gamma,\;x{:}\rho \complyOF \sigma_i$ \IN \\ 
\>\>$\Set{ \rec x \procdot  \bigvee_{i\in I} \orchAct{a_i}{\Dual{a}_i}.f_i \mid \forall i\in I. f_i\in \Synth(\Gamma', \rho_i, \sigma_i}$ \\ [2mm]

\ELSE \> \IF\ \= $\rho = \sum_{i\in I} \Dual{\alpha}_i\Actdot {\rho}_i $ \AND\ $\sigma = \sum_{j\in J} \alpha_j\Actdot {\sigma}_j $ (where $\alpha\in\Names\cup\CoNames$) \AND\ $|I|\geq 2$ \THEN \\ 
	\> \LET\ \= $\Gamma' = \Gamma,\; x{:}\rho\complyOF \sigma$ \IN \\ 
	\>\> $\bigcup_{i\in I}\Set{\rec x \procdot \orchAct{\alpha_i}{\Dual{\alpha}_i}^+\Actdot f \mid f\in \Synth\,(\Gamma',{\rho}_i ,\sigma_i) } $ 
	\\ [2mm]

\ELSE \> $\emptyset$
\end{tabbing}
\vspace{-6mm}
}\caption{The algorithm \Synth.}\label{fig:Synth}
\vspace{2mm}
\hrule
\end{figure}
\noindent
The algorithm $\Synth$ can be proved to be terminating.

\begin{prop}[Termination of \Synth]
\label{th:soundcomplSynth-i}
\label{lem:terminationSynth}
\lmcscenterline{$\mathbf{Synth}(\Gamma, \rho, \sigma)$ terminates for any $\Gamma,\rho$ and $\sigma$.}
\end{prop}
\proof
%
%
All session contracts in the recursive calls of  \Synth\  are sub-expressions of either $\rho$ or $\sigma$ or of a session contract in a judgement in $\Gamma$ (which is finite). Since session contracts are regular trees, their sub-expressions are a finite set, so that the test \mbox{$x : \rho\complyOF\sigma \in \Gamma$} in the first clause of  \Synth\  is always successfully reached in case the algorithm does not terminate because of the last clause. \qed

The algorithm $\Synth$ can be proved to be correct and complete in the following sense:
whenever it does not fail, it does return a set of correct orchestrators.
Moreover, if an orchestrator exists for given $\rho$ and $\sigma$, it is actually ''captured'' by our synthesis algorithm.\\

\noindent
Given an orchestrator
$f$ we denote by $\regtree{f}$ its corresponding (possibly infinite) regular tree.

\begin{prop}[Correctness and Completeness of \Synth]
\label{th:soundcomplSynth} \leavevmode
\begin{enumerate}[beginpenalty=99,label=\arabic*)]
\item
\label{fact:PdsFmSynthExt1}
If $f\in \mathbf{Synth}(\emptyset, \rho, \sigma)\neq\emptyset$~then ~$f: \rho\complyO \sigma$.
\item
\label{fact:PdsFmSynthExt}
If $f: \rho\complyO \sigma$~ then~ there exists $g\in\mathbf{Synth}(\emptyset, \rho, \sigma)\neq\emptyset$ such that with $\regtree{f}=\regtree{g}$.
\end{enumerate}
\end{prop}
\proof
See Appendix \ref{appendix:synth}.
\qed

%

\section{Subcontract relation: Definition and Main Results.}
\label{sec:subcontractrel}

The notion of compliance naturally induces a substitutability relation on servers that may be used for implementing contract-based query engines (see \cite{Padovani10} for a discussion).

\begin{defi}[Affectible subcontract relation] \label{def:osubcontr}
Let $\sigma, \sigma' \in \ASC$ . We define
 \[ \sigma \subcontr \sigma' ~~\ByDef~~ \forall \rho \Pred[ \rho\ACRel\sigma \implies \rho \ACRel \sigma' ]
\]
\end{defi}

\begin{exa}
Consider the following new version of $\mathsf{Seller}$, that also accepts cheques as payment for the bag and enables 
customers (may be those with a fidelity card or those who make shopping on Christmas eve) to win the bag by means of a scratch card.
\begin{tabbing}
$\mathsf{SellerII} =$ \= $\Belt. \DPrice.\Cash $\\
                             \> $+$ \\
                             \> $\Bag.$ \= ($\DPrice.(\Card + \Cash + \Cheque$)\\
                                          \> \> $+$\\
                                           \>\> $\DScrachcard$)
\end{tabbing}
It turns out that $\mathsf{Seller} \subcontr \mathsf{SellerII}$, so that in particular  
$\mathsf{Buyer}\, \ACRel\,\mathsf{SellerII}$ holds.
\end{exa}

As done for several notions of compliance for session contracts, decidability of the subcontract relation could be obtained as an immediate consequence of decidability of $\ACRel$ if we managed to have a proper notion of dual contract and if the following property could be proved:
\begin{equation}
\label{eq:nonpropertydual}
\sigma\subcontr\sigma' \ \Iff\  \Dual{\sigma}\ACRel\sigma'
\end{equation}

However, as already discussed in Remark \ref{rem:noduality}, in the present setting the notion of duality is hardly definable
so that we have no chance to get (\ref{eq:nonpropertydual}).

Nonetheless decidability of $\subcontr$ can be obtained in a direct way by means of a formal system
axiomatising the subcontract relation
and of a proof-search algorithm in the style of $\Prove$.

\subsection{A sound and complete formal system for $\subcontr$}

\begin{defi}[The Formal System $\dersc$ for $\subcontr$]\label{def:formalSubcontr}
A judgment in the formal system $\dersc$ is an expression of the form
$\Gamma\dersc \rho \subcontrF \sigma$ where 
$\Gamma$ is a finite set of expressions with the form  $\delta \subcontrF \gamma$, with
$\rho,\sigma,\delta,\gamma\in\Sbehav$.
Axioms and inference rules of $\dersc$ are as in Figure \ref{fig:forsystsubcontr}, where
the following provisos hold:
\begin{itemize}
\item in rule $(\mbox{\footnotesize $ \oplus\cdot+\mbox{\,-\!}\subcontrF $})$ we
assume
that a term of the form $\Dual{a}.\sigma_1$ can be used instead of $\Dual{a}.\sigma_1+\sigma'$;
\item  in rule $(\mbox{\footnotesize $ \oplus\cdot+\mbox{\,-\!}\subcontrF $})$ we
assume either $\mbox{\small $\sum$}_{i\in I} \alpha_i.{\sigma}_i$ or $\mbox{\small $\sum$}_{j\in I\cup J} \alpha_j.\sigma'_j$ (not both) can be of the form $\Dual{a}.\sigma$\footnote{This conditions are needed in order to let the system to be syntax-directed.}.
\end{itemize}
\end{defi}
\noindent
In system $\dersc$\!, the symbol $\subcontrF$ is used as syntactical counterpart of the relation $\subcontr$.

The ideas behind rules ($+\cdot+\mbox{\,-\!}\subcontr$) and ($\oplus\cdot\oplus\mbox{\,-\!}\subcontr$) are fairly intuitive.
Let us informally see why rule ($\oplus\cdot+\mbox{\,-\!}\subcontr$) can relate affectible and unaffectible outputs,  unlike what happens for other subcontract relations for session contracts.
We first observe that a contract $\rho$ which is compliant with a term of the form 
$\mbox{\small $\bigoplus$}_{h\in H} \Dual{a}_h.\sigma_h$
must be such that $\rho = \mbox{\small $\sum$}_{h\in H'\supseteq H} a_h.\rho_i$
with $\rho_h\complyP \sigma_h$ for any $h\in H$.
A term $\sigma'$ different from $\mbox{\small $\bigoplus$}_{h\in I} \Dual{a}_h.\sigma_i$
and such that $\rho\complyP \sigma'$ can be either of the form
$\mbox{\small $\bigoplus$}_{k\in K} \Dual{a}_k.\sigma_k$,  and this case is dealt
with by  rule $(\mbox{\footnotesize $ \oplus\cdot\oplus\mbox{\,-\!}\subcontrF $})$; or of the form
$\mbox{\small $\sum$}_{k\in K} \Dual{a}_k.\sigma'_k$. Notice that
in order to have $\mbox{\small $\sum$}_{h\in H'} a_h.\rho_i
\complyP \mbox{\small $\sum$}_{k\in K} \Dual{a}_k.\sigma''_i$ it is enough that there exists
$p\in H'\cap K$ such that $\rho_p\complyP \sigma'_p$.
 This is precisely what is guaranteed  by the premise of rule ($+\cdot+\mbox{\,-\!}\subcontr$).



We can prove system $\dersc$ to be sound and complete for the subcontract relation $\subcontr$. \\
\begin{figure}
\hrule
\vspace{2mm}
\[\begin{array}{c@{\hspace{8mm}}c}
(\mbox{\tiny $\CkptcomplAx \mbox{\,-\!}\subcontrF$}):\Inf{}{\Gamma\dersc \stopA \subcontrF \sigma'}
&
(\mbox{\tiny $ \CkptcomplHyp\mbox{\,-\!}\subcontrF $}):\Inf{\Gamma, \sigma\subcontrF\sigma' \dersc \sigma\subcontrF\sigma'}\\[8mm]
\multicolumn{2}{c}{
(\mbox{\footnotesize $ \oplus\cdot+\mbox{\,-\!}\subcontrF $}):\Inf{\Gamma, \Dual{a}.\sigma_1\oplus\sigma_2\subcontrF\Dual{a}.\sigma_1'+\sigma_2'
     	\dersc \sigma_1
    	  \subcontrF \sigma_1'}{\Gamma\dersc \Dual{a}.\sigma_1\oplus\sigma_2\subcontrF\Dual{a}.\sigma_1'+\sigma_2'}
}\\[8mm]
\multicolumn{2}{c}{
(\mbox{\footnotesize $ +\cdot+\mbox{\,-\!}\subcontrF $}):\Inf{\forall h\in I.~~ \Gamma,\mbox{\small $\sum$}_{i\in I} \alpha_i.{\sigma}_i\subcontrF
    	\mbox{\small $\sum$}_{j\in I\cup J} \alpha_j.\sigma'_j\dersc 
    	\sigma_h
    	\subcontrF
    		\sigma'_h}{\Gamma\dersc \mbox{\small $\sum$}_{i\in I} \alpha_i.{\sigma}_i\subcontrF
    	\mbox{\small $\sum$}_{j\in I\cup J} \alpha_j.\sigma'_j}
}
\\[8mm]
\multicolumn{2}{c}{
(\mbox{\footnotesize $ \oplus\cdot\oplus\mbox{\,-\!}\subcontrF $}):\Inf
{\forall h\in I.~~ \Gamma,\mbox{\small $\bigoplus$}_{j\in I\cup J}\Dual{a}_j.{\sigma}_j\subcontrF \mbox{\small $\bigoplus$}_{i\in I} \Dual{a}_i.\sigma'_i \dersc 
    	\sigma_h
    	\subcontrF
    		\sigma'_h}{\Gamma\dersc\mbox{\small $\bigoplus$}_{j\in I\cup J} \Dual{a}_j.\sigma'_j\subcontrF \mbox{\small $\bigoplus$}_{i\in I} \Dual{a}_i.\sigma'_i}
}
\end{array}\]
\caption{The formal system $\dersc$}\label{fig:forsystsubcontr}
\vspace{2mm}
\hrule
\end{figure}
Proof search termination for system $\dersc$ can be shown in the same way as done for $\der$.
\begin{prop}[Proof search termination]
\label{prop:proofsearchtermsubcontr}
For system ${\dersc}\!$, proof search does terminate.
\end{prop}
\medskip
We can now proceed with the soundness and completeness properties for $\dersc$ with respect to
the relation $\subcontr$.

\noindent
We begin by defining a non-involutive ``quasi-dual'' operator on affectible contracts,
that we shall use to build counterexamples in the proof of Proposition \ref{prop:subcontrproperties} below.

\begin{defi}[An operator of quasi-duality]
The operator $\QDual{\ \cdot\ }:\ASC\rightarrow\ASC$ is inductively defined as follows.
$$\begin{array}{rcl}
\QDual{\stopA} & = & \stopA\\
\QDual{\mbox{\small $\bigoplus$}_{i\in I} \Dual{a}_i.\sigma_i} & = & \mbox{\small $\sum$}_{i\in I} a_i.\QDual{\sigma}_i\\
\QDual{\mbox{\small $\sum$}_{i\in I} a_i.\sigma_i}& = & \mbox{\small $\bigoplus$}_{i\in I} \Dual{a}_i.\QDual{\sigma}_i\\
\QDual{\mbox{\small $\sum$}_{i\in I} \Dual{a}_i.\sigma_i}& = & \mbox{\small $\sum$}_{i\in I} a_i.\QDual{\sigma}_i
\end{array}$$
\end{defi}
It is immediate to check that the operator $\QDual{\cdot}$ is not involutive: $\QDual{\QDual{\Dual{a}+\Dual{b}}} = \QDual{a + b} = \Dual{a}\oplus\Dual{b}$.\\
However it is enough for us it to enjoy the following property.
\begin{lem}
\label{lem:qdual}
Let $\sigma\in\ASC$.
$$\QDual{\sigma} \ACRel \sigma$$
\end{lem}
\proof
By induction on the structure of $\sigma$.  The base case is immediate.
Let us just consider the most interesting case, the other ones being similar.
\begin{description}
\item[Case $\sigma=\mbox{\small $\sum$}_{i\in I} \Dual{a}_i.\sigma_i$]
By the induction hypothesis we have that $\forall i \in I.\QDual{\sigma_i}\ACRel \sigma_i$ (and hence, a fortiori, $\exists k\in (I\cap I).\QDual{\sigma_i}\ACRel \sigma_i$).
So, we get $\QDual{\sigma}=\mbox{\small $\sum$}_{i\in I} a_i.\QDual{\sigma}_i \ACRel
\mbox{\small $\sum$}_{i\in I} \Dual{a}_i.\sigma_i=\sigma$ by 
definition of $\ACRel$ (in particular Definition \ref{def:ACRel}(\ref{def:ACRel-2})).
\qed
\end{description}

\begin{prop}\label{prop:subcontrproperties}
 $\sigma\subcontr\sigma'$ if and only if one of the following conditions holds:
\begin{enumerate}[label=(\arabic*)]
\item \label{prop:subcontrproperties-1} $\sigma = \stopA$;
\item 
\label{prop:subcontrproperties-2}

 $\sigma = \mbox{\small $\bigoplus$}_{j\in J} \Dual{a}_j.\sigma_j$, $\sigma' = \mbox{\small $\sum$}_{i\in I} a_i.{\sigma}_i$ and 
	$\exists k\in (I\cap J)\neq\emptyset.\ {\sigma_k} \subcontr {\sigma_k'}$;

\item 
\label{prop:subcontrproperties-3} 

$\sigma = \mbox{\small $\sum$}_{i\in I} \alpha_i.{\sigma}_i$, $\sigma' = \mbox{\small $\sum$}_{j\in I\cup J} \alpha_j.\sigma'_j$ and
	$\forall h \in I. \; {\sigma_h} \subcontr {\sigma'_h}$;
\item 
\label{prop:subcontrproperties-4} 

$\sigma = \mbox{\small $\bigoplus$}_{j\in I\cup J} \Dual{a}_j.\sigma_j$, $\sigma' =\mbox{\small $\bigoplus$}_{i\in I} \Dual{a}_i.\sigma'_i$ and
	$\forall h \in I. \; {\rho_h} \subcontr {\sigma_h}$.
\end{enumerate}
\end{prop}

\proof
\noindent
($\Rightarrow$)
Let $\sigma\subcontr\sigma'$. Then the only possibilities are necessarily the following ones:
\begin{enumerate}
\item
\label{1}
$\sigma = \stopA$;
\item
\label{2}
$\sigma = \mbox{\small $\bigoplus$}_{j\in J} \Dual{a}_j.\sigma_j$ and $\sigma' = \mbox{\small $\sum$}_{i\in I} \Dual{a}_i.{\sigma}_i$;
\item
\label{3}
$\sigma = \mbox{\small $\sum$}_{i\in I} \alpha_i.{\sigma}_i$ and $\sigma' = \mbox{\small $\sum$}_{j\in I\cup J} \alpha_j.\sigma'_j$;
\item
\label{4}
$\sigma = \mbox{\small $\bigoplus$}_{j\in I\cup J} \Dual{a}_j.\sigma_j$ and $\sigma' =\mbox{\small $\bigoplus$}_{i\in I} \Dual{a}_i.\sigma'_i$.
\end{enumerate}
We proceed now by cases, according to the shapes of $\sigma$ and $\sigma'$.
\begin{enumerate}
\item[(\ref{1})] Immediate.
\item[(\ref{2})]
We show $\exists k\in I\cap J.\ {\sigma_k} \subcontr {\sigma_k'}$ by contradiction.
Let us then assume $\forall k\in (I\cap J)\neq\emptyset.\ {\sigma_k} \not\subcontr {\sigma_k'}$  and let $\Set{\rho_{p_k}}_{k\in(I\cap J) }$ be such that, for any $k\in(I\cap J) $,
$\rho_{p_k}\ACRel\sigma$ and $\neg(\rho_{p_k}\ACRel\sigma')$.
By this and Lemma \ref{lem:qdual}. It is easy to check that  $\mbox{\small $\sum$}_{k\in (I\cap J)} a_k\rho_{p_k}  +  \mbox{\small $\sum$}_{j\in J\setminus(I\cap J)}a_j.\widehat{\sigma}_j \ACRel \sigma$,
whereas
$\neg(\mbox{\small $\sum$}_{k\in (I\cap J)} a_k\rho_{p_k}  +  \mbox{\small $\sum$}_{j\in J\setminus(I\cap J)}a_j.\widehat{\sigma}_j \ACRel \sigma')$. That is $\sigma\not\subcontr\sigma'$.
\end{enumerate}
The other cases can be proved in a similar way.
\qed

Let us define now a stratified version of $\subcontr$ inspired by Proposition \ref{prop:subcontrproperties}.

\begin{defi}\label{def:subcontrproperties}\hfill
\begin{enumerate}[label=(\arabic*)]
\item
 For $n\in\Nat$, the relation $\subcontr_n\subseteq\ASC\times\ASC$ is defined as follows:\\
$\subcontr_0 =  \ASC\times\ASC$\\
For $n>0$, $\stopA\subcontr_n\sigma$ for any $\sigma$, whereas  $\sigma\subcontr_n\sigma'$ holds if one of the following conditions holds:
\begin{enumerate}
\item 
\label{def:subcontrproperties-2}

 $\sigma = \mbox{\small $\bigoplus$}_{j\in J} \Dual{a}_j.\sigma_j$, $\sigma' = \mbox{\small $\sum$}_{i\in I} a_i.{\sigma}_i$ and 
	$\exists k\in (I\cap J)\neq\emptyset.\ {\sigma_k} \subcontr_{n-1} {\sigma_k'}$;

\item 
\label{def:subcontrproperties-3} 

$\sigma = \mbox{\small $\sum$}_{i\in I} \alpha_i.{\sigma}_i$, $\sigma' = \mbox{\small $\sum$}_{j\in I\cup J} \alpha_j.\sigma'_j$ and
	$\forall h \in I. \; {\sigma_h} \subcontr_{n-1} {\sigma'_h}$;
\item 
\label{def:subcontrproperties-4} 

$\sigma = \mbox{\small $\bigoplus$}_{j\in I\cup J} \Dual{a}_j.\sigma_j$, $\sigma' =\mbox{\small $\bigoplus$}_{i\in I} \Dual{a}_i.\sigma'_i$ and
	$\forall h \in I. \; {\rho_h} \subcontr_{n-1} {\sigma_h}$.
\end{enumerate}

\item
We define $\subcontr_{\mathit{co}} \ByDef \bigcap_n \subcontr_n$
\end{enumerate}
\end{defi}

\begin{lem}
$\subcontr \ =\ \subcontr_{\mathit{co}}$
\end{lem}

\begin{defi}[$\subcontr$-semantics for system $\dersc\!$] \label{def:ACRelModel}
Let $\Gamma$ be a set of statements of the form $\rho\subcontrF\sigma$. We define
\begin{enumerate}[label=(\arabic*)]
\item 
$ \begin{array}{rcl}
\modelsSubcontr \Gamma &~~\mbox{ if }~~& \forall {(\rho' \subcontrF \sigma') \in \Gamma} .
\end{array}$ 
$\Pred[ \rho' \subcontr\sigma']$;
 
\item
$ \begin{array}{rcl}
\Gamma \modelsSubcontr \rho \subcontrF \sigma &~~ \mbox{ if } ~~& 
\modelsSubcontr \Gamma \implies \rho \subcontr \sigma
 \end{array} $.
\end{enumerate}
\end{defi}
\noindent
Soundness and completeness of $\dersc$ can hence be formalized as
$$\dersc \sigma\subcontrF\sigma' ~~\Iff~~ \modelsSubcontr\sigma\subcontrF\sigma'$$
As done for $\der$, we use a stratified version of Definition \ref{def:ACRelModel}.

\begin{defi}[Stratified $\subcontr$--semantics for $\dersc\!$] \label{def:StratComplModel}
Let $\Gamma$ be a set of statements of the form $\rho\subcontrF\sigma$ and let $k\in\Nat$. We define
\begin{enumerate}
\item
$ \begin{array}{rcl}
\modelsSubcontr_k \Gamma &\mbox{ if }& \forall {(\rho' \subcontrF \sigma') \in \Gamma} .
 \end{array} $ 
$\Pred[ \rho' \subcontrk{k}\sigma'] $;
\item 
$ \begin{array}{rcl}
\Gamma\modelsSubcontr_k \rho \subcontrF \sigma & \mbox{ if } & \modelsSubcontr_k \Gamma \implies \rho\subcontrk{k}\sigma
 \end{array} $.
\end{enumerate}
\end{defi}
\noindent It is possible now to verify the following:
\begin{lem}\label{lem:stratmodelsubcontr}\hfill
\begin{enumerate}
\item
$\subcontrk{k+1} \subseteq \;\subcontrk{k}$;
\item
$\modelsSubcontr_{k+1} \Gamma \implies \modelsSubcontr_{k} \Gamma$;
\item
\label{lem:stratmodelsubcontriii}
$\forall k. \Pred[ \Gamma\modelsSubcontr_k \sigma\subcontrF\sigma'] \implies \Gamma\modelsSubcontr \sigma\subcontrF\sigma'$.
\end{enumerate}
\end{lem}

\begin{prop}[Soundness of $\dersc$ w.r.t $\modelsSubcontr$]
\label{prop:soundnesssubcontr}
If \,$\Gamma\dersc\rho\subcontrF \sigma$, then \,$\Gamma\modelsSubcontr\rho\subcontrF\sigma$.
\end{prop}
\proof
By Lemma \ref{lem:stratmodelsubcontr}(\ref{lem:stratmodelsubcontriii}),  it is enough to show that 
\[ \begin{array}{rcl}
\Gamma\dersc \rho \subcontrF \sigma &\text{implies}& \Gamma\modelsSubcontr_k \rho \subcontrF \sigma \textrm{ for all }k. 
\end{array} \] 

To do that we proceed by simultaneous induction over the derivation $\Deriv :: \Gamma\dersc \rho \subcontrF \sigma$ and over $k$. \\
Since 
$\Gamma\modelsSubcontr_0\rho \subcontrF \sigma$ trivially holds, we shall keep the case $k = 0$ implicit.
We distinguish the possible cases of the last rule in $\Deriv$.

\begin{description}
\item [Case $(+\cdot+\mbox{\,-\!}\subcontrF)$] Then $\Der$ ends by 
\[ \Inf{\forall h\in I.~~ \Gamma,\mbox{\small $\sum$}_{i\in I} \alpha_i.{\sigma}_i\subcontrF
    	\mbox{\small $\sum$}_{j\in I\cup J} \alpha_j.\sigma'_j\dersc 
    	\sigma_h
    	\subcontrF
    		\sigma'_h}{\Gamma\dersc \mbox{\small $\sum$}_{i\in I} \alpha_i.{\sigma}_i\subcontrF
    	\mbox{\small $\sum$}_{j\in I\cup J} \alpha_j.\sigma'_j}
\]
We have to prove that $\Gamma\modelsSubcontr_k \mbox{\small $\sum$}_{i\in I} \alpha_i.{\sigma}_i\subcontrF
    	\mbox{\small $\sum$}_{j\in I\cup J} \alpha_j.\sigma'_j$ for all $k$.
	
	Let $k > 0$; assume, by the induction hypothesis over $k$, that $\Gamma\modelsSubcontr_{k-1} \mbox{\small $\sum$}_{i\in I} \alpha_i.{\sigma}_i\subcontrF
    	\mbox{\small $\sum$}_{j\in I\cup J} \alpha_j.\sigma'_j$.
	If $\modelsSubcontr_k\Gamma$, then $\modelsSubcontr_{k{-}1} \Gamma$, which implies 
$\mbox{\small $\sum$}_{i\in I} \alpha_i.{\sigma}_i\subcontr_{k{-}1}
    	\mbox{\small $\sum$}_{j\in I\cup J} \alpha_j.\sigma'_j$
and hence 
	$\modelsSubcontr_{k{-}1} \Gamma, \mbox{\small $\sum$}_{i\in I} \alpha_i.{\sigma}_i\subcontrF
    	\mbox{\small $\sum$}_{j\in I\cup J} \alpha_j.\sigma'_j$.\\
By the induction hypothesis over $\Der$ we can  hence get that
$\sigma_h
    	\subcontr_{k{-}1}
    		\sigma'_h$ for all $h$.
Now, by Definition \ref{def:subcontrproperties}, we get $\sigma \subcontr_{k}\sigma'$.
\end{description}
The other cases can be proved similarly.
\qed
\bigskip
\noindent

Completeness can be shown in the same way as done for $\der$ with respect to $\ACRel$.

\begin{prop}[Completeness of $\dersc$ w.r.t $\modelsSubcontr$]
\label{thm:soundcompletesubcontr}
~~$\sigma\subcontr\sigma' ~~\mbox{ implies }~~ \dersc \sigma\subcontrF\sigma'$
\end{prop}

We then get decidability as a corollary.

\begin{cor}[Decidability of $\subcontr$]
The relation $\subcontr$ is decidable.
\end{cor}

\begin{rem}
By Lemma \ref{lem:qdual} it easily descends the following relation between the systems $\der$ and $\dersc$:~~~~
$\dersc \sigma \subcontrF \sigma'  ~~\Rightarrow ~~ \der \QDual{\sigma} \complyP \sigma'$.\\
The opposite does not hold. In fact $\Dual{a}+\Dual{b}=a+b \complyP \Dual{a}$,
but it is easy to check that $\Dual{a}+\Dual{b}\not\subcontr \Dual{a}$.
\end{rem}

\begin{exa}
\label{ex:subcontrder}
Let us now formally derive that $\mathsf{Seller} \subcontr \mathsf{SellerII}$:
\renewcommand{\Price}{\tt{pr}}
\renewcommand{\Card}{\tt{crd}}
\renewcommand{\Cash}{\tt{csh}}
\renewcommand{\Belt}{\tt{blt}}
\renewcommand{\Bag}{\tt{bag}}
\begin{equation}
\label{eq:exampledersubcontr}
\Inf{
 \Inf{
      \Inf
           {
             \Inf
             {}
             {\Gamma_2 \dersc \stopA \subcontrF \stopA}
             }
            {\Gamma_1  \dersc \Cash \subcontrF  \Cash}
            }
            {\textsf{S}  \subcontrF \textsf{SII} \dersc \DPrice.\Cash \subcontrF \DPrice.\Cash}
\quad
\Inf{
      \Inf
           {
             \Inf
             {}
             {\Gamma_4 \dersc \stopA \subcontrF \stopA}
                  \qquad
            \Inf
             {}
             {\Gamma_4 \dersc \stopA \subcontrF \stopA}
             }
            {\Gamma_3  \dersc \Card + \Cash \subcontrF \Card + \Cash + \Cheque}
            }
            {\textsf{S}  \subcontrF \textsf{SII} \dersc \DPrice.(\Card + \Cash) \subcontrF \DPrice.(\Card + \Cash + \Cheque) + \DScrachcard}
         }{
     \dersc \textsf{Seller} \subcontrF \textsf{SellerII} 
           }
\end{equation}
where
\lmcscenterline{\begin{tabular}{l}
$\Gamma_1 = \textsf{Seller}  \subcontrF \textsf{SellerII},\ \DPrice.\Cash \subcontrF \DPrice.\Cash$
\\
$\Gamma_2 = \Gamma_1, \Cash \subcontrF  \Cash$
\\
$\Gamma_3 = \textsf{Seller}  \subcontrF \textsf{SellerII},\ \DPrice.(\Card + \Cash) \subcontrF \DPrice.(\Card + \Cash + \Cheque) + \DScrachcard$
\\
$\Gamma_4 = \Gamma_3, \Card + \Cash \subcontrF \Card + \Cash + \Cheque$
\end{tabular}}
\end{exa}

\subsection{Derivations as orchestrator functors}
By Definition \ref{def:osubcontr} and Theorem \ref{th:complyAwinstrat}, $\sigma\subcontr\sigma'$ implies that, in case there exists $f$  such that $f:\rho\complyO\sigma$, there exists also $f'$  such that $f':\rho\complyO\sigma'$. 
In the present subsection we show how the computation of $f'$ out of $f$ can be effectively carried out by 
any derivation of  $\dersc \sigma\subcontrF\sigma'$, which can indeed be interpreted as 
a functor.

In particular, the following definition shows how to get, out of a derivation $\Der ::\ \dersc \sigma\subcontrF\sigma'$, the functor $\Funct_\Der$ mapping
an orchestrator $f$ such that $f:\rho\complyO\sigma$ for a certain $\rho$, 
to an orchestrator  $f'$ such that $f':\rho\complyO\sigma$.\\

In functors, functor variables will be denoted by {\sc f}, \mbox{\sc f'}, \ldots
In order to better grasp how Definition \ref{def:orchfunctor} below works, we shall use 
the following equivalent presentation of rule 
$(\mbox{\scriptsize $\oplus\cdot+\mbox{\,-\!}\subcontrF$})$
$$
\begin{array}{l}
\Inf[\mbox{\scriptsize $\oplus\cdot+\mbox{\,-\!}\subcontrF$}]
{ \Gamma, \mbox{\small $\bigoplus$}_{i\in I}\Dual{a}_i.\sigma_i \subcontrF
\mbox{\small $\sum$}_{j\in J}\Dual{a}_j.\sigma'_j
     	\dersc \sigma_k
    	  \subcontrF
		\sigma'_k \qquad (k\in I\cap J) \qquad  \mid J\mid \geq 2
}{
    \Gamma \dersc  \mbox{ \small $\bigoplus$}_{i\in I}\Dual{a}_i.\sigma_i \subcontrF
\mbox{\small $\sum$}_{j\in J}\Dual{a}_j.\sigma'_j
}\\
\mbox{\small where  $\mid J\mid$ denotes the cardinality of $J$}
\end{array}
$$
Besides, we shall use the following two rules instead of the more compact
(\mbox{\scriptsize $+\cdot+\mbox{\,-\!}\subcontrF$})
$$\Inf[\mbox{\scriptsize $+\cdot+\mbox{\,-\!}\subcontrF$}\mbox{\,-\!}1]{
(\forall h\in I)\qquad \Gamma,\mbox{\small $\sum$}_{i\in I} a_i.{\sigma}_i\subcontrF
    	\mbox{\small $\sum$}_{j\in I\cup J} a_j.{\sigma'}_j\dersc 
    	\sigma_h
    	\subcontrF
    		\sigma'_h
}{
    \Gamma\dersc \mbox{\small $\sum$}_{i\in I} a_i.{\sigma}_i\subcontrF
    	\mbox{\small $\sum$}_{j\in I\cup J} a_j.{\sigma'}_j
}
$$
$$
\begin{array}{c}
\Inf[\mbox{\scriptsize $+\cdot+\mbox{\,-\!}\subcontrF$}\mbox{\,-\!}2]{
(\forall h\in I)\qquad \Gamma,\mbox{\small $\sum$}_{i\in I} \Dual{a}_i.{\sigma}_i\subcontrF
    	\mbox{\small $\sum$}_{j\in I\cup J} \Dual{a}_j.{\sigma'}_j\dersc 
    	\sigma_h
    	\subcontrF
    		\sigma'_h
}{
    \Gamma\dersc \mbox{\small $\sum$}_{i\in I} \Dual{a}_i.{\sigma}_i\subcontrF
    	\mbox{\small $\sum$}_{j\in I\cup J} \Dual{a}_j.{\sigma'}_j
}\\
\mbox{\small  where the same proviso for rule $(\mbox{\scriptsize $+\cdot+\mbox{\,-\!}\subcontrF$})$
(see Definition \ref{def:formalSubcontr}) applies.}
\end{array}
$$

\bigskip
We call {\em extended environment} a set of elements of the form
$\mbox{\sc f} : \sigma {\subcontrF} \sigma'$  where $\mbox{\sc f}$ is an orchestrator variable and $\sigma,\sigma'\in\ASC$. We use the symbols $\widehat\Gamma, \widehat\Gamma',\ldots$
to range over extended environments.

\begin{defi}
\label{def:orchfunctor}
Given $\Der ::\ \Gamma \dersc \sigma\subcontrF\sigma'$, the functor $\Funct_\Der$ is defined
by 
$$
\Funct_\Der = \auxFunct{\Der}{\emptyset}
$$
where $\mathcal{F}$ is binary partial function from derivations and extended environments to orchestrator functors. We define $\mathcal{F}$  by induction on the structure of the first argument. The second argument is used as an accumulator during the construction of the orchestrator functor.  It enables us to build recursive orchestrators by keeping the name of recursion variables to be used in case get to an application of rule $(\mbox{\scriptsize $\CkptcomplHyp\mbox{\,-\!}\subcontrF$})$. So a $\rec$ binder is introduced in any step
of the construction of an orchestrator. Of course the variable introduced will remain vacuously bind in case rule  $(\mbox{\scriptsize $\CkptcomplHyp\mbox{\,-\!}\subcontrF$})$ is not encountered. We assume any orchestrator variable introduced  to be fresh.
$ \auxFunct{\Der}{\widehat\Gamma}$ is inductively defined as follows.
\medskip

\begin{description}
\item[$\Der=\prooftree 
\justifies
	\Gamma\dersc \stopA \subcontrF \sigma'
\using(\mbox{\scriptsize $\CkptcomplAx \mbox{\,-\!}\subcontrF$})
\endprooftree$~]\hfill\\[2mm]
$\auxFunct{\Der}{\widehat\Gamma}= \bm{\lambda} f.~\stopf$\\

\bigskip
\item[$\Der=
\prooftree
\justifies
\Gamma, \sigma\subcontrF\sigma' \dersc \sigma\subcontrF\sigma'
\using(\mbox{\scriptsize $\CkptcomplHyp\mbox{\,-\!}\subcontrF$})
\endprooftree$~]\hfill\\[2mm]
$\auxFunct{\Der}{\widehat\Gamma}=\mbox{\sc f}$\\
if $\mbox{\sc f}: \sigma{\subcontrF}\sigma' \in \widehat\Gamma$
\bigskip

\item[$\Der=
\prooftree  \begin{array}{c}
                              \mathcal{D'}\\
\Gamma, \mbox{\small $\bigoplus$}_{i\in I}\Dual{a}_i.\sigma_i \subcontrF
\mbox{\small $\sum$}_{j\in J}\Dual{a}_j.\sigma'_j
\dersc \sigma_k
    	  \subcontrF
		\sigma'_k \quad (k\in I\cap J) \quad  \mid J\mid \geq 2
                \end{array}
\justifies
  \Gamma \dersc \mbox{\small $\bigoplus$}_{i\in I}\Dual{a}_i.\sigma_i \subcontrF
\mbox{\small $\sum$}_{j\in J}\Dual{a}_j.\sigma'_j
\using(\mbox{\scriptsize $\oplus\cdot+\mbox{\,-\!}\subcontrF$})
\endprooftree$~]\hfill\\[2mm]
$\auxFunct{\Der}{\widehat\Gamma} =
\rec \mbox{\sc f}.  \bm{\lambda} f.~\IF~~(f = \bigvee_{h\in H\supseteq I}  \orchAct{\Dual{a}_h}{a_h}.f_h)~~
 \THEN ~~\orchAct{\Dual{a}_k}{a_k}^+.\auxFunct{\Der'}{\widehat\Gamma'}(f_k)$~~
 \ELSE ~~$\stopf$,
where
$\widehat\Gamma' =  \widehat\Gamma, \mbox{\sc f}: \mbox{\small $\bigoplus$}_{i\in I}\Dual{a}_i.\sigma_i \subcontrF
\mbox{\small $\sum$}_{j\in J}\Dual{a}_j.\sigma'_j$.

Note that the interaction between a contract  and $\mbox{\small $\bigoplus$}_{i\in I}\Dual{a}_i.\sigma_i $ (with $|I|\geq 2$), in case they are affectible compliant, is mediated by an orchestrator
which is necessarily of the form $\bigvee_{h\in H\supseteq I}  \orchAct{\Dual{a}_h}{a_h}.f_h$
or $\stopf$ (in case the contract is $\stopA$).

\bigskip
\item[$\Der=
\prooftree 
    (\forall h\in I)~~ \begin{array}{c}
                          \mathcal{D}_h \\
\Gamma,\mbox{\small $\sum$}_{i\in I} a_i.{\sigma}_i\subcontrF
    	\mbox{\small $\sum$}_{j\in I\cup J} a_j.{\sigma'}_j\dersc 
    	\sigma_h
    	\subcontrF
    		\sigma'_h
                            \end{array}
  \justifies
    \Gamma\dersc \mbox{\small $\sum$}_{i\in I} a_i.{\sigma}_i\subcontrF
    	\mbox{\small $\sum$}_{j\in I\cup J} a_j.{\sigma'}_j
\using (\mbox{\scriptsize $+\cdot+\mbox{\,-\!}\subcontrF$}\mbox{\,-\!}1)
\endprooftree$~]\hfill\\[2mm]

\begin{tabbing}
$\auxFunct{\Der}{\widehat\Gamma}= $\\
 $\rec \mbox{\sc f}. \bm{\lambda} f.$~\IF~~\=$(f = \bigvee_{h\in H'}  \orchAct{a_h}{\Dual{a}_h}.f_h)$~~ \\
\>  $\THEN ~~\bigvee_{h\in H'\cap I} \orchAct{a_h}{\Dual{a}_h}.\auxFunct{\Der_h}{\widehat\Gamma'}(f_h)$\\
\>  \ELSE ~~~~\CASE~~\=$~~f$\\
\> \> ~~ \OF   ~~~~\=$\orchAct{a_k}{\Dual{a}_k}^+.f'_k$
 \quad
 \=$\textbf{:} \quad\orchAct{a_h}{\Dual{a}_h}^+.\auxFunct{\Der_h}{\widehat\Gamma'}(f'_h)$ \qquad $(\forall h\in I)$\\
\> \> \>\textbf{otherwise} \> \textbf{:} \quad $\stopf$
\end{tabbing}
where
$\widehat\Gamma' =  \widehat\Gamma, \mbox{\sc f}: \mbox{\small $\sum$}_{i\in I} a_i.{\sigma}_i\subcontrF
    	\mbox{\small $\sum$}_{j\in I\cup J} a_j.{\sigma'}_j$\\[1mm]

Note that a contract which is in $\ACRel$ relation with  $\mbox{\small $\sum$}_{i\in I} a_i.{\sigma}_i$ is either of the form $\mbox{\small $\oplus$}_{h\in H\subseteq I} \Dual{a}_h.{\sigma}'_h$ or of the form $\mbox{\small $\sum$}_{h\in H} \Dual{a}_h.{\sigma}'_h$ with $H\cap I \neq \emptyset$. Moreover, in the first case the orchestrator which mediates their interaction 
has necessarily the form $\bigvee_{h\in H'\supseteq H}  \orchAct{a_h}{\Dual{a}_h}.f_h$;
in the second case it has necessarily the form
 $\orchAct{a_k}{\Dual{a}_k}^+.f'$ with $k\in I$. We have also to take into account the possibility of $f$ to be $\stopf$ (in case the contract is $\stopA$).

\bigskip
\item[$\Der=
\prooftree 
    (\forall h\in I)~~ \begin{array}{c}
                          \mathcal{D}_h \\
\Gamma,\mbox{\small $\sum$}_{i\in I} \Dual{a}_i.{\sigma}_i\subcontrF
    	\mbox{\small $\sum$}_{j\in I\cup J} \Dual{a}_j.{\sigma'}_j\dersc 
    	\sigma_h
    	\subcontrF
    		\sigma'_h
                            \end{array}
  \justifies
    \Gamma\dersc \mbox{\small $\sum$}_{i\in I} \Dual{a}_i.{\sigma}_i\subcontrF
    	\mbox{\small $\sum$}_{j\in I\cup J} \Dual{a}_j.{\sigma'}_j
\using (\mbox{\scriptsize $+\cdot+\mbox{\,-\!}\subcontrF$}\mbox{\,-\!}2)
\endprooftree$~]\hfill\\[2mm]
\begin{tabbing}
$\auxFunct{\Der}{\widehat\Gamma}= $\\
$\rec \mbox{\sc f}. \bm{\lambda} f.$~\=\IF~ $|I|=1$ ~\AND\  $f=\orchAct{\Dual{a}}{a}.f'  \vee f''$\\
\> \THEN\  $\orchAct{\Dual{a}}{a}^+.\auxFunct{\Der'}{\widehat\Gamma}(f')$\\
\> \ELSE\   \=\CASE~~$~~f$\\
\> \> ~~ \OF   ~~~~\=$\orchAct{\Dual{a}_h}{a_h}^+.f'_h$
 \quad
 \=$\textbf{:} \quad\orchAct{\Dual{a}_h}{a_h}^+.\auxFunct{\Der_h}{\widehat\Gamma'}(f_h)$ \qquad $(\forall h\in I)$\\
\> \> \>\textbf{otherwise} \> \textbf{:} \quad $\stopf$
\end{tabbing}
where
$\widehat\Gamma' =  \widehat\Gamma, \mbox{\sc f}: \mbox{\small $\sum$}_{i\in I} \Dual{a}_i.{\sigma}_i\subcontrF
    	\mbox{\small $\sum$}_{j\in I\cup J} \Dual{a}_j.{\sigma'}_j$\\[1mm]
      
Note that a contract which is in $\ACRel$ relation with  $\mbox{\small $\sum$}_{i\in I} \Dual{a}_i.{\sigma}_i$ is of the form $\mbox{\small $\sum$}_{h\in H} a_h.{\sigma}'_h$ with $I\cap H\neq\emptyset$. Moreover, in case $|I|\geq 2$, the orchestrator which mediates their interaction 
has necessarily the form  $\orchAct{\Dual{a}_k}{a_k}^+.f'$
with $k\in I\cap H$. If, instead, $|I| = 1$, the orchestrator could also be of the form
$\orchAct{\Dual{a}}{a}.f' \vee f''$. In such a case, the functor has to transform the orchestration action $\orchAct{\Dual{a}}{a}$ into $\orchAct{\Dual{a}}{a}^+$
(see Example \ref{ex:functSeller} for an example of a case like that).
Notice that the present rule cannot be used in case $|I| = |I\cup J| = 1$.
We have also to take into account the possibility of $f$ to be $\stopf$ (in case the contract is $\stopA$).

\bigskip
\item[$\Der=
\prooftree 
     (\forall h\in I)~~ \begin{array}{c}
                          \mathcal{D}_h \\
 \Gamma,\mbox{\small $\bigoplus$}_{j\in I\cup J}\Dual{a}_j.{\sigma}_j\subcontrF \mbox{\small $\bigoplus$}_{i\in I} \Dual{a}_i.{\sigma'}_i \dersc 
    	\sigma_h
    	\subcontrF
    		\sigma'_h  
                         \end{array}
  \justifies
    \Gamma\dersc\mbox{\small $\bigoplus$}_{j\in I\cup J} \Dual{a}_j.{\sigma}_j\subcontrF \mbox{\small $\bigoplus$}_{i\in I} \Dual{a}_i.{\sigma'}_i    	
\using (\mbox{\scriptsize $\oplus\cdot\oplus\mbox{\,-\!}\subcontrF$})
\endprooftree$~]\hfill\\[2mm]
$\auxFunct{\Der}{\widehat\Gamma}= $\\
$ \rec \mbox{\sc f}. \bm{\lambda} f.$~\IF $f = \bigvee_{h\in H'}  \orchAct{\Dual{a}_h}{a_h}.f_h$
 \THEN~~ $\bigvee_{h\in I}\orchAct{\Dual{a}_h}{a_h}.\auxFunct{\Der_h}{\widehat\Gamma'}(f_h)$~~
 \ELSE ~~$\stopf$\\[1mm]
where
${\widehat\Gamma}'=  {\widehat\Gamma}, \mbox{\sc f}: \mbox{\small $\bigoplus$}_{j\in I\cup J} \Dual{a}_j.{\sigma}_j\subcontrF \mbox{\small $\bigoplus$}_{i\in I} \Dual{a}_i.{\sigma'}_i $

Note that a contract which is in $\ACRel$ relation with  $\mbox{\small $\bigoplus$}_{j\in I\cup J} \Dual{a}_j.{\sigma}_j$ is necessarily of the form $\mbox{\small $\sum$}_{h\in H\supseteq I\cup J} a_h.{\sigma}'_h$. Moreover, the orchestrator which mediates their interaction 
has necessarily the form $\bigvee_{h\in H'\supseteq I\cup J}  \orchAct{\Dual{a}_h}{a_h}.f_h$.
We have also to take into account the possibility of $f$ to be $\stopf$ (in case the contract is $\stopA$).

\end{description}
\end{defi}
It is easy to check the following.
\begin{fact}
The function $\mathcal{F}$ restricted to empty contexts is total, i.e.
given a derivation $\Der$ in $\dersc\!$,
the computation of $\Funct_\Der$ always returns an orchestrator functor.
\end{fact}
\bigskip

\begin{exa}
Let us consider the affectible contracts ${\tt d} + {\tt b}.(\Dual{\tt b} \oplus \Dual{\tt c}) $ and ${\tt d}.\Dual{\tt a} + {\tt b}.(\Dual{\tt a} + \Dual{\tt c} + \Dual{\tt e})$.
Is possible to prove that 
$$
{\tt d} + {\tt b}.(\Dual{\tt b} \oplus \Dual{\tt c}) \subcontr {\tt d}.\Dual{\tt a} + {\tt b}.(\Dual{\tt a} + \Dual{\tt c}+ \Dual{\tt e}) 
$$
by means of the following derivation $\Der$ in system $\dersc$.

\[
\Inf[\mbox{\scriptsize $+\cdot+\mbox{\,-\!}\subcontrF$}\mbox{\,-\!}1]{
     \Inf[\mbox{\scriptsize $\CkptcomplAx \mbox{\,-\!}\subcontrF$}]{}
     { \Gamma_1 \dersc \stopf \subcontrF \Dual{\tt a} } \qquad  
      \Inf[\mbox{\scriptsize $\oplus\cdot+\mbox{\,-\!}\subcontrF$}]{ \Inf[\mbox{\scriptsize $\oplus\cdot\oplus\mbox{\,-\!}\subcontrF$}]{ \Inf[\mbox{\scriptsize $\CkptcomplAx \mbox{\,-\!}\subcontrF$}]{}{ \Gamma_3 \dersc  \stopf \subcontrF \stopf} }{ \Gamma_2 \dersc  \Dual{\tt c} \subcontrF \Dual{\tt c}  }  }
    { \Gamma_1 \dersc  \Dual{\tt b} \oplus \Dual{\tt c} \subcontrF \Dual{\tt a} + \Dual{\tt c} }
      }
{ \dersc {\tt d} + {\tt b}.(\Dual{\tt b} \oplus \Dual{\tt c}) \subcontrF  {\tt d}.\Dual{\tt a} + {\tt b}.(\Dual{\tt a} + \Dual{\tt c} + \Dual{\tt e}) }
\]
where
\lmcscenterline{
\begin{tabular}{l}
$\Gamma_1 = \Set{ {\tt d} + {\tt b}.(\Dual{\tt b} \oplus \Dual{\tt c}) \subcontrF  {\tt d}.\Dual{\tt a} + {\tt b}.(\Dual{\tt a} + \Dual{\tt c} + \Dual{\tt e})}$\\
$\Gamma_2 = \Gamma_1, \Dual{\tt b} \oplus \Dual{\tt c} \subcontrF \Dual{\tt a} + \Dual{\tt c}$\\
$\Gamma_3 = \Gamma_2, \Dual{\tt c} \subcontrF \Dual{\tt c} $
\end{tabular}}

\medskip

\renewcommand{\Price}{\tt{pr}}
\renewcommand{\Card}{\tt{crd}}
\renewcommand{\Cash}{\tt{csh}}
\renewcommand{\Belt}{\tt{bt}}
\renewcommand{\Bag}{\tt{bg}}

\noindent
According to Definition  \ref{def:orchfunctor}, out of $\Der$ we can compute the following functor

\begin{tabbing}
$\mathbf{F}_{\mathcal{D}}$ =$\bm{\lambda} f.~$\=\IF~~$(f =  \orchAct{\tt d}{\Dual{\tt d}}.f_{\tt d}\vee \orchAct{\tt b}{\Dual{\tt b}}.f_{\tt b}\vee f')$\\~~
\>  \THEN ~~\= $\orchAct{\tt d}{\Dual{\tt d}}.\stopf$ $\vee$ $\orchAct{\tt b}{\Dual{\tt b}}.\ldots  $\\
\> \ELSE~~~ \=\CASE~~~~$f$\\
\> \>   \OF   ~~~~\=$\orchAct{\tt d}{\Dual{\tt d}}^+.f_{\tt d}$
 ~
 \=$\textbf{:} ~\orchAct{\tt d}{\Dual{\tt d}}^+.\stopf$\\[1mm]
\> \>  \>$\orchAct{\tt b}{\Dual{\tt b}}^+.f_{\tt b}$
 ~
 $\textbf{:} ~\orchAct{\tt b}{\Dual{\tt b}}^+$.\=\IF~~ $(f_{\tt b}=\orchAct{\Dual{\tt b}}{\tt b}.f'_{\tt b} \vee \orchAct{\Dual{\tt c}}{\tt c}.f'_{\tt c} \vee f'')$\\
\>\>\> \> \THEN ~~  $\orchAct{\Dual{\tt c}}{\tt c}^+.\stopf$    \\[0mm]
\>\>\>\> \ELSE ~~ $\stopf$\\[1mm] 
\>\>\> $\mathbf{otherwise}$  ~~ $ \textbf{:} ~\stopf$\\[1mm]
{\small (We have omitted the part '$\ldots$' for sake of readability, as well as the vacuous $\rec$ binders.)}
\end{tabbing}

It is possible to check that
$$
\Dual{\tt c} + \Dual{\tt b}. ({\tt b} + {\tt c}) \ACRel {\tt d} + {\tt b}.(\Dual{\tt b} \oplus \Dual{\tt c})
$$
and that  
$$\mathsf{f}: \Dual{\tt c} + \Dual{\tt b}. ({\tt b} + {\tt c}) 
\complyO 
{\tt d} + {\tt b}.(\Dual{\tt b} \oplus \Dual{\tt c})$$
where
$\mathsf{f} = \orchAct{\tt b}{\Dual{\tt b}}^+.(\orchAct{\Dual{\tt b}}{\tt b}.\stopf \vee \orchAct{\Dual{\tt c}}{\tt c}.\stopf)$

We have now that 
$$\mathbf{F}_{\mathcal{D}}(\mathsf{f}) = \orchAct{\tt b}{\Dual{\tt b}}^+.\orchAct{\Dual{\tt c}}{\tt c}^+.\stopf$$
and $\mathbf{F}_{\mathcal{D}}(\mathsf{f}) :   \Dual{\tt c} + \Dual{\tt b}. ({\tt b} + {\tt c}) 
\complyO 
{\tt d}.\Dual{\tt a} + {\tt b}.(\Dual{\tt a} + \Dual{\tt c} + \Dual{\tt e}))$
\end{exa}

\begin{exa}
\label{ex:functSeller}
Let us take $f = \orchAct{\Bag}{\DBag}^+.\orchAct{\DPrice}{\Price}.(\orchAct{\Card}{\DCard}.\stopf\vee\orchAct{\Cash}{\DCash}.\stopf)$. We have seen in Example  \ref{ex:startoutorch} that
$f : \textsf{Buyer}\complyO\textsf{Seller}$. If we dub now $\Der'$ the derivation \ref{eq:exampledersubcontr} in Example \ref{ex:subcontrder}, we get that
$$
\Funct_{\Der'}(f) = 
 \orchAct{\Bag}{\DBag}^+.\orchAct{\DPrice}{\Price}^+.(\orchAct{\Card}{\DCard}.\stopf\vee\orchAct{\Cash}{\DCash}.\stopf).
$$
Notice how the functor does transform the orchestration action $\orchAct{\DPrice}{\Price}$
of $f$ into the action $\orchAct{\DPrice}{\Price}^+$. In fact, whereas $\orchAct{\DPrice}{\Price}$ in $f$ has simply to enable the exchange of the $\Price$ message, the 
new orchestrator in that point has to deal with an affectible choice, since also the summand $\DScrachcard$ is present in  $\mathsf{SellerII}$.
\end{exa}

It is possible to prove that any functor $\Funct_{\Der}$ behaves as expected.

\begin{thm}[$\dersc$ derivations as orchestrator functors]
\label{thm:functder}
Given $\Der ::\ \dersc \sigma\subcontr\sigma'$, it is possible to
compute a functor $\Funct_\Der:\Orch\rightarrow\Orch$ such that:
\lmcscenterline{
for any $\rho$ and $f$ such that $f:\rho\complyO\sigma$, it holds that
$\Funct_\Der(f):\rho\complyO\sigma'$.}
\end{thm}
\proof
See  Appendix \ref{subsec:functder}.
\qed

%

\section{Conclusion and Future Work}
\label{sec:concl}

We have studied two approaches to loosening compliance among a client and a server in contract theory, based on the concepts of dynamic adaptation and of mediated interaction respectively. We have seen that these induce equivalent notions of compliance, which can be shown
via the abstract concept of winning strategy in a suitable class of games.

The byproduct is that the existence of the agreement among two contracts specifying adaptive behaviours is established by statically synthesizing the proper orchestrator, hence avoiding any trial and error mechanism at run time. The study in this paper has been limited to the case of binary sessions since this is the setting in which both orchestrators and retractable contracts have been introduced. However strategy based concepts of agreement have been developed in the more general scenario of multiparty interaction, which seems a natural direction for future work.

\subsubsection*{Acknowledgments.} The authors wish to thank Mariangiola Dezani for her support and Massimo Bartoletti for the preliminary insight that led to the development of the paper.

\bibliographystyle{plain}
\bibliography{session}

\appendix



\section{Proof of Proposition \ref{th:complyequivtbcomply} ($\ACRel$ and $\complyTB$ equivalence).}
\label{appendix:complyequivtbcomply}


The proof of Theorem \ref{th:complyequivtbcomply} is developed in the present section
along the following lines:
\begin{itemize}
\item
We first recall the equivalent stratified version of Definition \ref{def:ACRel};
\item
To each pair ${\rho}\pp{\sigma}$ we associate a set of regular trees $\rts{{\rho}\pp{\sigma}}$ 
and show that ${\rho}\ACRel{\sigma}$ holds if and only if 
there exists a tree in $\rts{{\rho}\pp{\sigma}}$ with no leaf labeled by the symbol  $\xmark$;
\item
A set $\rts{{\rho}\tbpp{\sigma}}$ of regular trees is also associated to any turn-based
system ${\rho}\tbpp{\sigma}$. Also for ${\rho}\tbpp{\sigma}$, we prove ${\rho}\complyTB{\sigma}$ to hold if and only if 
there exists a tree in $\rts{{\rho}\tbpp{\sigma}}$ with no leaf labeled by $\xmark$;
\item
We conclude by showing how to map $\rts{{\rho}\pp{\sigma}}$ to $\rts{{\rho}\tbpp{\sigma}}$
and vice versa so that a tree without $\xmark$ is always sent to a tree with the same property.
\end{itemize}

In Definition \ref{def:coACRel} we have seen that the co-inductive definition of
$\ACRel$ (Def. \ref{def:ACRel}) can be stratified by the family $\Set{ \!{\ACRelk{k}}\!}_{k\in \Nat}$, such that $\ACRelk{k} \subseteq \ACRelk{k-1}$ for all $k$ and  $\ACRel = \bigcap_{n\in\Nat} \ACRelk{n}$.

\begin{defi}
Let ${\rho},{\sigma}\in\ASC$. We define the set of regular trees of the pair ${\rho}\pp{\sigma}$, which we dub $\rts{{\rho}\pp{\sigma}}$, as follows:\\
\[
\begin{array}{rcl@{~~~~~~~~}l}

\rts{\stopA\pp{\sigma}} & = 
& \Big\{ ~
             \stopA\pp{\sigma} ~\Big\}\\[6mm]

\rts{{\rho}\pp{\sigma}} & = 
 & \Big\{ \begin{array}{c}
            {\rho}\pp{\sigma} \\
            / \cdots\backslash \\
            \mathsf{T}_1 \cdots ~~\mathsf{T}_n
             \end{array}~ \Big| ~~h\in I, \mathsf{T}_h\in \rts{{\rho}_h\pp{\sigma}_h}\Big\}\\[2mm]
& & \mbox{if either }  \rho = \bigoplus_{i\in I}\Dual{a}_i.\rho_i, \hspace{1mm}\sigma = \sum_{j\in J}a_j.\sigma_j,\hspace{1mm} I\subseteq J\\
& &  \mbox{\hspace{2.5mm} or \hspace{6mm}} \rho = \sum_{j\in J}a_j.\rho_j, \hspace{1mm}\sigma = \bigoplus_{i\in I}\Dual{a}_i.\sigma_i,
	\hspace{1mm}I\subseteq J\\
& & \mbox{where } I=\Set{1,..,n}. 
\\[4mm]

\rts{{\rho}\pp{\sigma}} & = 
 & \bigcup_{h\in I\cap J}\Big\{  
\begin{array}{c}
            {\rho}\pp{\sigma} \\
            \mid  \\
            \mathsf{T}
             \end{array}   ~ \Big| ~~ \mathsf{T} \in \rts{{\rho}_h\pp{\sigma}_h}
\Big\}\\
& & \mbox{if } \rho = \sum_{i\in I}\alpha_i.\rho_i, \hspace{1mm}\sigma = \sum_{j\in J}\Dual{\alpha}_j.\sigma_j,\hspace{1mm}
	h \in I \cap J 
	\\[6mm]

\rts{{\rho}\pp{\sigma}} & = 
& \Big\{  
\begin{array}{c}
            {\rho}\pp{\sigma} \\
            \mid  \\
            \xmark
             \end{array}
\Big\} ~~~~~~~~~~
  \mbox{ otherwise } 
\end{array}
\]
\end{defi}

%

\begin{lem}\label{lem:equivtb2}
 Let $\rho, \sigma \in \ASC$.
$\rho\ACRel\sigma$ ~~  iff  ~~ there exists a tree in $\rts{\rho\pp\sigma}$ without any leaf labeled by $\xmark$.
\end{lem}
\proof
($\Leftarrow$)~~
Let $\mathsf{T}\in\rts{\rho\pp\sigma}$ such that no leaf is labeled by $\xmark$ and consider the following relation:
$$\mathcal{R} = \Set{(\rho',\sigma') \mid \rho'\pp\sigma' \mbox{ is a node of } \mathsf{T}}$$
Then $\rho\, \mathcal{R}\, \sigma$ as ${\rho}\pp{\sigma}$ is the label of the root of $\mathsf{T}$. Besides,
it is easy to check that $\mathcal{R}$ is an affectible compliance relation according to
Definition \ref{def:coACRel}.\\

\noindent
($\Rightarrow$)~~ Recall that $\ACRel = \bigcap_{k\in\Nat} \ACRel_{\!k}$. Then we prove by induction over $k$ that 
if $\rho\ACRel_{\!k} \sigma$ then there exists a tree $\mathsf{T} \in \rts{\rho\pp\sigma}$ such the cut of $\mathsf{T}$ at level $k$, written
$\mathsf{T} {\downarrow}_k$ has no leaf labelled by $\xmark$. 

If $k =0$ we just observe that $\mathsf{T} {\downarrow}_0$ is the root of a tree $\mathsf{T} \in \rts{\rho\pp\sigma}$, that cannot be labelled by 
$\xmark$. 

Let $k>0$, then we proceed according to the cases in the definition of $\mathsf{T}$. 
If $\mathsf{T}{\downarrow}_k\  = \stopA\pp{\sigma}$ there is nothing to prove.
The case $\mathsf{T}{\downarrow}_k ~ = \Big\{  
\begin{array}{c}
            {\rho}\pp{\sigma} \\
            \mid  \\
            \xmark
             \end{array}
\Big\}$ is impossible since it implies that none of the cases of Definition \ref{def:coACRel}(\ref{def:coACRel-i}) holds. This contradicts $\rho\ACRel_{\!k} \sigma$, as $k>0$.

Suppose that $\mathsf{T} ~ = \begin{array}{c}
            {\rho}\pp{\sigma} \\
            / \cdots\backslash \\
            \mathsf{T}_1 \cdots ~~\mathsf{T}_n
             \end{array} \in \rts{\rho\pp\sigma}$, where 
$\rho$ and $\sigma$ are as described in the definition.
The hypothesis $\rho\ACRel_{\!k}\sigma$ implies that $\rho_h\ACRel_{\! k-1}\sigma_h$ for all $h$; hence by induction, for all $h$
there exists $\mathsf{T}'_h \in \rts{\rho_h\pp\sigma_h}$ such that $\mathsf{T}'_h {\downarrow}_{k-1}$ has no leaf labelled by $\xmark$. 
It follows that $\mathsf{T}' = \begin{array}{c}
            {\rho}\pp{\sigma} \\
            / \cdots\backslash \\
            \mathsf{T}'_1 \cdots ~~\mathsf{T}'_n
             \end{array}$ 
is a tree in $\rts{\rho\pp\sigma}$, such that $\mathsf{T}' {\downarrow}_k$ has no leaf labelled by $\xmark$.

Finally if $\mathsf{T} ~ = \begin{array}{c}
            {\rho}\pp{\sigma} \\
            \mid  \\
            \mathsf{T}_h
             \end{array}$
where $\rho$ and $\sigma$ are as in the third case of the definition and $h\in I\cap J$ and
$\mathsf{T}_h \in \rts{\rho_h\pp\sigma_h}$, then by induction and reasoning as in the previous case we get
a tree $\mathsf{T}'_h \in \rts{\rho_h\pp\sigma_h}$ with no leaf labelled by $\xmark$. Hence 
$\mathsf{T}' ~ = \begin{array}{c}
            {\rho}\pp{\sigma} \\
            \mid  \\
            \mathsf{T}'_h
             \end{array} \in \rts{\rho\pp\sigma}$
is such that $\mathsf{T}' {\downarrow}_k$ has no leaf labelled by $\xmark$.
\qed

\begin{defi}
\label{def:rtstbpp}
Let $\tilde{\rho},\tilde{\sigma}\in\sctb$. We define the set of regular trees of the system $\tilde{\rho}\tbpp\tilde{\sigma}$, which we dub $\rts{\tilde{\rho}\tbpp\tilde{\sigma}}$; 
let $B=\{\playerA{:}a,\playerA{:}\Dual{a},\playerB{:}a,\playerB{:}\Dual{a} \mid a\in\Names\}$, then:
\[
\begin{array}{rcl@{~~~~~~~~}l}

\rts{\stopA\tbpp\tilde{\sigma}} & = 
& \Big\{ ~\begin{array}{c}
             \stopA\tbpp\tilde{\sigma}\\
            \mid \\
           \Zero\tbpp\tilde{\sigma}
            \end{array} ~\Big\}\\[6mm]

\rts{\tilde{\rho}\tbpp\tilde{\sigma}} & = 
 & \Big\{ \begin{array}{c}
            \tilde{\rho}\tbpp\tilde{\sigma} \\
            / \cdots\backslash \\
            \mathsf{T}_1 \cdots ~~\mathsf{T}_n
             \end{array}~ \Big| ~~ \mathsf{T}_1\in \rts{\tilde{\rho}_1\tbpp\tilde{\sigma}_1}, \ldots ,\mathsf{T}_n\in \rts{\tilde{\rho}_n\tbpp\tilde{\sigma}_n}\Big\}\\
 & &
\mbox{where } \Set{\tilde{\rho}_i\tbpp\tilde{\sigma}_i}_{i=1..n} = \Set{\tilde{\rho}'\tbpp\tilde{\sigma}' ~\mid~\tilde{\rho}\tbpp\tilde{\sigma} \tblts{\beta}\tilde{\rho}'\tbpp\tilde{\sigma}', ~\beta\in B } \neq \emptyset \\[6mm]


\rts{\tilde{\rho}\tbpp\tilde{\sigma}} & = 
 & \Big\{  
\begin{array}{c}
            \tilde{\rho}\tbpp\tilde{\sigma} \\
            \mid  \\
            \mathsf{T}
             \end{array}   ~ \Big| \quad \exists \tilde{\rho}', \tilde{\sigma}', a.~
             \tilde{\rho}\tbpp\tilde{\sigma} \tblts{\playerC:a}\tilde{\rho}'\tbpp\tilde{\sigma}', \mathsf{T} \in
             	\rts{\tilde{\rho}'\tbpp\tilde{\sigma}'} \Big\}
             \\[6mm]

\rts{\tilde{\rho}\tbpp\tilde{\sigma}} & = 
& \Big\{  
\begin{array}{c}
            \tilde{\rho}\tbpp\tilde{\sigma} \\
            \mid  \\
            \xmark
             \end{array}
\Big\}
~~~~~~~~  \mbox{ if } ~~\rho\neq\stopA \mbox{ and }\tilde{\rho}\tbpp\tilde{\sigma} ~\notblts{}\\[2mm]

\end{array}
\]
\end{defi}

%

\begin{lem}\label{lem:equivtb1}
 Let $\rho, \sigma \in \ASC (\subseteq \tbASC)$.\\
$\rho\complyTB\sigma$ ~~  iff  ~~ there exists a tree in $\rts{\rho\tbpp\sigma}$
such that no leaf is labeled by $\xmark$. 
\end{lem}
\begin{proof}
Similar to the proof of Lemma \ref{lem:equivtb2}.
\end{proof}

Next we define a function $\ff : \rts{\rho\tbpp\sigma} \rightarrow  \rts{\rho\pp\sigma}$, for which some technical results are in order.

\begin{lem}
\label{lem:tbpptoppA}
Let ${\rho},{\sigma}\in\ASC\subseteq\tbASC$ and let  $\tilde{\rho},\tilde{\sigma}\in\tbASC$ such that $~{\rho}\tbpp{\sigma}\tbltsstar{} \tilde{\rho}\tbpp\tilde{\sigma}$.
\begin{enumerate}
\item If $~\tilde{\rho}\tbpp\tilde{\sigma}\tblts{\beta}$ then only one of the following cases can occur:
\begin{enumerate}
\item 
$\beta\in\{\playerA:a,\playerB:a \mid a\in\Names\}$ and $\beta$ is unique;
\item
$\beta\in\{\playerA:\Dual{a},\playerB:\Dual{a} \mid a\in\Names\}$;
\item
$\beta\in\{\playerC:a \mid a\in\Names\}$ and $\beta$ is unique.
\end{enumerate}
\item
If ${\tilde{\rho}}\complyTB{\tilde{\sigma}}$ and  $\tilde{\rho}\tbpp\tilde{\sigma}\tblts{\beta} \tilde{\rho}'\tbpp\tilde{\sigma}'\tblts{}$, where $\beta\in\{\playerA:\Dual{a},\playerB:\Dual{a} \mid a\in\Names\}$, then $\tilde{\rho},\tilde{\sigma}\in\ASC$ and there exists a unique
$\beta'\in\{\playerA:a,\playerB:a \mid a\in\Names\}$ such that $\tilde{\rho}'\tbpp\tilde{\sigma}'\tblts{\beta'}{\rho}'\tbpp{\sigma}'$ with ${\rho'},{\sigma'}\in\ASC$. Moreover, $\tilde{\rho}\tbpp\tilde{\sigma}\Lts{} {\rho}'\tbpp{\sigma}'$
\end{enumerate}
\end{lem}

\begin{proof} Immediate by definition of $\tbASC$ and $\tblts{}$.
\end{proof}

\begin{defi}
Let $\mathsf{T}\in\rts{\rho\tbpp\sigma}$ where ${\rho},{\sigma}\in\ASC$. We define the regular tree $\ff(\mathsf T)$ as follows:\\
let  $B=\{\playerA{:}a,\playerB{:}a \mid a\in\Names\}$ and $\Dual{B}=\{\playerA{:}\Dual{a},\playerB{:}\Dual{a} \mid a\in\Names\}$
\[
\begin{array}{ccc@{~~~~~~~~}l}

\ff\Big( ~\begin{array}{c}
             \stopA\tbpp\tilde{\sigma}\\
            \mid \\
           \Zero\tbpp\tilde{\sigma}
            \end{array} ~\Big) & = &\stopA\pp\tilde{\sigma}
\\[6mm]


\ff  \Big( \begin{array}{c}
            {\rho}\tbpp{\sigma} \\
            / \cdots\backslash \\
            \buf{\Dual{a}_1}.\rho_1\tbpp{\sigma} \cdots ~~\buf{\Dual{a}_n}.\rho_n\tbpp{\sigma}\\
            \mid ~~\cdots~~\mid \\
             ~~\mathsf{T}_1 \cdots \mathsf{T}_n
             \end{array}~ \Big)  & = 
 &  \begin{array}{c}
            {\rho}\pp{\sigma} \\
            / \cdots\backslash \\
            \ff(\mathsf{T}_1) \cdots \ff(\mathsf{T}_n)
             \end{array} \\[-2mm]
& & \multicolumn{2}{l}{ \begin{array}{l} \mbox{if ${\rho}\tbpp{\sigma} \tblts{\playerA:\Dual{a}_k} \buf{\Dual{a}_k}\rho_k \tbpp \tilde\sigma$ for $k = 1,\ldots,n$,}\\
\mbox{and similarly if 
${\rho}\tbpp{\sigma} \tblts{\playerA:\Dual{a}_k} {\tilde \rho} \tbpp  \buf{\Dual{a}_k} \sigma_k$.}
\end{array}
}\\[6mm]


\ff\Big(  
\begin{array}{c}
            {\rho}\tbpp{\sigma} \\
            \mid  \\
            \mathsf{T'}
             \end{array}   ~ \Big) & = 
 &  
\begin{array}{c}
            {\rho}\pp{\sigma} \\
            \mid  \\
            \ff(\mathsf{T'})
             \end{array}  
&  \mbox{ if } ~~\begin{array}{c}
                       \exists a\in\Names.\, {\rho}\tbpp{\sigma} \tblts{\playerC:a}\\
                       \mbox{ or }\\
                     \exists \beta\in B.\, {\rho}\tbpp{\sigma} \tblts{\beta}
                       \end{array}
\\[8mm]
\ff\Big(
\begin{array}{c}
            \tilde{\rho}\tbpp\tilde{\sigma} \\
            \mid  \\
            \xmark
             \end{array}
~ \Big) & = 
 &  
\begin{array}{c}
            \tilde{\rho}\pp\tilde{\sigma} \\
            \mid  \\
            \xmark
             \end{array} &  \mbox{else}
\end{array}  
\]
\end{defi}


\begin{lem}
\label{lem:tbpptopp}
Let $\mathsf{T}\in\rts{\rho\tbpp\sigma}$ where ${\rho},{\sigma}\in\ASC$ and such that
all its leaves are of the form $\Zero\tbpp\tilde{\sigma}$.
Then all leaves of $\ff(\mathsf T)$ are of the form $\stopA\pp\tilde{\sigma}$. 
\end{lem}
\begin{proof}
By definition of  $\ff$\,.
\end{proof}

We define now a function $\fftb : \rts{\rho\pp\sigma} \rightarrow  \rts{\rho\tbpp\sigma}$.
The following facts are immediate by definition.

\begin{lem}
\label{lem:Ltsltsoplustau}
Let ${\rho},{\sigma},{\rho_h},{\sigma_h}\in\ASC$ have the form
as in (\ref{lem:coinductiveChar-3}) or
(\ref{lem:coinductiveChar-4}) of  \ref{def:coACRel}, where $h\in I$.\\
Then
 ${\rho}\pp{\sigma}\tblts{X{:}\Dual{a}}\tilde{\rho}\pp\tilde{\sigma}\tblts{\Dual{X}{:}a}{\rho_h}\pp{\sigma_h}$\\ for some $\tilde{\rho},\tilde{\sigma}\in\tbASC$, $a\in\Names$ and $X\in\Set{\playerA,\playerB}$,
where $\Dual{X}=\playerA$ if $X=\playerB$ and $\Dual{X}=\playerB$ if $X=\playerA$.
%
\end{lem}

\begin{lem}
\label{lem:Ltsltsoplustaubis}
Let ${\rho},{\sigma},{\rho_h},{\sigma_h}\in\ASC$ have the form
as in (\ref{lem:coinductiveChar-2}) of  \ref{def:coACRel}, where $h\in I$.\\
Then
 ${\rho}\pp{\sigma}\tblts{\playerC{:}a}{\rho_h}\pp{\sigma_h}$  for some $a\in\Names$.
\end{lem}


\begin{defi}
Let $\mathsf{T}\in\rts{\rho\pp\sigma}$ where ${\rho},{\sigma}\in\ASC$. We define the regular tree $\fftb(\mathsf T)$ as follows:\\
\[
\begin{array}{ccl@{~~~~~~~~}l}

\fftb(\stopA\pp\tilde{\sigma}) & = & 
            \begin{array}{c}
             \stopA\tbpp\tilde{\sigma}\\
            \mid \\
           \Zero\tbpp\tilde{\sigma}
            \end{array}
\\[6mm]

\fftb  \Big( \begin{array}{c}
            {\rho}\pp{\sigma} \\
            / \cdots\backslash \\
            \mathsf{T}_1 \cdots ~~\mathsf{T}_n
             \end{array}~ \Big)  & = 
 &  \begin{array}{c}
            {\rho}\tbpp{\sigma} \\
            / \cdots\backslash \\
        \tilde{\rho}_1\tbpp\tilde{\sigma}_1 \cdots \tilde{\rho}_n\tbpp\tilde{\sigma}_n\\  
            \mid ~~~~\cdots~~~~\mid \\
            \fftb(\mathsf{T}_1) \cdots \fftb(\mathsf{T}_n)
             \end{array}
 ~~~~~~ \mbox{ if $\rho$ and $\sigma$ are as in (\ref{lem:coinductiveChar-3}) or
(\ref{lem:coinductiveChar-4}) of  \ref{def:coACRel}}
\\[6mm]
& &
\mbox{where }\Set{\tilde{\rho}_i\tbpp\tilde{\sigma}_i}_{i=1..n} = \Set{\tilde{\rho}'\tbpp\tilde{\sigma}' ~\mid~{\rho}\tbpp{\sigma} \tblts{\beta}\tilde{\rho}'\tbpp\tilde{\sigma}', \beta\in\Dual{B}}\\
& &  \text{with}~ \Dual{B}=\{\playerA{:}\Dual{a},\playerB{:}\Dual{a} \mid a\in\Names\}
\\[6mm]

\fftb\Big(  
\begin{array}{c}
            {\rho}\pp{\sigma} \\
            \mid  \\
            \mathsf{T'}
             \end{array}   ~ \Big) & = 
 &  
\begin{array}{c}
            {\rho}\tbpp{\sigma} \\
            \mid  \\
            \fftb(\mathsf{T'})
             \end{array}  
 ~~~~~~  \mbox{if $\rho$ and $\sigma$ are as in (\ref{lem:coinductiveChar-2})
 of  \ref{def:coACRel}}
\\[8mm]
\fftb\Big( \begin{array}{c}
              \tilde{\rho}\pp\tilde{\sigma} \\
              \mid  \\
              \xmark
             \end{array} ~ \Big) & = 
 &   
          \begin{array}{c}
            \tilde{\rho}\tbpp\tilde{\sigma} \\
            \mid  \\
            \xmark
             \end{array} 
\end{array} 
\]
\end{defi}


\begin{lem}
\label{lem:pptbpp}
Let $\mathsf{T}\in\rts{\rho\pp\sigma}$ where ${\rho},{\sigma}\in\ASC$ and such that
all its leaves are of the form $\stopA\pp\tilde{\sigma}$.
Then all the leaves of $\fftb(\mathsf T)$ are of the form $\Zero\tbpp\tilde{\sigma}$.
\end{lem}
\begin{proof}
By definition of  $\fftb$\,.
\end{proof}

The following immediately descends from Lemmas \ref{lem:tbpptopp} and \ref{lem:pptbpp}.

\begin{thm}
\label{lem:toutoft}
$\rho, \sigma \in \ASC (\subset\tbASC)$.
\begin{enumerate}
\item
Let $\mathsf{T}\in\rts{{\rho}\tbpp{\sigma}}$
such that all its leaves are of the form $\Zero\tbpp\tilde{\sigma}$, then there exists $\mathsf{T'}\in\rts{{\rho}\pp{\sigma}}$ such that
all its leaves are of the form $\stopA\pp\tilde{\sigma}$.
\item
Let $\mathsf{T}\in\rts{{\rho}\pp{\sigma}}$
such that all its leaves are of the form $\stopA\pp\tilde{\sigma}$, then there exists $\mathsf{T'}\in\rts{{\rho}\tbpp{\sigma}}$ such that
all its leaves are of the form $\Zero\tbpp\tilde{\sigma}$.
\end{enumerate}
\end{thm}
\noindent
We then get Theorem \ref{th:complyequivtbcomply} as a corollary of Lemmas \ref{lem:equivtb1}, \ref{lem:equivtb2} and Theorem \ref{lem:toutoft}.

%
%

\section{Proof of Main Theorem I (Theorem \ref{th:complyAwinstrat})}
\label{appendix:mainthmI}

\subsection{Proof of $\ref{th:complyAwinstrat-1} \Iff \ref{th:complyAwinstrat-2}$ of Theorem \ref{th:complyAwinstrat} ( $\complyR\ =\ \ACRel$).}\hfill
\label{subsect:rbkiffaffect}


In order to prove the equivalence of items \ref{th:complyAwinstrat-1} and
\ref{th:complyAwinstrat-2} of the Main Theorem I (Theorem \ref{th:complyAwinstrat}),
that is 
$$\complyR\ =\ \ACRel,$$
we need to prove that $\complyR$ satifies the properties 
in Lemma \ref{lem:stack-len}.

\subsection{Proof of Lemma \ref{lem:stack-len} (\cite{BDLdL15})  (Rollback properties)}
\footnote{The  proof of Lemma \ref{lem:stack-len} is from  the workshop paper \cite{BDLdL15}.
We restate it here for the reader's convenience.}\hfill\\
\label{subsec:rollprop}
{\bf Lemma \ref{lem:stack-len}(\ref{lem:stack-len-1})}
 If
	$\np{\emptystack}{\rho} \pp \np{\emptystack}{\sigma} \ored{*} \np{\vec{\delta}}{\rho'} \pp \np{\vec{\gamma}}{\sigma'} \not\!\!\ored{}$, then	 $\vec{\delta} = \vec{\gamma} = \emptystack$.
\proof
Clearly $\np{\vec{\delta}}{\rho'} \pp \np{\vec{\gamma}}{\sigma'} \not\!\!\ored{}$ implies either $\vec{\delta} =\emptystack$ or $\vec{\gamma} = \emptystack$.
Observe that:
\begin{itemize}
\item rule  $(\CommRule)$ adds one element to both stacks;
\item rule  $(\tau)$ does not modify both stacks;
\item rule  $(\RbkRule)$ removes one element from both stacks.
\qed
\end{itemize}

\noindent
{\bf Lemma \ref{lem:stack-len}(\ref{lem:stack-len-2})} If
$\np {\vec{\delta}}{\rho} \complyR \np {\vec{\gamma}}{\sigma}$, then 
	 $\np {\vec{\delta'} \cons \vec{\delta}}{\rho} \complyR \np {\vec{\gamma'} \cons \vec{\gamma}}{\sigma}$
	 for all $\vec{\delta'}$ , $\vec{\gamma'}$. 
\proof
It suffices to show that
\[\np {\vec{\delta}}{\rho} \complyR \np {\vec{\gamma}}{\sigma} \implies 
	\np {\rho' \cons \vec{\delta}}{\rho} \complyR \np {\vec{\gamma}}{\sigma} \text{ and }
	\np { \vec{\delta}}{\rho} \complyR \np {\sigma' \cons \vec{\gamma}}{\sigma}
\]
which we prove by contraposition.

Suppose that $\np {\rho' \cons \vec{\delta}}{\rho} \ncomplyR \np {\vec{\gamma}}{\sigma}$; then
\[\np {\rho' \cons \vec{\delta}}{\rho} \pp \np {\vec{\gamma}}{\sigma} \ored{*} \np {\vec{\delta}'}{\rho''} \pp \np {\vec{\gamma}'}{\sigma''}
\not\!\!\ored{} \text{ and } \rho''\neq \stopA
\]
If $\rho'$ is never used, then  $\vec{\delta}' = \rho' \cons \vec{\delta}''$ and $\vec{\gamma}'=\emptystack$, so that
we get\[\np {\vec{\delta}}{\rho} \pp \np {\vec{\gamma}}{\sigma} 
\ored{*} \np {\vec{\delta}''}{\rho''} \pp \np \emptystack{\sigma''} \not\!\!\ored{}\]
Otherwise we have that
\[\np {\rho' \cons \vec{\delta}}{\rho} \pp \np {\vec{\gamma}}{\sigma} \ored{*} \np {\rho'}{\rho''} \pp \np {\vec{\gamma}'}{\sigma''}
\ored{} \np\emptystack\rho'\pp\np {\vec{\gamma}''}{\sigma'''}\]
and we assume that $\ored{*}$ is the shortest such reduction.
It follows that $\rho''\neq \stopA $. By the minimality assumption about the length of $\ored{*}$ we know
that $\rho'$ neither has been restored by some previous application of rule $(\rlbk)$, nor pushed back into the stack before. We get
\[\np {\vec{\delta}}{\rho} \pp \np {\vec{\gamma}}{\sigma} \ored{*} \np {\emptystack}{\rho''} \pp \np {\vec{\gamma}''}{\sigma''} \not\!\!\ored{}\]
In both cases we conclude that $\np {\vec{\delta}}{\rho} \ncomplyR \np {\vec{\gamma}}{\sigma}$ as desired. 

\medskip
\noindent
Similarly we can show that $\np { \vec{\delta}}{\rho} \ncomplyR \np {\sigma' \cons \vec{\gamma}}{\sigma} \implies
\np {\vec{\delta}}{\rho} \ncomplyR \np {\vec{\gamma}}{\sigma}$.
\qed



\noindent
We can now show that the rollback compliance and the relation $\ACRel$ do coincide. \\


\begin{lem}\label{lem:coinductiveCharRbk}
We have ${\rho} \complyR {\sigma}$ if and only if one of the following conditions holds:
\begin{enumerate}
\item \label{lem:coinductiveChaRbkr-1} $\rho = \stopA$;
\item \label{lem:coinductiveCharRbk-2} $\rho = \sum_{i\in I}\alpha_i.\rho_i$, $\sigma = \sum_{j\in J}\Dual{\alpha}_j.\sigma_j$ and 
	$\exists k \in I \cap J.\; {\rho_k} \complyR {\sigma_k}$;
\item \label{lem:coinductiveCharRbk-3} $\rho = \bigoplus_{i\in I}\Dual{\alpha}_i.\rho_i$, $\sigma = \sum_{j\in J}\alpha_j.\sigma_j$,
	$I\subseteq J$ and $\forall k \in I. \; {\rho_k} \complyR {\sigma_k}$;
\item \label{lem:coinductiveCharRbk-4} $\rho = \sum_{i\in I}\Dual{\alpha}_i.\rho_i$, $\sigma = \bigoplus_{j\in J}\alpha_j.\sigma_j$,
	$I\supseteq J$ and $\forall k \in J. \; {\rho_k} \complyR {\sigma_k}$.
\end{enumerate}
\end{lem}

\proof
($\Leftarrow$)\ Immediate. \\
($\Rightarrow$)\footnote{This proof of the direction of the Lemma was presented in the workshop paper \cite{BDLdL15}. We restate it here for the reader's convenience.}
We prove this
by contraposition and by cases according to the possible shapes of $\rho$ and
$\sigma$ in the conditions of Definition \ref{def:ACRel}.
Suppose $\rho = \sum_{i\in I}\alpha_i.\rho_i$, $\sigma = \sum_{j\in J}\Dual{\alpha}_j.\sigma_j$, $I \cap J=\set{k_1,\ldots,k_n}$ and 
$ \rho_{k_i} \ncomplyR \sigma_{k_i}$ for $1\leq i\leq n$. Then we get 
\[\np\emptystack\rho_{k_i} \pp \np\emptystack\sigma_{k_i} \ored{*} \np{\vec{\delta}_i}{\rho'_i} \pp \np{\vec{\gamma}_i}{\sigma'_i} \not\!\!\ored{}\]
for $1\leq i\leq n$, where $\rho'_i\neq\stopA$ and $\vec{\delta}_i = \vec{\gamma}_i = \emptystack$ by Lemma~\ref{lem:stack-len}. This implies 
\[\np{\mbox{\small$\sum$}_{i\in I\setminus\Set{k_1}}\alpha_i.\rho_i}\rho_{k_1} \pp \np{\mbox{\small$\sum$}_{j\in J\setminus\Set{k_1}}\Dual{\alpha}_j.\sigma_j}\sigma_{k_1} \ored{*}
 \np{\mbox{\small$\sum$}_{i\in I\setminus\Set{k_1}}\alpha_i.\rho_i}{\rho'_1} \pp \np{\mbox{\small$\sum$}_{j\in J\setminus\Set{k_1}}\Dual{\alpha}_j.\sigma_j}{\sigma'_1}\]
Let $I'=I\setminus J$ and $J'=J\setminus I$. We can reduce $\np\emptystack\rho \pp \np\emptystack\sigma$ only as follows:
\[\begin{array}{llll}
\np\emptystack\rho \pp \np\emptystack\sigma & \ored{} &
\np{\sum_{i\in I\setminus\Set{k_1}}\alpha_i.\rho_i}\rho_{k_1} \pp \np{\sum_{j\in J\setminus\Set{k_1}}\Dual{\alpha}_j.\sigma_j}\sigma_{k_1} & \mbox{by $(\CommRule)$}\\
&\ored{*}&
 \np{\sum_{i\in I\setminus\Set{k_1}}\alpha_i.\rho_i}{\rho'_1} \pp \np{\sum_{j\in J\setminus\Set{k_1}}\Dual{\alpha}_j.\sigma_j}{\sigma'_1} \\
 &\ored{}&
 \np\emptystack{\sum_{i\in I\setminus\Set{k_1}}\alpha_i.\rho_i} \pp \np\emptystack{\sum_{j\in J\setminus\Set{k_1}}\Dual{\alpha}_j.\sigma_j}
 	& \mbox{by $(\RbkRule)$}\\
	&~~\vdots&~\qquad\qquad\qquad\qquad\quad\vdots\\
	&\ored{*}&
 \np{\sum_{i\in I'}\alpha_i.\rho_i}{\rho'_n} \pp \np{\sum_{j\in J'}\Dual{\alpha}_j.\sigma_j}{\sigma'_n} \\
 &\ored{}&
 \np\emptystack{\sum_{i\in I'}\alpha_i.\rho_i} \pp \np\emptystack{\sum_{j\in J'}\Dual{\alpha}_j.\sigma_j}
 	& \mbox{by $(\RbkRule)$}
\end{array}\]
and $ \np\emptystack{\sum_{i\in I'}\alpha_i.\rho_i} \pp \np\emptystack{\sum_{j\in J'}\Dual{\alpha}_j.\sigma_j}$ is stuck since $I'\cap J'=\emptyset$.

\bigskip

Suppose $\rho = \bigoplus_{i\in I}\Dual{\alpha}_i.\rho_i$ and $\sigma = \sum_{j\in J}\alpha_j.\sigma_j$. If $I\not\subseteq J$ let
$k\in I\setminus J$; then we get
\[\begin{array}{llll}
\np\emptystack\rho \pp \np\emptystack\sigma & \ored{} &
\np{\emptystack}{\Dual{\alpha}_k.\rho_k}\pp \np\emptystack{\sigma} & \mbox{by $(\tau)$} \\
& \not\!\!\ored{} &
\end{array}\]
Otherwise $I\subseteq J$ and $ {\rho_k} \ncomplyR {\sigma_k}$ for some $k\in I$. By reasoning as above we have
\[\np\emptystack\rho_k \pp \np\emptystack\sigma_k \ored{*} \np\emptystack{\rho'} \pp \np\emptystack{\sigma'} \not\!\!\ored{}\]
and 
\[  \np{\circ}{\rho_k}\pp \np{\mbox{\small$\sum$}_{j\in J\setminus\Set{k}}\alpha_j.\sigma_j}{\sigma_k} 
 \ored{*} \np{\circ}{\rho'}\pp \np{\mbox{\small$\sum$}_{j\in J\setminus\Set{k}}\alpha_j.\sigma_j}{\sigma'}\]
which imply
\[\begin{array}{llll}
\np\emptystack\rho \pp \np\emptystack\sigma 
& \ored{} &
\np{\emptystack}{\Dual{\alpha}_k.\rho_k}\pp \np\emptystack{\sigma} & \mbox{by $(\tau)$} \\
&\ored{} & \np{\circ}{\rho_k}\pp \np{\sum_{j\in J\setminus\Set{k}}\alpha_j.\sigma_j}{\sigma_k} & \mbox{by $(\CommRule)$} \\
& \ored{*} & \np{\circ}{\rho'}\pp \np{\sum_{j\in J\setminus\Set{k}}\alpha_j.\sigma_j}{\sigma'} &\\
& \ored{} & \np\emptystack\circ \pp \np\emptystack{\sum_{j\in J\setminus\Set{k}}\alpha_j.\sigma_j} & \mbox{by $(\RbkRule)$} \\
& \not\!\!\ored{} &
\end{array}\]
In both cases we conclude that ${\rho} \ncomplyR {\sigma}$.

\noindent
The proof in case  $\rho = \sum_{i\in I}\Dual{\alpha}_i.\rho_i$, $\sigma = \bigoplus_{j\in J}\alpha_j.\sigma_j$ is similar.
\qed

From Lemma \ref{lem:coinductiveCharRbk} and definition of $\ACRel$ we  immediately have
$\complyR\ =\ \ACRel$.



\subsection{Proof of $\ref{th:complyAwinstrat-1} \Iff \ref{th:complyAwinstrat-3}$ of Theorem \ref{th:complyAwinstrat}}
\label{subsec:tbstratequiv}
Since 
(by Theorem \ref{th:complyequivtbcomply}) we have that $\ACRel\ =\ \complyTB$,
 it is enough to show that 
\begin{equation}
\label{eq:aiffb}
\rho \complyTB \sigma \mbox{\ \em if and only if \
there exists a winning strategy for player $\playerC$ in $\game_{\rho\pp\sigma}$} 
\end{equation}

We recall that a strategy $\Sigma$ is {\em univocal} if $|\Sigma(\vec{e})| \leq 1$ for all $\vec{e}$.

We begin by (\ref{eq:aiffb})($\Leftarrow$)

\begin{defi}\hfill
\label{def:treefromstrat}
\begin{enumerate}
\item
Let $\tilde\rho,\tilde\sigma\in\tbASC$, let $\Sigma$ be a univocal strategy for player $\playerC$ and let $\vec{e}$ be a play.
We define the tree $\rtaux{\Sigma}{\tilde\rho\tbpp\tilde\sigma}{\vec{e}}$ as follows:
\[
\begin{array}{rcl@{~~~~~~~~}l}

\rtaux{\Sigma}{\stopA\tbpp\tilde{\sigma}}{\vec{e}} & = 
&  ~\begin{array}{c}
             \stopA\tbpp\tilde{\sigma}\\
            \mid \\
           \Zero\tbpp\tilde{\sigma}
            \end{array} ~\\[6mm]

\rtaux{\Sigma}{\tilde{\rho}\tbpp\tilde{\sigma}}{\vec{e}} & = 
 &  \begin{array}{c}
            \tilde{\rho}\tbpp\tilde{\sigma} \\
            / \cdots\backslash \\
            \mathsf{T}_1 \cdots ~~\mathsf{T}_n
             \end{array}\\
& & \mbox{ if } ~~\exists \beta\in B.\, \tilde{\rho}\tbpp\tilde{\sigma} \tblts{\beta}\\
 & &
\mbox{ where } 
\mathsf{T}_1 = \rtaux{\Sigma}{\tilde{\rho}_1\tbpp\tilde{\sigma}_1}{\vec{e}\beta_1}, \ldots ,\mathsf{T}_n = \rtaux{\Sigma}{\tilde{\rho}_n\tbpp\tilde{\sigma}_n}{\vec{e}\beta_n}\\
 & &
\mbox{ with } \Set{\tilde{\rho}_i\tbpp\tilde{\sigma}_i}_{i=1..n} = \Set{\tilde{\rho}'\tbpp\tilde{\sigma}' ~\mid~\tilde{\rho}\tbpp\tilde{\sigma} \tblts{\beta}\tilde{\rho}'\tbpp\tilde{\sigma}', ~\beta\in B }\\
 & &
\mbox{ and }  \tilde{\rho}\tbpp\tilde{\sigma} \tblts{\beta_i}\tilde{\rho}_i\tbpp\tilde{\sigma}_i
\\[6mm]

\rtaux{\Sigma}{\tilde{\rho}\tbpp\tilde{\sigma}}{\vec{e}} & = 
 &  
\begin{array}{c}
            \tilde{\rho}\tbpp\tilde{\sigma} \\
            \mid  \\
            \mathsf{T}
             \end{array}   
\\
&  & \mbox{ if } ~~\tilde{\rho}\tbpp\tilde{\sigma} \tblts{\playerC:a} \tilde{\rho}'\tbpp\tilde{\sigma}'\\
& &
\mbox{ where }  \Sigma({\vec{e}}) = (k, \playerC{:}a) \mbox{ for some $k$ },\\
& &
\mbox{ and where } \mathsf{T} = \rtaux{\Sigma}{\tilde{\rho}'\tbpp\tilde{\sigma}'}{\vec{e}\Sigma({\vec{e}})}\\[6mm]

\rtaux{\Sigma}{\tilde{\rho}\tbpp\tilde{\sigma}}{\vec{e}} & = 
&   
\begin{array}{c}
            \tilde{\rho}\tbpp\tilde{\sigma} \\
            \mid  \\
            \xmark
             \end{array}

~~~~~~~~  \mbox{ otherwise}\\[2mm]

\end{array}
\]

\item
Let $\rho,\sigma\in\ASC$ and let $\Sigma$ be a univocal strategy for player $\playerC$ for the game $\game_{\rho\!\pp\!\sigma}$.
We define 
$$\rtsigma{\Sigma}{\rho\!\tbpp\!\sigma} \ByDef \rtaux{\Sigma}{\rho\!\tbpp\!\sigma}{\varepsilon}$$

\end{enumerate}
\end{defi}

\begin{lem}
\label{lem:treefromstrat}
Let $\Sigma$ be a univocal winning strategy for player $\playerC$ for the game $\game_{\rho\!\pp\!\sigma}$. Then all the leaves of $\rtsigma{\Sigma}{\rho\!\tbpp\!\sigma}$ are of the form $\Zero\tbpp\tilde{\sigma}$.
\end{lem}

We hence get (\ref{eq:aiffb})($\Leftarrow$)
by Lemmas \ref{lem:treefromstrat} and \ref{lem:equivtb1}.

\paragraph{We can now proceed with (\ref{eq:aiffb})($\Rightarrow$)}

\begin{defi}
\label{def:rtsedgeslabelling}
Let $\mathsf{T}\in\rts{\rho\!\tbpp\!\sigma}$ such that all its leaves are of the form $\Zero\tbpp\tilde{\sigma}$.
\begin{enumerate}

\item
The moves-labelled tree of $\mathsf{T}$, dubbed $\movlabtree{\mathsf{T}}$ is obtained out of $\mathsf{T}$ by labelling its edges as follows:
if $E$ is an edge from a node $N$ to a child $M$ of its, we label it by $(L(N),\beta)$, where $\beta$ is such that $N\tblts{\beta}M$ and $L(N)$ is the level of $N$.

\item
Given a finite path $\bm p$ in $\movlabtree{\mathsf{T}}$ starting from the root, we define
$\finplay{\vec{p}}$ as the sequence of labels of the edges of the path.

\item
Given a finite path $\bm p$ in $\movlabtree{\mathsf{T}}$ starting from the root,
we define 
$$\nextmove{\bm p} \ByDef \left\{ \begin{array}{l@{\hspace{8mm}}l}
                                     (n,\beta)  & if ~~(*)\\[2mm]  
                                      \emptyset & \mbox{otherwise}
                                        \end{array} \right.
$$
(*) the last node $N$ in $\bm p$ is of the form $\sum_{i\in I}\alpha_i.\tilde\rho_i \tbpp \sum_{j\in J}\Dual{\alpha}_j.\tilde\sigma_j$ or $\stopA \tbpp \tilde\sigma$ and $(n,\beta)$ is the label of the only edge out of $N$. 
\end{enumerate}
\end{defi}

\begin{lem}
\label{lem:playsfrompaths}
Let $\mathsf{T}\in\rts{\rho\!\tbpp\!\sigma}$ such that all its leaves are of the form $\Zero\tbpp\tilde{\sigma}$ and let $\bm p$ be a finite path in $\movlabtree{\mathsf{T}}$.
Then $\finplay{\vec{p}}$ is a finite play
 of  $\game_{\rho\!\pp\!\sigma}$
\end{lem}
\proof
By Definition \ref{def:rtsedgeslabelling} and by definition of play of $\game_{\rho\!\pp\!\sigma}$.
\qed

\begin{defi}
\label{defi:stratfromtree}
Let $\mathsf{T}\in\rts{\rho\!\tbpp\!\sigma}$ such that all its leaves are of the form $\Zero\tbpp\tilde{\sigma}$. We define the strategy $\Sigma_\mathsf{T}$ in $\game_{\rho\!\pp\!\sigma}$, for player $\playerC$, by

$$
\Sigma_\mathsf{T}(\vec{e}) \ByDef \left\{ \begin{array}{l@{\hspace{8mm}}l}
                                     \nextmove{\bm p}  & if ~~(*)\\[2mm]  
                                      \emptyset & \mbox{otherwise}
                                        \end{array} \right.
$$
(*) $\vec{e} = \finplay{\vec{p}}$ for some $\bm p$ which is a finite path in $\movlabtree{\mathsf{T}}$.
\end{defi}

\begin{lem}
\label{lem:stratfromtree}
Let $\mathsf{T}\in\rts{\rho\!\tbpp\!\sigma}$ such that all its leaves are of the form $\Zero\tbpp\tilde{\sigma}$. Then $\Sigma_\mathsf{T}$ is a univocal winning strategy for player $\playerC$ for the game $\game_{\rho\!\pp\!\sigma}$.
\end{lem}

We hence get (\ref{eq:aiffb})($\Rightarrow$) from the above Lemma and \ref{lem:equivtb1}.

%

\subsection{Proof of $\ref{th:complyAwinstrat-1} \Iff \ref{th:complyAwinstrat-4} $ 
 of Theorem \ref{th:complyAwinstrat}}
\label{subsec:affectorchequiv}

In the following proof we could have proceed by using a lemma similar to \ref{lem:coinductiveCharRbk}. However, in order to get also a correspondence between 
orchestrators and derivations to be used for the proof of the Main Theorem II (\ref{th:derstratorchequiv}), we proceed by providing two formal systems axiomatizing the relation of orchestrated compliance.

\subsubsection{
(Formal systems and synthesis for $\complyO$)}
\label{appendix:soundcomplsynth}

In the following, orchestrators are considered as explicit recursive terms rather than as possibly infinite
regular trees. 
We first define a formal system $\derOrch$ in which the relation of derivability characterizes the relation $\complyO$. 
In System $\derOrch$ the relation $\complyO$ is the intended interpretation of the symbol  $\complyOF$.

\begin{defi}[Formal System for Orchestrated Compliance]\label{def:formalOrchCompl}
An environment $\Gamma$ is a finite set of expressions of the form $f : \delta \complyOF \gamma$ where
$\delta,\gamma\in\ASC$ and $f\in\Orch$. The judgments of System $\derOrch$ are expressions of the form
$\Gamma\der f : \rho \complyOF \sigma$. The axioms and rules of $\derOrch$ are as in Figure
\ref{fig:forsystorchder}.
\end{defi}

\begin{figure}
\hrule
\vspace{2mm}
\[\begin{array}{c@{\hspace{8mm}}c}
\mbox{\scriptsize$(\CkptcomplAx)$}:\Inf{}
	{\Gamma\derOrch \stopf : \stopA \complyOF \sigma}
&
\mbox{\scriptsize$(\CkptcomplHyp)$}:\Inf{}
{\Gamma, f:\rho\complyOF\sigma \derOrch f:\rho\complyOF\sigma }
\end{array}\]
\vspace{0mm}
\[\begin{array}{l}
(\mbox{\footnotesize$+\cdot+$}):
\Inf{
     \Gamma'
     	\derOrch f:\rho
    	  \complyOF
		\sigma}
{ \Gamma\derOrch \orchAct{\Dual{\alpha}}{\alpha}^+.f : \alpha.\rho+\rho'\complyOF\Dual{\alpha}.\sigma+\sigma'}\\[4mm]
\mbox{\footnotesize where $\Gamma' = \Gamma, \langle\Dual{\alpha},{\alpha}\rangle^+.f: \alpha.\rho+\rho'\complyOF\Dual{\alpha}.\sigma+\sigma'$}
\\[8mm]
(\mbox{\footnotesize$\oplus\cdot+$}):
\Inf{
    \forall i\in I.~ \Gamma' \derOrch 
    	f_i:\rho_i
    	\complyOF
    		\sigma_i}
{ \Gamma\derOrch  \bigvee_{i\in\Set{1..n}}\langle{a_i},\Dual{a}_i\rangle.f_{a_i}: \mbox{\small $\bigoplus$}_{i\in I} \Dual{a}_i.{\rho}_i\complyOF
    	\mbox{\small $\sum$}_{j\in I\cup J} a_j.{\sigma}_j} \\[4mm]
\mbox{\footnotesize where $\Gamma' = \Gamma, \bigvee_{i\in\Set{i=1..n}}\langle{a_i},\Dual{a}_i\rangle.f_i:
 \mbox{\footnotesize $\bigoplus$}_{i\in I} \Dual{a}_i.{\rho}_i\complyOF
    	\mbox{\footnotesize $\sum$}_{j\in I\cup J} a_j.{\sigma}_j $}
\\[8mm]
(\mbox{\footnotesize$+\cdot\oplus$}):
\Inf{
    \forall i\in I.~ \Gamma \derOrch 
    	f_i:\rho_i
    	\complyOF
    		\sigma_i}
{ \Gamma\derOrch \bigvee_{i\in\Set{i=1..n}}\langle\Dual{a}_i,{a_i}\rangle.f_i:\mbox{\small $\sum$}_{j\in I\cup J} a_j.{\sigma}_j\complyOF \mbox{\small $\bigoplus$}_{i\in I} \Dual{a}_i.{\rho}_i    }\\[4mm]
\mbox{\footnotesize where $\Gamma' = \Gamma, \bigvee_{i\in\Set{1..n}}\langle\Dual{a}_i,{a_i}\rangle.f_i:
  \mbox{\footnotesize $\sum$}_{j\in I\cup J} a_j.{\sigma}_j\complyOF \mbox{\footnotesize $\bigoplus$}_{i\in I} \Dual{a}_i.{\rho}_i $}	\\[2mm]
\end{array}\]
\caption{System $\derOrch$}
\label{fig:forsystorchder}
\vspace{2mm}
\hrule
\end{figure}

\begin{thm}[Soundness and Completeness of System $\derOrch$ w.r.t $\complyO$]
\label{th:scder}
$$ \derOrch f:\rho\complyOF\sigma ~~~~ \Iff ~~~~ f: \rho\complyO\sigma$$
\end{thm}
\begin{proof}
The proof can be developed along the very same lines of the proofs for Proposition \ref{prop:soundnessACRel} and 
Lemma \ref{prop:provecorr} for what concerns, respectively, soundness and completeness.
\end{proof}

We provide now a formal system  $\derinfOrch$ equivalent to $\derOrch$.
In such a system, unproper (namely open) orchestrators can be used.
However we shall apply the system only for proper (namely closed) orchestrators.

The algorithm $\Synth$ corresponds to a proof-search algorithm in system $\derinfOrch$:
it synthesises,
given $\rho$ and $\sigma$, an orchestrator $f$ such that $f: \rho\complyO\sigma$, and hence
a univocal winning strategy for player $\playerC$ in the game $\game_{\rho\!\pp\!\sigma}$.

\begin{defi}[The equivalent System $\derinfOrch$]\label{def:formalOrchComplInf}
An environment $\Gamma$ is a finite set of expressions of the form $x : \delta \complyOF \gamma$ where
$\delta,\gamma\in\ASC$ and $x$ is an orchestrator variable. The judgments of System $\derOrch$ are expressions of the form
$\Gamma\der f : \rho \complyOF \sigma$, where $f$ is an orchestrator, possibly open. The axioms and rules of $\derinfOrch$ are as in Figure
\ref{fig:forsystinforchder}.
\end{defi}

\begin{figure}
\hrule
\vspace{2mm}
\[\begin{array}{c@{\hspace{8mm}}c}
\mbox{\scriptsize$(\CkptcomplAx)$}:\Inf{}
	{\Gamma\derinfOrch \stopf : \stopA \complyOF \sigma}
&
\mbox{\scriptsize$(\CkptcomplHyp)$}:\Inf{}
{\Gamma, x:\rho\complyOF\sigma \derinfOrch x:\rho\complyOF\sigma }
\end{array}\]
\vspace{0mm}
\[\begin{array}{l}
(\mbox{\footnotesize$+\cdot+$}):
\Inf{
     \Gamma'
     	\derinfOrch f:\rho
    	  \complyOF
		\sigma}
{ \Gamma\derinfOrch \rec x.\orchAct{\Dual{\alpha}}{\alpha}.f : \alpha.\rho+\rho'\complyOF\Dual{\alpha}.\sigma+\sigma'}\\[4mm]
\mbox{\footnotesize where $\Gamma' = \Gamma, x: \alpha.\rho+\rho'\complyOF\Dual{\alpha}.\sigma+\sigma'$}
\\[8mm]
(\mbox{\footnotesize$\oplus\cdot+$}):
\Inf{
    \forall i\in I.~ \Gamma' \derinfOrch 
    	f_{a_i}:\rho_i
    	\complyOF
    		\sigma_i}
{ \Gamma\derinfOrch  \rec x.\bigvee_{\Set{a_i\mid i=1..n}}f_{a_i}:\mbox{\small $\bigoplus$}_{i\in I} \Dual{a}_i.{\rho}_i\complyOF
    	\mbox{\small $\sum$}_{j\in I\cup J} a_j.{\sigma}_j} \\[4mm]
\mbox{\footnotesize where $\Gamma' = \Gamma, x:
 \mbox{\footnotesize $\bigoplus$}_{i\in I} \Dual{a}_i.{\rho}_i\complyOF
    	\mbox{\footnotesize $\sum$}_{j\in I\cup J} a_j.{\sigma}_j $}
\\[8mm]
(\mbox{\footnotesize$+\cdot\oplus$}):
\Inf{
    \forall i\in I.~ \Gamma \derinfOrch 
    	f_{a_i}:\rho_i
    	\complyOF
    		\sigma_i}
{ \Gamma\derinfOrch \rec x. \bigvee_{\Set{a_i\mid i=1..n}}f_{a_i}: \mbox{\small $\sum$}_{j\in I\cup J} a_j.{\sigma}_j\complyOF \mbox{\small $\bigoplus$}_{i\in I} \Dual{a}_i.{\rho}_i    }\\[4mm]
\mbox{\footnotesize where $\Gamma' = \Gamma, \bigvee_{\Set{a_i\mid i=1..n}}f_{a_i}:
  \mbox{\footnotesize $\sum$}_{j\in I\cup J} a_j.{\sigma}_j\complyOF \mbox{\footnotesize $\bigoplus$}_{i\in I} \Dual{a}_i.{\rho}_i $}	\\[2mm]
\end{array}\]
\caption{System $\derinfOrch$}
\label{fig:forsystinforchder}
\vspace{2mm}
\hrule
\end{figure}

\begin{prop} 
\label{prop:derOrchderinfOrchequiv}
Let $f$ be a proper (closed) orchestrator.
$$\derOrch f:\rho\complyOF\sigma \hspace{3mm}\Iff \hspace{3mm}\derinfOrch f:\rho\complyOF\sigma$$
\end{prop}

Now we show how  to build an orchestrator $f$  and a derivation in $\derinfOrch$ such that $\derinfOrch f : \rho\complyF \sigma$ when
 a derivation of $\der \rho\complyF \sigma$ is given.

\begin{defi}\hfill
\label{def:derinforch}
\begin{enumerate}
\item
The partial functions
$\dertoderinfaux{\mbox{-},\mbox{-}}$ and  $\fofderaux{\mbox{-},\mbox{-}}$, from derivations in $\der$ and environments  to, respectively,  derivations in $\derinfOrch$ and (possibly open) orchestrators, are inductively and simultaneously defined as follows:
\medskip
\begin{description}

\item[$\Der = \Inf{}
	{\Gamma\der \stopA \complyF \sigma}~(\mbox{\footnotesize$\CkptcomplAx$})$]\hfill\\[2mm]
$\begin{array}{lcl}
\dertoderinfaux{\Der,\Gammainf} &=& \Inf{}
	{\Gammainf\derinfOrch \stopf : \stopA \complyOF \sigma}~\mbox{\scriptsize$(\CkptcomplAx)$}\\
\fofderaux{\Der,\Gammainf} &=& \stopf
\end{array}$

\vspace{2mm}
\item[$\Der =  
\Inf{}
{\Gamma, \rho\complyF\sigma \der \rho\complyF\sigma}
~(\mbox{\footnotesize$\CkptcomplHyp$})$]\hfill\\[2mm]
$\begin{array}{lcl}
\dertoderinfaux{\Der,\Gammainf} & = & \Inf{}
{\Gammainf \derinfOrch x:\rho\complyOF\sigma }~\mbox{\scriptsize (\CkptcomplHyp)}\\
\fofderaux{\Der,\Gammainf} &=& x
\end{array}$ ~~\Big\} if $x:\rho\complyOF\sigma \in \Gammainf$.

\vspace{2mm}
\item[$\Der =  
\Inf{
     \Der'}
{ \Gamma\der \alpha.\rho+\rho'\complyF\Dual{\alpha}.\sigma+\sigma'}~(\mbox{\footnotesize$+\cdot+$})
$]\hfill\\[2mm]
$\begin{array}{lcl}
\dertoderinfaux{\Der,\Gammainf} & = &
\Inf{\dertoderinfaux{\Der',(\Gammainf,x{:} \alpha.\rho+\rho'\complyOF\Dual{\alpha}.\sigma+\sigma')={\Gammainf}'}}
{\Gammainf\derinfOrch \rec x.\orchAct{\Dual{\alpha}}{\alpha}.\fofderaux{\Der',{\Gammainf}'} : \alpha.\rho+\rho'\complyOF\Dual{\alpha}.\sigma+\sigma'} ~ (\mbox{\footnotesize$+\cdot+$})\\

\fofderaux{\Der,\Gammainf} & = & \rec x.\orchAct{\Dual{\alpha}}{\alpha}.\fofderaux{\Der',{\Gammainf}'}\\
& & \mbox{\footnotesize where $x$  is a fresh variable.}
\end{array}$

\vspace{2mm}
\item[$\Der =  
\Inf{
    \forall i\in I.~ \Der_i}
{ \Gamma\der \mbox{\small $\bigoplus$}_{i\in I} \Dual{a}_i.{\rho}_i\complyF
    	\mbox{\small $\sum$}_{j\in I\cup J} a_j.{\sigma}_j} ~ (\mbox{\footnotesize$\oplus\cdot+$})
$]\hfill\\[2mm]
$\begin{array}{lcl}
\dertoderinfaux{\Der,\Gammainf} & = &
\Inf{
    \forall i\in I.~ \dertoderinfaux{\Der_i,(\Gammainf,x{:} \mbox{\small $\bigoplus$}_{i\in I} \Dual{a}_i.{\rho}_i\complyOF
    	\mbox{\small $\sum$}_{j\in I\cup J} a_j.{\sigma}_j)={\Gammainf}'}}
{ \Gamma\derinfOrch  \rec x.\bigvee_{\Set{a_i\mid i=1..n}}f_{a_i}:\mbox{\small $\bigoplus$}_{i\in I} \Dual{a}_i.{\rho}_i\complyOF
    	\mbox{\small $\sum$}_{j\in I\cup J} a_j.{\sigma}_j}  ~  (\mbox{\footnotesize$\oplus\cdot+$})\\

\fofderaux{\Der,\Gammainf} & = & \rec x.\bigvee_{\Set{a_i\mid i=1..n}}\fofderaux{\Der_i,{\Gammainf}'}\\
& & \mbox{{\footnotesize where $f_{a_i} = \fofderaux{\Der_i,{\Gammainf}'}$ and $x$  is a fresh variable.}}
\end{array}$\\

\vspace{2mm}
\item[$\Der =  
\Inf{
    \forall i\in I.~ \Der_i}
{ \Gamma\der\mbox{\small $\sum$}_{j\in I\cup J} a_j.{\sigma}_j\complyF \mbox{\small $\bigoplus$}_{i\in I} \Dual{a}_i.{\rho}_i    }(\mbox{\footnotesize$+\cdot\oplus$})$]\hfill\\[2mm]
$\begin{array}{lcl}
\dertoderinfaux{\Der,{\Gammainf}} & = &  
\Inf{
    \forall i\in I.~ \dertoderinfaux{\Der_i,(\Gammainf,x{:} \mbox{\small $\sum$}_{j\in I\cup J} a_j.{\sigma}_j\complyOF \mbox{\small $\bigoplus$}_{i\in I} \Dual{a}_i.{\rho}_i  )={\Gammainf}'}}
{ \Gamma\derinfOrch \rec x. \bigvee_{\Set{a_i\mid i=1..n}}f_{a_i}: \mbox{\small $\sum$}_{j\in I\cup J} a_j.{\sigma}_j\complyOF \mbox{\small $\bigoplus$}_{i\in I} \Dual{a}_i.{\rho}_i    }\\
\fofderaux{\Der,{\Gammainf}} & = & \rec x.\bigvee_{\Set{a_i\mid i=1..n}}\fofderaux{\Der_i,{\Gammainf}'}\\
 & & \mbox{{\footnotesize where $f_{a_i} = \fofderaux{\Der_i,{\Gammainf}'}$ and $x$  is a fresh variable.}}
\end{array}$\\

\end{description}
\item
The partial function
$\dertoderinf{\mbox{-}}$ from derivations in $\der$  to derivations in $\derinfOrch$, and the partial function
$\fofderaux{\mbox{-}}$ from derivations in $\der$ to (possibly open) orchestrators are defined by
$$\dertoderinf{\Der} = \dertoderinfaux{\Der,\emptyset}$$
$$\fofder{\Der} = \fofderaux{\Der,\emptyset}$$
\end{enumerate}
\end{defi}

\begin{lem}
\label{prop:derderorch-i}
Let $\Der :: \der \rho\complyF \sigma$. Then $\dertoderinf{\Der}$ and
$\fofder{\Der}$ are well-defined,  $\fofder{\Der}$ is a proper (i.e. closed) orchestrator and
$$ \dertoderinf{\Der} ::\ \derinfOrch \fofder{\Der}: \,\rho\complyOF \sigma$$
\end{lem}
\proof
By Induction.
\qed

\medskip
We hence get $\ref{th:complyAwinstrat-1} \Rightarrow \ref{th:complyAwinstrat-4}$
as an immediate consequence of Lemma \ref{prop:derderorch-i} above, Proposition \ref{prop:derOrchderinfOrchequiv} and Theorem \ref{th:scder}.\\

The  implication $\ref{th:complyAwinstrat-4} \Rightarrow \ref{th:complyAwinstrat-1}$
is instead an easy consequence of the observation that by erasing all orchestrators in a derivation
of $\derOrch f: \rho\complyF \sigma$ we get a derivation of $\der \rho\complyF \sigma$.

%


\section{Proof of Main Theorem II (\ref{th:derstratorchequiv}) \\
(Getting derivations, orchestrators and strategies out of each other)}\hfill
\label{appendix:mainthmII}

\noindent
We begin by providing a stratified version of orchestrated compliance.
\begin{defi}[Coinductive  orchestrated compliance]\label{def:coskipcompl}\hfill
\\Let $\Set{\complyOcok{k}}_{k \in \Nat}$ be the family of relations over $\Orch\times\ASC\times\ASC$ such that 

\begin{enumerate}
\item ${\complyOcok{0}} = \Orch\times\ASC\times\ASC$ and
\item $f:\rho~\complyOcok{k+1}~\sigma$ if either:
\begin{enumerate}
\item\label{def:coinductiv-compl1} $\rho = \stopA $; or
\item\label{def:coinductiv-compl3} $\rho \neq \stopA $, $\rho \pf{f} \sigma \Lts{} $, and $\rho \pf{f} \sigma \Lts{} \rho' \pf{f'} \sigma' $ implies $f':\rho'~\complyOcok{k}~\sigma'$.
\end{enumerate}
\end{enumerate}
Then we define $\complyOco = \bigcap_{k\in \Nat} \complyOcok{k}$.
\end{defi}

\begin{prop}\label{lem:orchcoinductiv-char}
The relation $\complyOco$ and the compliance relation $\complyO$ coincide, i.e.
 \[ \begin{array}{rcl}
f:\rho \complyOco \sigma &\Iff& f:\rho \complyO \sigma.
 \end{array} \]
\end{prop}

\begin{figure}[t]
\hrule
\vspace*{2mm}
\[ \begin{array}{rcl}
\Internal_{i\in I}\Dual{a}_i.\rho_i\tbpf{f} \sigma &\tbolts{\playerA:\Dual{a}_k} &\buf{\Dual{a}_k}\rho_k \tbpf{f} \sigma \hspace{6mm}(k\in I)\\[4mm]
\External_{i\in I} a_i.\rho_i \tbpf{\bigvee_{h\in H} \langle{\Dual{a}_h},{a_h}\rangle.f_h} \buf{\Dual{a}_k} \sigma &\tbolts{\playerA:a_k}& \rho_k \tbpf{f_k} \sigma \hspace{6mm}(k\in (I\cap H))
\end{array}
\]
\vspace{2mm}
\[ \begin{array}{rcl}
\rho\tbpf{f}  \Internal_{i\in I}\Dual{a}_i.\sigma_i&\tbolts{\playerB:\Dual{a}_k} & \rho \tbpf{f} \buf{\Dual{a}_k}\sigma_k \hspace{6mm}(k\in I)\\[4mm]
\buf{\Dual{a}_k} \rho\tbpf{\bigvee_{h\in H} \langle{a_h},{\Dual{a}_h}\rangle.f_h}  \External_{i\in I} a_i.\sigma_i  &\tbolts{\playerB:a_k}& \rho \tbpf{f_k} \sigma_k \hspace{6mm}(k\in (I\cap H))
\end{array}
\]
\vspace{2mm}
\[
\begin{array}{rcl@{\hspace{6mm}}rcl}
\Dual{a}.\rho + \rho' \tbpf{\langle{a},{\Dual{a}}\rangle^+.f'} a. \sigma + \sigma'  &\tbolts{\playerC:a}& \rho\tbpf{f'} \sigma
&
a.\rho + \rho' \tbpf{\langle{\Dual{a}},{a}\rangle^+.f'} \Dual{a}. \sigma + \sigma'  &\tbolts{\playerC:a}& \rho\tbpf{f'} \sigma\\[4mm]
\multicolumn{6}{c}{
\stopA\tbpf{f} \tilde{\rho} ~~ \tbolts{\playerC:\cmark}  ~~\Zero \tbpf{f} \tilde{\rho}
}
\end{array} \]
\caption{Turn-based operational semantics of orchestrated-configurations systems}\label{fig:tb-orchopsem}\label{fig:tborch-opsem}
\vspace*{2mm}
\hrule
\end{figure}

As done for the relation $\ACRel$, we provide an equivalent ''turn-based'' version of the 
relation $\complyO$.

\begin{defi}[Turn-based operational semantics of turn-based orchestrated configurations]\hfill
\label{def:tbLTS}
Let  $\tbAct = \{ \playerA,\playerB,\playerC\} \times (\Act \cup \Set{\cmark})$. In Figure \ref{fig:tborch-opsem} we define the LTS $\tbolts{}$ over
turn-based configurations, with labels in $\tbAct$.
\end{defi}

We define $\tbolts{}=\, \bigcup_{\beta\in \tbAct}\!\tbolts{\beta}$.

\begin{defi}[Turn-based orchestrated compliance $\complyTBO$]
\label{def:tbdisstrictcompl}
Let $f\in\Orch$ and $\tilde\rho,\tilde\sigma\in\tbASC$.
\begin{enumerate}
\item
\label{def:tbdisstrictcompl-i}
$f: \tilde\rho\complyTBO \tilde\sigma$ \hspace{1mm} if \hspace{1mm} 
\[ \begin{array}{rcl}
\tilde\rho \tbpf{f} \tilde\sigma \tbolts{}^{\hspace{-2mm}*}~ \tilde\rho' \tbpf{f'} \tilde\sigma' \nottbolts{} &\mbox{implies}& \tilde\rho'=\mbox{\em $\stopA$}.
 \end{array} \]
\item
$ \begin{array}{@{}rcl}
\tilde\rho\complyTBO \tilde\sigma &\textrm{if}& \exists f .~\Pred[ f: \tilde\rho\complyTBO \tilde\sigma].
 \end{array} $

\end{enumerate}
\end{defi}

Along the very same lines of the proof of Theorem \ref{th:complyequivtbcomply}, it is possible to show the equivalence of $\complyO$ and $\complyTBO$.

\begin{lem}\label{lem:complyOeqcomplyTBO}
Let ${\rho},{\sigma}\in\ASC (\subseteq\tbASC)$ and $f\in\Orch$.
$$f: \rho\complyO \sigma \hspace{2mm} \Iff \hspace{2mm}  f: \rho\complyTBO \sigma$$
\end{lem}


\subsection{Proof of Theorem \ref{th:derstratorchequiv}\ref{th:derstratorchequiv-i}}
\label{subsect:derstrat-i}
We can simply use the function $\fofder{\mbox{-},\emptyset}$  described in the proof of Lemma \ref{prop:derderorch-i}.


\subsection{Proof of Theorem \ref{th:derstratorchequiv}\ref{th:derstratorchequiv-ii}}
\label{subsect:derstrat-ii}

We  proceed as follows:
given an orchestrator $f$ such that $f: \rho\complyTBO \sigma$ (and hence $f: \rho\complyO \sigma$, by Lemma \ref{lem:complyOeqcomplyTBO}) we build a regular tree
which corresponds to a tree in $\rts{\rho\tbpp \sigma}$ (see Def. \ref{def:rtstbpp}) with no leaf of the form $\xmark$.
We can then decorate such a tree so that it is easy to obtain a winning strategy 
for player $\playerC$ in $\game_{\rho\!\pp\! \sigma}$.\\

We begin by showing how to get a regular tree out of a turn-based orchestrated system

\begin{defi}
\label{def:defofrt}
Let $\tilde{\rho},\tilde{\sigma}\in\sctb$ and $f\in\Orch$. We define the regular tree of the orchestrated system $\tilde{\rho}\tbpf{f}\tilde{\sigma}$, which we dub $\rt{\tilde{\rho}\tbpf{f}\tilde{\sigma}}$, as follows:\\
let $B=\{\playerA:a,\playerA:\Dual{a},\playerB:a,\playerB:\Dual{a} \mid a\in\Names\}$
\[
\begin{array}{rcc@{~~~~~~~~}l}

\rt{\stopA\tbpf{f}\tilde{\sigma}} & = 
&         \begin{array}{c}
             \stopA\tbpf{f}\tilde{\sigma}\\
            \mid \\
           \Zero\tbpf{f}\tilde{\sigma}
            \end{array} \\[8mm]

\rt{\tilde{\rho}\tbpf{f}\tilde{\sigma}} & = 
 &        \begin{array}{c}
            \tilde{\rho}\tbpf{f}\tilde{\sigma} \\
            / \cdots\backslash \\
            \mathsf{T}_1 \cdots \mathsf{T}_n
             \end{array}
&  \mbox{ if } ~~\exists \beta\in B.\, \tilde{\rho}\tbpf{f}\tilde{\sigma} \tbolts{\beta}\\[2mm]
\multicolumn{4}{l}{
\mbox{where $\Set{\tilde{\rho}_i\tbpf{f_i}\tilde{\sigma}_i}_{i=1..n} = \Set{\tilde{\rho}'\tbpf{f}\tilde{\sigma}' ~\mid~\tilde{\rho}\tbpf{f}\tilde{\sigma} \tbolts{\beta}\tilde{\rho}'\tbpf{f'}\tilde{\sigma}', ~\beta\in B }$  }
}\\[1mm]
\multicolumn{4}{l}{
\mbox{and   $\mathsf{T}_i = \rt{\tilde{\rho}_i\tbpf{f_i}\tilde{\sigma}_i}$ ($i=1..n$) }
}
\\[6mm]

\rt{\tilde{\rho}\tbpf{f}\tilde{\sigma}} & = 
 &         \begin{array}{c}
            \tilde{\rho}\tbpf{f}\tilde{\sigma} \\
            \mid  \\
            \mathsf{T}
             \end{array}   
&  \mbox{ if } ~~\exists a\in\Names.\, \tilde{\rho}\tbpf{f}\tilde{\sigma} \tbolts{\playerC:a} \tilde{\rho}'\tbpf{f'}\tilde{\sigma}'\\[6mm]
\multicolumn{4}{l}{
\mbox{where  $\mathsf{T} =  \rt{\tilde{\rho}'\tbpf{f'}\tilde{\sigma}'}$ }
}\\[6mm]

\rt{\tilde{\rho}\tbpf{f}\tilde{\sigma}} & = 
&          \begin{array}{c}
            \tilde{\rho}\tbpf{f}\tilde{\sigma} \\
            \mid  \\
            \xmark
             \end{array}
&  \mbox{ if } ~~\rho\neq\stopA \mbox{ and }\; \tilde{\rho}\tbpf{f}\tilde{\sigma} ~~\nottblts{}\\[2mm]

\end{array}
\]
\end{defi}

Notice that the condition of the third clause in the above definition is actually nondeterministic;
so, strictly speaking, we are not defining a function.
We can get a proper function definition by any method through which it is possible to get rid of such an ambiguity. For instance, by totally ordering the set $\Names$ and considering 
the first element of its satisfying the condition.


\begin{lem}\label{lem:equivtb}
 Let $\tilde\rho, \tilde\sigma \in \tbASC$ and $f\in\Orch$.
$f: \tilde\rho\complyTBO\tilde\sigma$ ~~  iff  ~~ in $\rt{\tilde\rho\tbpf{f}\tilde\sigma}$
all the leaves are of the form $\Zero\tbpf{f}\tilde{\sigma}$
\end{lem}

Given $\rt{\rho\tbpf{f}\sigma}$, we denote by $\rt{\rho\tbpf{f}\sigma}^-$ the tree obtained out of $\rt{\rho\tbpf{f}\sigma}$ by erasing the 
orchestrators from the label of its nodes.

\begin{lem}
\label{lem:rtinrts}
Let $\rho, \sigma \in \ASC (\subseteq\tbASC)$ and $f\in\Orch$.
If $f:\rho\complyTBO\sigma$, then $\rt{\rho\tbpf{f}\sigma}^-\in \rts{\rho\tbpp\sigma}$.
\end{lem}

\begin{defi}
Let ${\rho},{\sigma}\in\ASC (\subseteq\tbASC)$ and $f\in\Orch$. We define the strategy regular tree of the orchestrated system ${\rho}\tbpf{f}{\sigma}$, which we dub $\srt{{\rho}\tbpf{f}{\sigma}}$, as follows:\\
let 
$B=\{\playerA:a,\playerA:\Dual{a},\playerB:a,\playerB:\Dual{a} \mid a\in\Names\}$ and 
let $\srtaux{\tilde{\rho}\tbpf{f}\tilde{\sigma}}{n}$ with $\tilde{\rho},\tilde{\sigma}\in\ASC $ and $n\in\Nat$ be defined by   
\[
\begin{array}{rcc@{~~~~~~~~}l}

\srtaux{\stopA\tbpf{f}\tilde{\sigma}}{n} & = 
&         \begin{array}{c}
             \stopA\tbpf{f}\tilde{\sigma}\\
            \hspace{40pt}\mid \mbox{\scriptsize $(n+1,\playerC:\cmark)$}\\
           \Zero\tbpp\tilde{\sigma}
            \end{array} \\[8mm]

\srtaux{\tilde{\rho}\tbpf{f}\tilde{\sigma}}{n} & = 
 &        \begin{array}{c}
            \tilde{\rho}\tbpf{f}\tilde{\sigma} \\
           \mbox{\scriptsize $(n+1,\beta_1)$} /  \cdots\backslash\mbox{\scriptsize $(n+1,\beta_n)$} \\
            \mathsf{T}_1 \cdots \mathsf{T}_n
             \end{array}
&  \mbox{ if } ~~\exists \beta\in B.\, \tilde{\rho}\tbpp\tilde{\sigma} \tbolts{\beta}\\[2mm]
\multicolumn{4}{l}{
\mbox{where $\tilde{\rho}\tbpf{f}\tilde{\sigma} \tbolts{\beta_i}\tilde{\rho}_i\tbpf{f'}\tilde{\sigma}_i, ~\beta_i\in B$, $i=1..n$, and  $\mathsf{T}_i = \srtaux{\tilde{\rho}_i\tbpp\tilde{\sigma}_i}{n+1}$}  }
\\[6mm]

\srtaux{\tilde{\rho}\tbpf{f}\tilde{\sigma}}{n} & = 
 &         \begin{array}{c}
            \tilde{\rho}\tbpf{f}\tilde{\sigma} \\
            \hspace{40pt}\mid \mbox{\scriptsize $(n+1,\playerC{:}a)$}  \\
            \mathsf{T}
             \end{array}   
& \hspace{-1mm} \mbox{ if } ~~\exists a\in\Names.\, \tilde{\rho}\tbpf{f}\tilde{\sigma} \tbolts{\playerC:a} \tilde{\rho}'\tbpf{f'}\tilde{\sigma}'\\[6mm]
\multicolumn{4}{l}{
\mbox{where  $\mathsf{T} =  \rt{\tilde{\rho}'\tbpf{f'}\tilde{\sigma}'}$ }
}\\[6mm]

\srtaux{\tilde{\rho}\tbpf{f}\tilde{\sigma}}{n} & = 
&          \begin{array}{c}
            \tilde{\rho}\tbpf{f}\tilde{\sigma} \\
            \mid  \\
            \xmark
             \end{array}
&  \mbox{ if } ~~\rho\neq\stopA \mbox{ and }\; \tilde{\rho}\tbpp\tilde{\sigma} ~~\nottblts{}_{\!\!\!\!\!o}\\[2mm]

\end{array}
\]
hence  
\[
\srt{{\rho}\tbpf{f}{\sigma}} = \srtaux{{\rho}\tbpf{f}{\sigma}}{0}.
\]
\end{defi}
The same observation we made after Definition \ref{def:defofrt} holds here for what concerns
the third clause of the definition of ${\sf srt\mbox{-}aux}$.

Notice that, by construction, $\srt{{\rho}\tbpf{f}{\sigma}}$ is but a ''decorated'' version of
$\rt{{\rho}\tbpp{\sigma}}$. We hence get the following

\begin{fact}
\label{fac:leavesform}
All the leaves of $\srt{{\rho}\tbpf{f}{\sigma}}$ are of the form  $\Zero\tbpp\tilde{\sigma}$
if and only if all the leaves of $\rt{{\rho}\tbpp{\sigma}}$ are so.
\end{fact}

\begin{defi}
\label{def:Sigmaf}
Given an orchestrated system $\rho\tbpf{f} \sigma$, the strategy $\Sigma_f$ for $\playerC$ in the game  $\game_{\rho\pp\sigma}$ is defined as follows:\\
let $\vec{e}=\seq{e_0\cdots e_n}$ be a finite play in $\game_{\rho\pp\sigma}$
$$\Sigma_f(\vec{e}) = \left\{ \begin{array}{l@{\hspace{8mm}}l}
                                      (k+1,\playerC{:}a)  & if ~~(*)\\[2mm]  
                                      \emptyset & \mbox{otherwise}
                                        \end{array} \right.
$$
(*) $\vec{e}$ is a sequence of labels in $\srt{\tilde{\rho}\tbpf{f}\tilde{\sigma}}$ from the root to a
node of the form\linebreak
$\Dual{\alpha}.\rho' + \rho'' \tbpf{f'} \alpha. \sigma' + \sigma''$, and where  where $(k,\beta)$ is the label of the arc above such a node node ($k=0$ if the node coincides with the root).
\end{defi}

\begin{fact}
\label{fac:sigmafunivocal}
Given an orchestrated system $\rho\tbpf{f} \sigma$,
the strategy $\Sigma_f$ is univocal.
\end{fact}

\begin{lem}
\label{lem:rightimplth}
Let $ f: \rho\complyO \sigma$. The strategy $\Sigma_f$ is univocal and winning  for $\playerC$ in the game  $\game_{\rho\pp\sigma}$.
\end{lem}
\begin{proof}
Let $f: \rho\complyO \sigma$, then by Lemma \ref{lem:complyOeqcomplyTBO} we have that
$f: \rho\complyTBO \sigma$. By Lemma  \ref{lem:equivtb}, no leaf in
$\rt{\rho\tbpf{f}\sigma}$, and hence in $\rt{\rho\tbpf{f}\sigma}$ (by Fact \ref{fac:leavesform}), is of the form $\xmark$. By definition of $\Sigma_f$ 
and by
Lemma \ref{lem:winC}, 
it is winning  for $\playerC$ in the game  $\game_{\rho\pp\sigma}$. It is also univocal by Fact \ref{fac:sigmafunivocal}.
\end{proof}


\subsection{Proof of Theorem \ref{th:derstratorchequiv}\ref{th:derstratorchequiv-iii}}\hfill\\
\label{subsect:derstrat-iii}

Let $\Sigma$ be a univocal winning strategy for player $\playerC$ for the game $\game_{\rho\!\pp\!\sigma}$. We then take into account the tree $\rtsigma{\Sigma}{\rho\!\tbpp\!\sigma}$
as defined in Definition \ref{def:treefromstrat}. 
We show now how to get an orchestrator such that $f: \rho\complyTBO \sigma$ (and hence $f: \rho\complyO \sigma$) out of $\rtsigma{\Sigma}{\rho\!\tbpp\!\sigma}$.

By means of the following definition we shall be able to get an orchestrator out of a tree in 
$\rts{\rho\tbpp\sigma}$ which does not contain leaves of the form $\xmark$.

\begin{defi}
Let $\mathsf T\in \rts{\rho\tbpp\sigma}$, the unproper orchestrator $\orch{\mathsf T}$ (that is possibly containing leaves of the form $\xmark$) is defined as
follows:
\[
\begin{array}{lcc@{~~~~~~~~}l}
\mathsf{orch}\Big( ~\begin{array}{c}
             \stopA\tbpp\tilde{\sigma}\\
            \mid \\
           \Zero\tbpp\tilde{\sigma}
            \end{array} ~\Big)
& = &
\stopf
\\[8mm]
\mathsf{orch}\Big(\begin{array}{c}
            \tilde{\rho}\tbpp\tilde{\sigma} \\
            / \cdots\backslash \\
            \mathsf{T}_1 \cdots \mathsf{T}_n
             \end{array}\Big)
&=&
\bigvee_{i\in\{1..n\}}\orch{ \mathsf{T}_i}\\
\multicolumn{4}{l}{
\mbox{where, for $i=1..n$, $\mathsf{T}_i \in \rts{\rho_i\tbpp\sigma_i}$ with  ${\rho}\tbpp{\sigma} \tbolts{Y:\Dual{a}_i}{\rho}'\tbpp{\sigma}'$ and $Y\in\Set{\playerA,\playerB}$.
}}
\\[4mm]
\mathsf{orch}\Big( ~\begin{array}{c}
            \tilde{\rho}\tbpp\tilde{\sigma} \\
            \mid  \\
            \mathsf{T}
             \end{array}\Big)
  
& =  & 
\orchAct{\Dual{a}}{a}.\orch{\mathsf{T}}
\\
\multicolumn{4}{l}{
\mbox{where $\mathsf{T} \in \rts{\rho'\tbpp\sigma'}$ with  ${\rho}\tbpp{\sigma} \tbolts{\playerA:a}{\rho}'\tbpp{\sigma}'$.
}}
\\[4mm]
\mathsf{orch}\Big( ~\begin{array}{c}
            \tilde{\rho}\tbpp\tilde{\sigma} \\
            \mid  \\
            \mathsf{T}
             \end{array}\Big)
  
& =  & 
\orchAct{a}{\Dual{a}}.\orch{\mathsf{T}}
\\
\multicolumn{4}{l}{
\mbox{where $\mathsf{T} \in \rts{\rho'\tbpp\sigma'}$ with  ${\rho}\tbpp{\sigma} \tbolts{\playerB:a}{\rho}'\tbpp{\sigma}'$.
}}
\\[4mm]
\mathsf{orch}\Big( ~\begin{array}{c}
            \tilde{\rho}\tbpp\tilde{\sigma} \\
            \mid  \\
            \mathsf{T}
             \end{array}\Big)
  
& =  & 
\orchAct{\alpha}{\Dual{\alpha}}^+.\orch{\mathsf{T}}
\\
\multicolumn{4}{l}{
\mbox{where $\tilde{\rho}\tbpp\tilde{\sigma}= \Dual{\alpha}.\rho + \rho' \tbpp \alpha.\sigma + \sigma'$ and $\mathsf{T} \in \rts{\rho\tbpp\sigma}$.
}}\\[6mm]

\mathsf{orch}\Big( ~\begin{array}{c}
            \tilde{\rho}\tbpp\tilde{\sigma} \\
            \mid  \\
            \xmark
             \end{array}\Big) 
& = &   
\xmark
\end{array}
\]
\end{defi}
By the above definition it is immediate to check the following fact.

\begin{fact}
\label{fact:orchnocirc}
Let $\mathsf T\in \rts{\rho\tbpp\sigma}$ such that all its leaves are of the form $\Zero\tbpp\tilde{\sigma}$.\\
 Then $\orch{\mathsf T} $ is a proper orchestrator, i.e. it does not contain any leaf of the form $`\xmark$'.
\end{fact}

\begin{lem}
Let $\mathsf T\in \rts{\rho\tbpp\sigma}$.
The definition of $\orch{\mathsf T} $ is well-founded
\end{lem}
\begin{proof}
By the regularity of $\mathsf T$.
\end{proof}

\begin{lem}
\label{lem:forchT}
Let $\mathsf T\in \rts{\rho\tbpp\sigma}$ and let $f=\orch{\mathsf T}$.
Then $\rt{\rho\tbpf{f}\sigma}^- = \mathsf T$.
\end{lem}
\begin{proof}
By construction.
\end{proof}

By Lemma \ref{lem:treefromstrat} we have that $\rtsigma{\Sigma}{\rho\!\tbpp\!\sigma}$ is such that all its leaves are of the form $\Zero\tbpp\tilde{\sigma}$.  Let now $f=\orch{\mathsf T}$. By Fact \ref{fact:orchnocirc}, 
$f$ is a proper orchestrator.
 By Lemma \ref{lem:forchT} we have that
$\rt{\rho\tbpf{f}\sigma}^- = \mathsf T$. So $\rt{\rho\tbpf{f}\sigma}^-$, and hence $\rt{\rho\tbpf{f}\sigma}$, is such that all its leaves are of the form $\Zero\tbpp\tilde{\sigma}$.
We now get the thesis by Lemma \ref{lem:equivtb}.


\subsection{Proof of Theorem \ref{th:derstratorchequiv}\ref{th:derstratorchequiv-iv}}
\label{subsec:orchtoder}

\begin{figure}[b]
\hrule
\vspace{2mm}
\centering
{\small
\begin{tabbing}
\OrchToDerAux\=$(f, \Gamma, \rho\complyF\sigma)$ \\ [1mm]
\IF \> $\rho = \stopA$~~ \THEN ~~~~~
	$(\mbox{\scriptsize$\CkptcomplAx$}):\Inf{}{\Gamma\der \stopA \complyF \sigma}$
\\[1mm]

\ELSE \> \IF \= ~~~$\rho\complyF\sigma \in \Gamma$ ~~\THEN~~~~~
	$(\mbox{\scriptsize$\CkptcomplHyp$}):\Inf{}{\Gamma, \rho\complyF\sigma \der \rho\complyF\sigma}$ 
\\ [1mm]

\ELSE \> \IF ~~\=~~ \= $f =\orchAct{\Dual{\alpha}_k}{\alpha_k}.f'$~ \AND~ $\rho = \sum_{i\in I}\alpha_i.\rho_i$ ~\= \AND~$\sigma = \sum_{j\in J}\Dual{\alpha}_j.\sigma_j$ \\[2mm]
\> \> \> \AND~~~$k \in I\cap J$ ~\AND~ $\mathcal{D} = $ \OrchToDerAux$(f',\Gamma, \rho\complyF\sigma \der \rho_k\complyF\sigma_k) \neq \FAIL$ 
\\[2mm]
\>  \THEN  ~~~~~
	$(\mbox{\scriptsize$+\cdot+$}):\Inf{\Der}
	{\Gamma\der \rho\complyF\sigma}$ \hspace{2mm} \ELSE ~~\FAIL
\\ [2mm]

\ELSE \> \IF \> $f=\bigvee_{k\in K}\orchAct{a_k}{\Dual{a}_k}.f_{k}$ ~\AND~$\rho = \bigoplus_{i\in I}\Dual{a}_i.\rho_i$ ~\AND~ $\sigma = \sum_{j\in J}a_j.\sigma_j$\\[2mm]
\>\>\> \AND~ $K\supseteq I\subseteq J$ \\[2mm]
\>\>\> \AND~\FA~$i\in I$ ~~$\mathcal{ D}_i = $~\OrchToDerAux$(f_{i},\Gamma, \rho\complyF\sigma \der \rho_i\complyF\sigma_i) \neq \FAIL$\\ [2mm]
\>  \THEN ~~~~~
	$(\mbox{\scriptsize$\oplus\cdot+$}):\Inf{
    		\forall i\in I ~~\mathcal{ D}_i }{
    		\Gamma\der \rho\complyF\sigma}$ \hspace{2mm} \ELSE ~~\FAIL
\\ [1mm]

\ELSE \> \IF \> ~~~$f=\bigvee_{k\in K}\orchAct{\Dual{a}_k}{a_k}.f_{k}$ ~\AND~ $\rho = \sum_{i\in I}a_i.\rho_i$ ~\AND~ $\sigma = \bigoplus_{j\in J}\Dual{a}_j.\sigma_j$\\[2mm]
\>\>\> \AND~ $I\supseteq J\subseteq K$\\ [2mm]
\>\>\> \AND~~\FA~$j\in J$~~~$\mathcal{ D}_j = $~\OrchToDerAux$(f_j,\Gamma, \rho\complyF\sigma \der \rho_j\complyF\sigma_j)\neq \FAIL$\\ [2mm]
\>  \THEN ~~~~~
	$(\mbox{\scriptsize$+\cdot\oplus$}):\Inf{
    		\forall j\in J ~~\mathcal{ D}_j}{
    		\Gamma\der \rho\complyF\sigma}$  \hspace{2mm} \ELSE ~~\FAIL\\ [1mm]
\ELSE \> \FAIL
\end{tabbing}
}
\caption{The procedure\;\OrchToDerAux.}\label{fig:OrchToDerAux}
\vspace{2mm}
\hrule
\end{figure}

\begin{defi}
We define the procedure \OrchToDer\ by
$$\OrchToDer(f,\rho,\sigma) \ByDef \OrchToDerAux(f,\emptyset, \rho\complyF\sigma)$$
where \OrchToDerAux\ is defined as in Figure \ref{fig:OrchToDerAux}.
\end{defi}

\begin{lem}
Let $f$ be such that $f:\rho\complyO\sigma$ holds.\\ Then ${\OrchToDer}(f,\rho,\sigma)$ terminates and $\OrchToDer(f,\rho,\sigma)::\ \der \rho\complyF\sigma$
\end{lem}
\proof
Similar to the proof of Lemma \ref{prop:provecorr}.
\qed

%

\section{Proof of Proposition \ref{th:soundcomplSynth}\\ (Correctness and completeness of \Synth)}
\label{appendix:synth}

Recall that in the algorithm \Synth\ (defined in Figure \ref{fig:Synth}) we are not considering orchestrators up to recursion unfolding.

We prove now a version of Proposition $\ref{th:soundcomplSynth}$, in which derivability
in system $\derinfOrch$ (defined in Figure \ref{fig:forsystinforchder}) is taken into account 
instead of the orchestrated compliance relation.
%
%
Given an orchestrator
$f$ we denote by $\regtree{f}$ its corresponding (possibly infinite) regular tree.

\begin{lem}\hfill
\label{lem:scSynth}
\begin{enumerate}
\item
\label{fact:PdsFmSynthExt1}
If $f=$\Synth$(\emptyset, \rho, \sigma)\neq\FAIL$~then ~$\derinfOrch f: \rho\complyOF \sigma$.
\item
\label{fact:PdsFmSynthExt}
If $\derinfOrch f: \rho\complyOF \sigma$~then~ there exists g such that 
$g=$\Synth$(\emptyset, \rho, \sigma)$ with $\regtree{f}=\regtree{g}$.
\end{enumerate}
\end{lem}
\begin{proof}\hfill
\begin{enumerate}
\item[(\ref{fact:PdsFmSynthExt1})]
Immediate, since the procedure \Synth\ is the formalisation of a proof search in System $\derinfOrch$,
as defined in Definition \ref{def:formalCompl}
\item[(\ref{fact:PdsFmSynthExt})]
Given a derivation tree for $\derinfOrch f: \rho\complyOF \sigma$, let us consider one
path $p$ starting from the root, and such that 
\begin{itemize}
\item
$p$ ends with an occurrence of (\CkptcomplHyp):
$\Gamma', x: \rho'\complyOF \sigma'\derinfOrch x: \rho'\complyOF \sigma'$
\item
$p$ contains more than one other judgments of the form
$\Gamma'' \derinfOrch f'': \rho'\complyOF \sigma'$.
\end{itemize}
If no such a path exists, then any rule in the derivation does precisely correspond to a clause of the algorithm \Synth\ and hence the algorithm returns $f$.
Otherwise, since the rules of System $\derinfOrch$ are such that
the proof search is deterministic it is possible to modify the derivation such that in 
the path from the root to the last judgment of $p$ there is just one other 
judgment of the form $\Gamma'' \derinfOrch f'': \rho'\complyOF \sigma'$.
The conclusion of the new derivation will now be 
$\derinfOrch g': \rho\complyOF \sigma$ with $\regtree{f}=\regtree{g'}$.
We can now keep on applying such a procedure until paths like $p$ above no longer exists.
\qedhere
\end{enumerate}
\end{proof}
\noindent
Now we can get Proposition \ref{th:soundcomplSynth} as a corollary of Lemma 
\ref{lem:scSynth}, using 
Proposition \ref{lem:terminationSynth}, Theorem \ref{th:scder}
and Proposition
\ref{prop:derOrchderinfOrchequiv}.

\section{Proof of Theorem \ref{thm:functder}\\ (Derivations as orchestrator functors)}\hfill
\label{subsec:functder}


Form now on, we consider orchestrators, contracts and functors as the (possibly infinite) regular trees they represent.
%
We consider now  infinitary versions of $\derOrch$, $\dersc$ and $\der$.
\begin{defi}
We define
\begin{itemize}
\item 
$\derOrchinfty$ as the system $\derOrch$ without rule \mbox{\scriptsize $(\CkptcomplHyp)$};
\item
{$\derscinfty$} as the system $\dersc$ without rule \mbox{\scriptsize $(\CkptcomplHyp\mbox{\,-\!}\subcontrF)$};
\item
$\derinfty$ as the system $\der$ without rule \mbox{\scriptsize $(\CkptcomplHyp)$}.
\end{itemize}
Moreover, infinite derivations are allowed in the above systems.
\end{defi}

It is not difficult now to check the following lemma.
\begin{lem}\hfill
\label{lem:derderinftyequiv}
\begin{enumerate}
\item
$
\hspace{1mm} \derOrch f: {\rho} \complyOF \sigma~~ \Iff~~ \derOrchinfty f: {\rho} \complyOF \sigma
$
\item
$
\dersc \sigma \subcontrF \sigma'~~ \Iff~~ \derscinfty \sigma \subcontrF \sigma'
$
\end{enumerate}
\end{lem}

\noindent
We now prove Theorem \ref{thm:functder} by proceeding as follows:
We first define a proof reconstruction procedure $\PProcinfty$ taking as argument two derivations $\Deriv' :: \,\derOrchinfty f: {\rho} \complyOF \sigma$ and
$\Deriv'' :: {\derscinfty} \!\sigma \subcontrF \sigma'$, and, in case it does not fails, it produces a 
(possibly infinite)  derivation $\Deriv'''$ in system $\derinfty$ partially
decorated with orchestration actions.
 We then show that $\PProc$ does not fail and that the derivation $\Deriv'''$ can be easily turned in a derivation $\widetilde{\Deriv'''}$ is such that 
 $\widetilde{\Deriv'''}  :: \, \derOrchinfty \mathbf{F}(f): {\rho} \complyOF
 \sigma$.

Theorem \ref{thm:functder} hence descends immediately by Theorem \ref{th:scder}
and Lemma \ref{lem:derderinftyequiv}.

\begin{defi}[The algorithm $\PProcinfty$] 
Let 
$\Deriv' :: \, \derOrchinfty f: {\rho} \complyOF \sigma$ and
$\Deriv'' :: \derscinfty \sigma \subcontrF \sigma' $.
The algorithm $\PProc$ is defined by
$$\PProcinfty(\Deriv' ,\,\Deriv'')\ =\ \Procinfty(\Deriv' ,\,\Deriv'',\, \emptyset)$$
where $\Procinfty$ is a procedure with an extra argument (an environment),
defined by
cases according to the following clauses.
We name the clauses with the name of the
last rules applied in the first two derivations.
The algorithm fails in case no clause can be applied.

\def\Clause#1{\vspace{10pt} \noindent {\bf Clause #1}\,{\rm :}} 

\Clause{$*${-}$(\TcomplAx\mbox{\,-\!}\subcontrF)$} 

\[\Procinfty\Big( ~
{\InfBox {\Deriv_1 }{\Gamma_1\derinfOrch  f:{\rho} \complyOF \sigma}}, ~
\Inf[\TcomplAx]{}
{ \Gamma_2\dersc \stopA \subcontrF \sigma'}
, ~
\Gamma_3~
\Big) =
\Inf[\TcomplAx]{}{ \Gamma_3\derinfty\stopf: \stopA \complyF \sigma'}\] \\

\noindent
Notice that it is not necessary to have a {\bf Clause $*$-$(\TcomplAx)$}, since in that case
$\sigma=\stopA$ and hence also $\rho=\stopA$. This means that {\bf Clause $(\TcomplAx)$-$*$} applies.

\Clause {$(+\cdot\oplus)$-$(\oplus\cdot+\mbox{\,-\!}\subcontrF)$}

\[ 
\begin{array}{ll}
\Procinfty\Big( &
\Inf	[+\cdot\oplus] 
	{\InfBox {\Deriv'_i }
		{ \Gamma'_1 \derOrchinfty f_i : \rho_i \complyOF {\sigma}_i} 
	~~ ( \forall i \in I)
	}{ \Gamma_1\derOrchinfty \bigvee_{i\in I}\langle\Dual{a}_i,{a_i}\rangle.f_i : \bigExternal_{j \in I\cup J} a_j.{\rho}_j\complyOF\bigInternal_{i \in I} \Dual{a}_i.{\sigma}_i   }~~,\\[2mm]
&\mbox{\scriptsize where~~$
\Gamma_1'= \Gamma_1, \bigvee_{i\in I}\langle\Dual{a}_i,{a_i}\rangle.f_i  :\sum_{j \in I\cup J} a_j.{\rho}_j \complyOF \oplus_{i \in I} \Dual{a}_i.{\sigma}_i $}
\\[4mm]
&
\Inf	[\oplus\cdot+\mbox{\,-\!}\subcontrF] 
	{(h\in I\cap K) \quad \InfBox {\Deriv''_h }
		{ \Gamma_2'
     	\derscinfty \sigma_h \subcontrF \sigma'_h}
	}{ \Gamma_2\derscinfty \bigInternal_{i \in I} \Dual{a}_i.{\sigma}_i \subcontrF\bigExternal_{k \in K} \Dual{a}_k.\sigma_k' }~~,   \Gamma_3  ~\Big) =\\[2mm]
&\mbox{\scriptsize where~~$
\Gamma_2'=\Gamma_2, \bigInternal_{i \in I} \Dual{a}_i.{\sigma}_i \subcontrF\bigExternal_{k \in K} \Dual{a}_k.\sigma_k'  $} 
\end{array} 
\] 
\vspace{2mm}
{\hfill $\Inf[+\cdot+]
	{\InfBox {\Deriv_h }
		{ \Gamma_3, \rho \complyF \sigma' \derinfty \mu:\rho_h \complyOF \sigma_h'} 
	\quad (h\in I\cap K)
	}{ \Gamma_3\derinfty \orchAct{\Dual{a}_h}{a_h}^+ : \bigExternal_{j \in I\cup J} a_j.{\rho}_j \complyOF \bigExternal_{k \in K} \Dual{a}_k.\sigma_k'  }
$}\\
\WHERE\\[3mm]
 {\hfill $\InfBox {\Deriv_h }
		{ \Gamma_3, \rho \complyF \sigma' \derinfty \mu:\rho_h \complyOF \sigma_h'} $ 
\ =\ 
$\Procinfty \Big( \InfBox{\Deriv'_i }{ \Gamma'_1\derOrchinfty f_i : \rho_i \complyOF {\sigma}_i }
, ~
\InfBox {\Deriv''_h }
		{ \Gamma'_2 \dersc  \sigma_h \subcontrF {\sigma}'_h} 
, ~
(\Gamma_3,  \rho  \complyF \sigma')~\Big)$} \\[2mm]
\AND~~ $\rho=\bigExternal_{j \in I\cup J} a_j.{\rho}_j$ and $\sigma' = \bigExternal_{k \in K} a_k.\sigma_k' $.

\vspace{2mm}

\Clause {$(+\cdot+)$-$(+\cdot+\mbox{\,-\!}\subcontrF)$}

\[ 
\begin{array}{ll}
\Procinfty\Big( &
\Inf	[+\cdot+] 
	{\InfBox {\Deriv_p }
		{ \Gamma'_1 \derOrchinfty f' : \rho_p \complyOF {\sigma}_p} 
	\qquad (p  \in I\cap K)
	}{ \Gamma_1\derOrchinfty \orchAct{\alpha_p}{\Dual{\alpha}_p}^+.f' : \bigExternal_{k\in K} \Dual{\alpha}_k.{\rho}_k \complyOF \bigExternal_{i \in I}\alpha_i.\sigma_i  }~~,\\[2mm]
&\mbox{\footnotesize where~~ $\Gamma_1'= \Gamma_1, \orchAct{\alpha_p}{\Dual{\alpha}_p}^+.f' : \bigExternal_{k\in K} \Dual{\alpha}_k.{\rho}_k \complyOF \bigExternal_{i \in I}\alpha_i.\sigma_i $ }
\\ [4mm]
&
\Inf	[+\cdot+\mbox{\,-\!}\subcontr] 
	{\InfBox {\tilde\Deriv_i }
		{ \Gamma'_2 \derscinfty \sigma_i \subcontrF \sigma'_i }
	\qquad ( \forall i \in I)
	}{ \Gamma_2\derscinfty \bigExternal_{i \in I}\alpha_i.\sigma_i \subcontrF \bigExternal_{j \in I\cup J} \alpha_j.\sigma'_j }~~,\ \Gamma_3 \ ~\Big) =\\[2mm]
&\mbox{\footnotesize where~~
$\Gamma_2'=\Gamma_2, \bigExternal_{i \in I}\alpha_i.\sigma_i \subcontrF \bigExternal_{j \in I\cup J} \alpha_j.\sigma'_j$} 
\end{array} \vspace{2mm}\]

{\hfill $\Inf	[+\cdot+]
	{
	\InfBox {\Deriv'_p }
		{ \Gamma_3, \rho \complyF \sigma' \derinfty \mu : \rho_p \complyF \sigma'_p} 
	\qquad ( p \in   (I\cup J) \cap K)
	}{ \Gamma_3\derinfty \orchAct{\alpha_p}{\Dual{\alpha}_p}^+ : \rho \complyF \sigma' }$}\\[2mm]
\noindent
\WHERE\\[3mm]
{\hfill $\InfBox {\Deriv'_p }{ \Gamma_3, \rho \complyOF \sigma' \derinfty \mu: \rho_p \complyF \sigma'_p}\
=\ 
\Procinfty \big( \InfBox {\Deriv_p }{ \Gamma'_1\derOrchinfty f' : \rho_p \complyOF {\sigma}_p }
, ~
\InfBox {\tilde\Deriv_p}
		{ \Gamma'_2 \derscinfty \sigma_p \subcontrF \sigma'_p }
, ~
(\Gamma_3,  \rho  \complyF \sigma') \ \big)
$}\\[2mm]
\AND~~ $\rho=\bigExternal_{k \in K} \Dual{\alpha}_k.{\rho}_k $ and $\sigma' =\bigExternal_{j\in I\cup J} \alpha_j.\sigma'_j$,

\Clause {$(\oplus\cdot+)$-$(+\cdot+\mbox{\,-\!}\subcontrF)$}

The construction follows a definition pattern similar to those of the previous clauses.

\Clause {$(+\cdot\oplus)$-$(\oplus\cdot\oplus\mbox{\,-\!}\subcontrF)$}

The construction follows a definition pattern similar to those of the previous clauses.

\setlength {\unitlength} {1\point}
\end{defi}

\begin{prop}\hfill\\
Let 
$\Deriv' ::\ \derOrchinfty f: {\rho} \complyOF \sigma$ and
$\Deriv'' ::\ \derscinfty \sigma \subcontrF \sigma' $, and let $\Funct_{\Der''}$ be defined as in Definition \ref{def:orchfunctor}. 
\begin{enumerate}
\item
The computation of $\PProcinfty(\Deriv', \Deriv'' ) $ never fails.
\item
$\PProcinfty(  \Deriv', \Deriv'' )\ =\ \Deriv ::\ \derinfty \rho \complyF \sigma' $, 
where $\Deriv$ is decorated with orchestration actions. Moreover, out of $\Deriv$ it is possible to get $\widetilde{\Deriv}$ 
such that $\widetilde{\Deriv}::\ \derinfty \Funct_{\Der''}(f) : \rho \complyOF \sigma'$.
\end{enumerate}
\end{prop}
\proof
By inspection of the clauses of the procedure, the computation never fails if we 
start from $\Deriv'$ and $\Deriv''$ as above.\\
Out of the (possible infinite) decorated derivation $\Der$ is is possible to get an
orchestrator $\tilde f$ such that  $\tilde f : \rho \complyOF \sigma'$ (because of 
the  regularity of the derivation tree $\Der$) and a derivation $\widetilde{\Deriv}::\ \derinfty \tilde f : \rho \complyOF \sigma'$. Besides, working on the
form of the clauses of the procedure, it can be shown that $\tilde f=\Funct_{\Der''}(f)$.
\qed

\noindent
Theorem \ref{thm:functder} descends now as a corollary from the above proposition.


 \if false   

\subsection{Alternative Proof of Theorem \ref{thm:functder} (with explicit orchestrators and functors)}\hfill
\label{subsec:functder}\hfill\\

In order to prove Theorem \ref{thm:functder}, we proceed as follows:\\
We first define a proof reconstruction procedure $\PProc$ taking as argument two derivations $\Deriv' :: \derinfOrch f: {\rho} \complyOF \sigma$ and
$\Deriv'' :: \dersc \sigma \subcontrF \sigma' $, and, in case of successful termination, returns a 
derivation $\Deriv'''$ in system $\derinfOrch$.\\
 We then prove that $\PProc$ terminates and that the derivation $\Deriv'''$ returned is such that 
 $\Deriv''' = :: \derinfOrch \mathbf{F}(f'): {\rho} \complyOF \sigma$ (
where $f'$ represents the same regular tree as $f$).\\
Theorem \ref{thm:functder} hence descends immediately by Theorem \ref{th:scder}
and Proposition \ref{prop:derOrchderinfOrchequiv}.

\begin{defi}[The algorithm $\PProc$] 
Let 
$\Deriv' :: \derinfOrch f: {\rho} \complyOF \sigma$ and
$\Deriv'' :: \dersc \sigma \subcontrF \sigma' $.\\
The algorithm $\PProc$ is defined by
$$\PProc(\Der',\, \Deriv'') = \Proc(\Deriv',\Deriv'',\emptyset,\Deriv',\Deriv'')$$

where the auxiliary procedure $\Proc$ (whose first and fourth arguments are derivations in $\derinfOrch$, the second and fifth ones are derivatiosn in $\derinfOrch$ and the third one is 
 an environment) is defined by cases, according to the following clauses.\\
We name the clauses with the name of the
last rules applied in the first two derivations.In names like $R$-$*$, the symbol $*$ stands for any rule that can be paired with $R$. \\
The algorithm fails in case no clause can be applied. \\

Notice how the fourth and  fifth arguments of $\Proc$ are used to keep track of the two derivations $\Deriv'$ and $\Deriv''$
of the call of $\PProc(\Deriv',\Deriv'')$ .
We assume the Barendregt's convention (i.e. we assume all the bound variables in an orchestrator to be distinct).

\Clause{Init} 

\[
\Proc\big( ~
\Inf[\TcomplHyp]{}{ \Gamma_1,x:{\rho} \complyOF \sigma \derinfOrch x:{\rho} \complyOF \sigma},~
\Inf[\TcomplHyp\mbox{\,-\!}\subcontr]{}{ \Gamma_2, \sigma \subcontrF \sigma' \dersc\sigma \subcontrF \sigma'},~
\Gamma_3, x{:} {\rho}  \complyOF \sigma',\ \Deriv',\ \Deriv''~
\big)
~=~\]
{\hfill $\Inf[\TcomplHyp]{}
{ \Gamma_3, x{:} {\rho}  \complyF \sigma' \derinfOrch x{:}{\rho}  \complyF \sigma'}
$}

\Clause{$*${-}$(\TcomplAx\mbox{\,-\!}\subcontr)$} 

\[\Proc\Big( ~
{\InfBox {\Deriv_1 }{\Gamma_1\derinfOrch  f:{\rho} \complyOF \sigma}}, ~
\Inf[\TcomplAx]{}
{ \Gamma_2\dersc \stopA \subcontrF \sigma'}
, ~
\Gamma_3, ~\Deriv',~\Deriv''~
\Big) =
\Inf[\TcomplAx]{}{ \Gamma_3\derinfOrch\stopf :\stopA \complyF \sigma'}\] \\

\noindent
Notice that it is not necessary to have a {\bf Clause $*$-$(\TcomplAx)$}, since in that case
$\sigma=\stopA$ and hence also $\rho=\stopA$. This means that {\bf Clause $(\TcomplAx)$-$*$} applies.

\Clause{$(\TcomplHyp)$-$(\TcomplHyp\mbox{\,-\!}\subcontr)$}\hfill\\[1mm]
\underline{\sc in case } $x:{\rho} \complyOF \sigma'\not\in\Gamma_3\ $:
 \[ \Proc\big( ~
\Inf[\TcomplHyp]{}{ \Gamma_1,x:{\rho} \complyOF \sigma \derinfOrch x:{\rho} \complyOF \sigma},~
\Inf[\TcomplHyp\mbox{\,-\!}\subcontr]{}{ \Gamma_2, \sigma \subcontrF \sigma' \dersc\sigma \subcontrF \sigma'},~
\Gamma_3,\ \Deriv',\ \Deriv''~
\big) = \]
{\hfill $\Proc\big( ~
{\InfBox {\Deriv_1 }{ \Gamma'_1 \derinfOrch f':{\rho} \complyOF \sigma}},~
{\InfBox {\Deriv_2 }{ \Gamma'_2 \dersc \sigma \subcontrF \sigma'}},~ \Gamma_3, ~\Deriv',~\Deriv'' ~\big)$
}\\
\WHERE~~~ \\[1mm] $\InfBox {\Deriv_1 }{ \Gamma'_1 \derinfOrch f':{\rho} \complyOF \sigma}$
is a subderivation of $\Deriv'$ such that $x:{\rho} \complyOF \sigma\not\in\Gamma'_1 $, if there is one like that;\\[2mm]
\AND    ~~~\\[1mm]$\InfBox {\Deriv_2 }{ \Gamma'_2 \dersc \sigma \subcontrF \sigma'}$
is a subderivation of $\Deriv''$ such that $\sigma \subcontrF \sigma'\not\in \Gamma'_2$, if there is one like that.\\

\Clause{$(\TcomplHyp)$-$*$}
\[\Proc\big( ~
\Inf[\TcomplHyp]{}{ \Gamma_1,x:{\rho} \complyOF \sigma \derinfOrch x:{\rho} \complyOF \sigma},~
\InfBox {\Deriv_2 }{ \Gamma_2\dersc \sigma \subcontrF \sigma' },~
\Gamma_3,\ \Deriv',\ \Deriv''~
\big) = \]
{\hfill $\Proc\big( 
\InfBox {\Deriv_1 }{ \Gamma'_1 \derinfOrch f' : {\rho} \complyOF \sigma},~
\InfBox {\Deriv_2 }{ \Gamma_2\dersc \sigma \subcontrF \sigma' } ,~
\Gamma_3,\ \Deriv',\ \Deriv''~\big)$}\\
\WHERE ~~~\\[1mm] 
$\InfBox {\Deriv_1 }{ \Gamma'_1 \derinfOrch f' : {\rho} \complyOF \sigma}$
is a subderivation of $\Deriv'$,
such that $x:{\rho} \complyOF \sigma\not\in\Gamma'_1 $, if there is one like that.\\

\Clause{$*$-$(\TcomplHyp\mbox{\,-\!}\subcontr)$}
\[
\Proc\Big( 
\InfBox {\Deriv_1 }{ \Gamma_1 \derinfOrch f : {\rho} \complyF \sigma},~
\Inf[\TcomplHyp\mbox{\,-\!}\subcontr]{}{ \Gamma_2, \sigma \subcontrF \sigma' \dersc\sigma \subcontrF \sigma'} ,~
\Gamma_3,\ \Deriv',\ \Deriv''~\big) =
\] 
{\hfill $\Proc( 
\InfBox {\Deriv_1 }{ \Gamma_1 \derinfOrch f : {\rho} \complyF \sigma},~
\InfBox {\Deriv_2 }{ \Gamma'_2\dersc \sigma \subcontrF \sigma' } ,~ \Gamma_3 ,\ \Deriv',\ \Deriv''~\big)
$}\\
\WHERE ~~~\\[1mm] 
$\InfBox {\Deriv_2 }{ \Gamma'_2\dersc \sigma \subcontrF \sigma' }$
is a subderivation of $\Deriv''$
such that ${\sigma} \subcontrF \sigma'\not\in\Gamma' $, if there is one like that.\\

\begin{rem}
No rule among the last three above can be applied immediately after
the application of one of them. This will be formally shown in Lemma \ref{lem:noinvinv} below.
\end{rem}

\Clause {$(+\cdot\oplus)$-$(\oplus\cdot+\mbox{\,-\!}\subcontr)$}

\[ 
\begin{array}{ll}
\Proc\Big( &
\Inf	[+\cdot\oplus] 
	{\InfBox {\Deriv'_i }
		{ \Gamma'_1 \derinfOrch f_i : \rho_i \complyOF {\sigma}_i} 
	~~ ( \forall i \in I)
	}{ \Gamma_1\derinfOrch \rec x.\bigvee_{i\in I}\langle\Dual{a}_i,{a_i}\rangle.f_i : \bigExternal_{j \in I\cup J} a_j.{\rho}_j\complyOF\bigInternal_{i \in I} \Dual{a}_i.{\sigma}_i   }~~,\\
&\mbox{where~~}
\Gamma_1'= \Gamma_1,x{:}\bigExternal_{j \in I\cup J} a_j.{\rho}_j \complyOF \bigInternal_{i \in I} \Dual{a}_i.{\sigma}_i 
\\[4mm]
&
\Inf	[\oplus\cdot+\mbox{\,-\!}\subcontr] 
	{(h\in I\cap K) \quad \InfBox {\Deriv''_h }
		{ \Gamma_2'
     	\dersc \sigma_h \subcontrF \sigma'_h}
	}{ \Gamma_2\dersc \bigInternal_{i \in I} \Dual{a}_i.{\sigma}_i \subcontrF\bigExternal_{k \in K} \Dual{a}_k.\sigma_k' }~~,\\
& \mbox{where~~}
\Gamma_2'=\Gamma_2, \bigInternal_{i \in I} \Dual{a}_i.{\sigma}_i \subcontrF\bigExternal_{k \in K} \Dual{a}_k.\sigma_k' 
\\[2mm]
&
 \Gamma_3 ,\ \Deriv',\ \Deriv''~\Big) =
\end{array} 
\] 
{\hfill $\Inf[+\cdot+]
	{\InfBox {\Deriv_h }
		{ \Gamma_3, x{:}\rho \complyF \sigma' \derinfOrch g:\rho_h \complyOF \sigma_h'} 
	\quad (h\in I\cap K)
	}{ \Gamma_3\derinfOrch \rec x.\orchAct{\Dual{a}_h}{a_h}^+.g : \bigExternal_{j \in I\cup J} a_j.{\rho}_j \complyOF \bigExternal_{k \in K} \Dual{a}_k.\sigma_k'  }
$}\\
\WHERE\\[3mm]
 {\hfill $\InfBox {\Deriv_h }
		{ \Gamma_3, x{:}\rho \complyF \sigma' \derinfOrch g:\rho_h \complyOF \sigma_h'} $ 
\hspace{2mm}= \\[4mm]
$\Proc \Big( \InfBox{\Deriv'_i }{ \Gamma'_1\derinfOrch f_i : \rho_i \complyOF {\sigma}_i }
, ~
\InfBox {\Deriv''_h }
		{ \Gamma'_2 \dersc  \sigma_h \subcontrF {\sigma}'_h} 
, ~
\Gamma_3,  x{:}\rho  \complyOF \sigma',\ \Deriv',\ \Deriv''~\Big)$} \\[2mm]
\AND~~ $\rho=\bigExternal_{j \in I\cup J} a_j.{\rho}_j$ and $\sigma' = \bigExternal_{k \in K} a_k.\sigma_k' $.

\vspace{2mm}

\Clause {$(+\cdot+)$-$(+\cdot+\mbox{\,-\!}\subcontr)$}

\[ 
\begin{array}{ll}
\Proc\Big( &
\Inf	[+\cdot+] 
	{\InfBox {\Deriv_p }
		{ \Gamma'_1 \derinfOrch f' : \rho_p \complyOF {\sigma}_p} 
	\qquad (p  \in I\cap K)
	}{ \Gamma_1\der \rec x.\orchAct{\alpha_p}{\Dual{\alpha}_p}^+.f' : \bigExternal_{k\in K} \Dual{\alpha}_k.{\rho}_k \complyOF \bigExternal_{i \in I}\alpha_i.\sigma_i  }~~,\\[2mm]
&\mbox{\footnotesize where~~ $\Gamma_1'= \Gamma_1, x: \bigExternal_{k\in K} \Dual{\alpha}_k.{\rho}_k \complyOF \bigExternal_{i \in I}\alpha_i.\sigma_i $ }
\\ [4mm]
&
\Inf	[+\cdot+\mbox{\,-\!}\subcontr] 
	{\InfBox {\tilde\Deriv_i }
		{ \Gamma'_2 \dersc \sigma_i \subcontrF \sigma'_i }
	\qquad ( \forall i \in I)
	}{ \Gamma_2\dersc \bigExternal_{i \in I}\alpha_i.\sigma_i \subcontrF \bigExternal_{j \in I\cup J} \alpha_j.\sigma'_j }~~,\\[2mm]
&\mbox{\footnotesize where~~
$\Gamma_2'=\Gamma_2, \bigExternal_{i \in I}\alpha_i.\sigma_i \subcontrF \bigExternal_{j \in I\cup J} \alpha_j.\sigma'_j$}
\\[2mm]
&
 \Gamma_3 ,\ \Deriv',\ \Deriv''~\Big) =
\end{array} \] 
{\hfill $\Inf	[+\cdot+]
	{
	\InfBox {\Deriv'_p }
		{ \Gamma_3,x: \rho \complyF \sigma' \derinfOrch g : \rho_p \complyF \sigma'_p} 
	\qquad ( p \in   (I\cup J) \cap K)
	}{ \Gamma_3\der \rec x.\orchAct{\alpha_p}{\Dual{\alpha}_p}^+.g : \rho \complyF \sigma' }$}\\[2mm]

\WHERE\\[3mm]
{\hfill $\InfBox {\Deriv'_p }{ \Gamma_3, x:\rho \complyOF \sigma' \der g: \rho_p \complyF \sigma'_p}
= \\[4mm]
\Proc \big( \InfBox {\Deriv_p }{ \Gamma'_1\der f' : \rho_p \complyOF {\sigma}_p }
, ~
\InfBox {\tilde\Deriv_p}
		{ \Gamma'_2 \dersc \sigma_p \subcontrF \sigma'_p }
, ~
\Gamma_3,  x:\rho  \complyF \sigma' ,\ \Deriv',\ \Deriv''  \big)
$}\\[2mm]
\AND~~ $\rho=\bigExternal_{k \in K} \Dual{\alpha}_k.{\rho}_k $ and $\sigma' =\bigExternal_{j\in I\cup J} \alpha_j.\sigma'_j$,

\Clause {$(\oplus\cdot+)$-$(+\cdot+\mbox{\,-\!}\subcontr)$}

The construction follows the same pattern as before.

\Clause {$(+\cdot\oplus)$-$(\oplus\cdot\oplus\mbox{\,-\!}\subcontr)$}

The construction follows the same pattern as before.
\setlength {\unitlength} {1\point}
\end{defi}

\bigskip
\noindent
We proceed now by showing succesful termination of $\Proc$.

\begin{defi}
The set $\Inv$ of clauses is defined as follows:\\
\lmcscenterline{$\Inv =
\left\{
\begin{array}{c}
(\TcomplHyp)\cdot(\TcomplHyp\mbox{\,-\!}\subcontr),~ (\TcomplHyp)\cdot*,~ *\cdot(\TcomplHyp\mbox{\,-\!}\subcontr)
\end{array}
\right \}$ }
We call {\em invariant} the clauses belonging to the above set.
\end{defi}

\bigskip
\begin{lem}
\label{lem:noinvinv}
Let 
$\Deriv' :: \derinfOrch f: {\rho} \complyOF \sigma$ and
$\Deriv'' :: \dersc \sigma \subcontrF \sigma' $.\\
Then during the evaluation of the procedure $\PProc(\Deriv,\Deriv')$, no two consecutive recursive calls of $\Proc$ can both be due to clauses belonging to the set $\Inv$.
\end{lem}

\proof
By contradiction, let us assume to have two recursive calls as described above.
We proceed by cases according to the clause corresponding to the first recursive call:
\begin{description}
\item [$(\TcomplHyp)\cdot(\TcomplHyp\mbox{\,-\!}\subcontr)$]
In such a case the recursive call due to application of such a clause produce the call\\
\lmcscenterline{
$\Proc\big( ~
\InfBox {\Deriv_1 }{ \Gamma'_1 \derinfOrch f':{\rho} \complyOF \sigma},~
\InfBox {\Deriv_2 }{ \Gamma'_2 \dersc \sigma \subcontrF \sigma'},~ \Gamma_3, ~\Deriv',~\Deriv'' ~\big)$
}
where 
\begin{equation}
\label{eq:cond}
x:{\rho} \complyOF \sigma\not\in \Gamma'_1 \mbox{ and } \sigma \subcontrF \sigma'\not\in \Gamma'_2
\end{equation}
In order for such a call to be evaluated by means of a clause in $\Inv$, at least one among
$\InfBox {\Deriv_1 }{ \Gamma'_1 \derinfOrch f':{\rho} \complyOF \sigma}$ ~ and ~ 
$\InfBox {\Deriv_2 }{ \Gamma'_2 \dersc \sigma \subcontrF \sigma'}$
should have the form, respectively,
$$\Inf[\TcomplHyp]{}{ \Gamma''_1, x:{\rho} \complyOF \sigma \derinfOrch x: {\rho} \complyOF \sigma} ~\mbox{ and }~ \Inf[\TcomplHyp]{}{ \Gamma''_2,\sigma \subcontrF \sigma' \dersc \sigma \subcontrF \sigma'}$$
 where $\Gamma'_1=\Gamma''_1, x:{\rho} \complyOF \sigma \derinfOrch x: {\rho} \complyOF \sigma$
and  $\Gamma'_2 =\Gamma''_2,\sigma \subcontrF \sigma'$
which is impossible by (\ref{eq:cond}).

\item[$(\TcomplHyp)\cdot*$]
If a call corresponding to such a clause were followed by a call due to $(\TcomplHyp)\cdot(\TcomplHyp\mbox{\,-\!}\subcontr)$ or $(\TcomplHyp)\cdot*$, 
the very same argument of the previous case would apply.\\
If it were followed by $*\cdot(\TcomplHyp\mbox{\,-\!}\subcontr)$, we notice that such a call produces a call of the procedure
with the second argument unchanged. So the last rule of the second derivation to be
actually $(\TcomplHyp\mbox{\,-\!}\subcontr)$. That is the call $(\TcomplHyp)\cdot*$ is actually a $(\TcomplHyp)\cdot(\TcomplHyp\mbox{\,-\!}\subcontr)$ call, and the argument of the first case applies.
\item[$*\cdot(\TcomplHyp\mbox{\,-\!}\subcontr)$] Similar to the previous case.
\end{description}
\qed

\begin{lem}\hfill
\label{lem:dersatcond}

Let us assume that in a derivation $\Deriv ::\  \emptyset\derinfOrch \rho \complyOF \sigma $ (resp. $\Deriv ::\   \emptyset\dersc \sigma \subcontrF \sigma' $)
is present an axiom of the form 
$\Inf[\TcomplHyp]{}{ \Gamma', x:\rho' \complyF \sigma' \derinfOrch x:\rho' \complyOF \sigma'}$
 (resp. $\Inf[\TcomplHyp\mbox{\,-\!}\subcontr]{}{ \Gamma', \sigma'' \subcontrF \sigma''' \dersc \sigma'' \subcontrF \sigma'''}$). \\
Then there exists a subderivation $\Deriv'$ of $\Deriv$ such that
$\Deriv' ::\ \Gamma''\derinfOrch f' : \rho' \complyOF \sigma' $
(resp. $\Deriv' ::\ \Gamma''\dersc \sigma'' \subcontrF \sigma''' $)
with $x:\rho' \complyOF \sigma' \not\in\Gamma''$ (resp. $\sigma'' \subcontrF \sigma''' \not\in\Gamma''$).
\end{lem}
\proof
Easy. The conclusion of $\Deriv$ has an empty environment, hence the judgment
$x:\rho' \complyF \sigma'$ (resp. $\sigma'' \subcontrF \sigma'''$) will have to be discharged, sooner or later, from the environment $\Gamma'$. 
\qed

 \begin{lem}\hfill\\
\label{lem:terminv}
Let 
$\Deriv' ::\ \derinfOrch f: {\rho} \complyOF \sigma$ and
$\Deriv'' ::\ \dersc \sigma \subcontrF \sigma' $.
\begin{enumerate}
\item
\label{lem:terminv-i}
If
$\Proc (~\widetilde{\Deriv}_1 ,~\widetilde{\Deriv}_2
,~\Gamma_3, \ \Deriv',\Deriv''~)$ is any call that can occur during the computation
of $\PProc(  \Deriv', \Deriv'' ) $,
then $\widetilde{\Deriv}_1$ and $\widetilde{\Deriv}_2$ are subderivation of $\Deriv'$ and $\Deriv''$, respectively.
\item
\label{lem:terminv-ii}
The computation of $\PProc(  \Deriv', \Deriv'' ) $ cannot fail because of a clause in $\Inv$.
\end{enumerate}
\end{lem}
\proof
(\ref{lem:terminv-i})~
Easy, by inspection of the clauses of the algorithm $\Proc$.\\
(\ref{lem:terminv-ii})~
By contradiction, let us assume the computation of $\PProc(  \Deriv', \Deriv'' ) $ does fail because of  the clause $*\cdot(\TcomplHyp\mbox{\,-\!}\subcontr)$ (all the other clauses in  $\Inv$ can be 
treated similarly). This means that we cannot go on with the computation of 
\[
\Proc\Big( 
\InfBox {\Deriv_1 }{ \Gamma_1 \derinfOrch f:{\rho} \complyOF \sigma},~
\Inf[\TcomplHyp\mbox{\,-\!}\subcontr]{}{ \Gamma_2, \sigma \subcontrF \sigma' \dersc \sigma \subcontrF \sigma'} ,~
\Gamma_3,\ \Deriv',\ \Deriv''~\big) =
\] 
because there exists no subderivation of $\Deriv''$ such that 
\begin{enumerate}[(a)]
\item
\label{c1}
 it has the form $\InfBox {\Deriv_2 }{ \Gamma'_2 \der  \sigma \subcontrF \sigma'}$ and
\item
\label{c2}
 $\sigma \subcontrF \sigma'\not\in\Gamma'_2 $.
\end{enumerate}
This is impossible because, by item (\ref{lem:terminv-i}),  
$\Inf[\TcomplHyp\mbox{\,-\!}\subcontr]{}{ \Gamma_2, \sigma \subcontrF \sigma' \dersc \sigma \subcontrF \sigma'}$ is an axiom of $\Deriv''$ and hence, by Lemma \ref{lem:dersatcond},
there exists a subderivation satisfying both the conditions (\ref{c1}) and (\ref{c2}) above.
\qed

\begin{prop} \label{prop:Proctermination}
Let 
$\Deriv' ::\ \derinfOrch f: {\rho} \complyOF \sigma$ and
$\Deriv'' ::\ \dersc \sigma \subcontrF \sigma' $.
the computation of $\PProc(  \Deriv', \Deriv'' ) $ always succesfully terminates.
\end{prop}

\proof
We recall that $\PProc(\Deriv',~\Deriv'') = \Proc(\Deriv',\Deriv'',\emptyset,\Deriv',\Deriv'')$.
The algorithm $\Proc$ is a derivation reconstruction algorithm that tries to build a derivation for the judgment $\derinfOrch \Funct_{\Der''}(f) : \rho \complyOF \sigma'$ in a bottom-up way,
driven by the two derivations 
$\Deriv_1$ and $\Deriv_2$, and by the third argument (the environment). \\
By inspection of the clauses of $\Proc$ we can check that it is impossible the computation of $\PProc(\Deriv',~\Deriv'')$ to  be not terminating because of the existence of an infinite sequence
of recursive calls of $\Proc$, or to fail because of a clause in $\Inv$ that cannot be applied because
the conditions for its application are not satisfied. The latter case is impossible by
Lemma \ref{lem:terminv}(\ref{lem:terminv-ii}). So we proceed now by showing that also the
former case cannot occur.
The proof of this fact is basically rooted in the same ideas as the termination proof for the reconstruction algorithm {\Prove} (see proof of Proposition \ref{prop:provecorr}) used in the completeness proof for the formal system $\der$. Hence in the following we avoid detailing notions similar to those used there.\\

Given a call $\Proc\Big( ~
{\InfBox {\Deriv_1 }{ \Gamma_1\derinfOrch \tilde f: {\tilde\rho} \complyOF \tilde\sigma}} ,~
{\InfBox {\Deriv_2 }{ \Gamma_2\dersc \tilde\sigma \subcontrF \tilde\sigma' }} ,~
\Gamma_3,~ \Deriv', \Deriv''~\Big)$ of the procedure $\Proc$, we refer to $\tilde f$ as {\em the orchestrator}, to $\tilde\rho$ as {\em the client}, to $\tilde\sigma'$ as {\em the server} and to $\sigma$ as 
{\em the intermediate server} of the call.

We first observe that for any recursive call 
corresponding to an application of a clause not in $\Inv$, 
 the client $\rho'$ and the server $\sigma''$ in the call 
are subterms (i.e. they correspond to subtrees) of, respectively, the orchestrator, the client and the server in the call of $\PProc(  \Deriv', \Deriv'' )$, i.e. $\rho$ and $\sigma$, furthermore $x:\rho' \complyOF \sigma''$ is added to the environment parameter (the third parameter) for some variable $x$.  We make a further observations:
 the number of variables that can be associated to an element $\rho' \complyOF\sigma''$ inside the third argument (the environment) of any recursive call during the computation of
$\PProc(  \Deriv', \Deriv'' )$ is necessarily finite (no alpha-conversion is assumed to be performed
on terms during the procedure).
By the above observations and the regularity of the trees generated by our contracts, it follows that, for any sequence of recursive calls out of the initial call of $\PProc(  \Deriv', \Deriv'' )$,
either the clause $(\TcomplAx)$-$*$ will eventually apply, or we get to a call of the form 
$$\Proc \Big(~
{\InfBox {\Deriv_1 }{\Gamma_1\derinfOrch  f:{\rho} \complyOF \sigma}}~
,~
{\InfBox {\Deriv_2 }{ \Gamma_2\dersc \sigma \subcontrF \sigma' }}
,~
\Gamma_3, x{:} {\rho}  \complyOF \sigma',\ \Deriv',\Deriv''~
\Big)$$
By a similar argument as before, by going on with the recursive call, we shall get to
 a clause $(\TcomplAx)$-$*$, or to one of the two clauses:
$(\TcomplHyp)$-$*$ or 
$*$-$(\TcomplHyp\mbox{\,-\!}\subcontr)$ {\bf Init}.
Such an argument holds unless the sequence of calls, from a give point on, is definitely made of calls due to clauses in $\Inv$, a possibility which however can never happen by  Lemma \ref{lem:noinvinv}.\\
So our derivation reconstruction does succesfully terminate.
\qed

\begin{prop} \label{prop:Proccorrectness}
Let 
$\Deriv' ::\ \derinfOrch f: {\rho} \complyOF \sigma$ and
$\Deriv'' ::\ \dersc \sigma \subcontrF \sigma' $, and let $\Funct_{\Der''}$ be defined as in Definition \ref{def:orchfunctor}. Then
$$\PProc(  \Deriv', \Deriv'' )\ =\ \Deriv ::\ \derinfOrch \Funct_{\Der''}(f') : \rho \complyOF \sigma' $$ for some
 derivation $\Deriv$, where $f'$ is obtained out of $f$ by a finite number of $\rec$ expansions (possibly zero).
\end{prop}
\proof
[Sketch] By working on the derivation $\Deriv$ and by inspection of the clauses of
$\Proc$.
\qed

\fi

\end{document}